\documentclass[sigconf,balance=false]{acmart}
\usepackage{popets}

\setcopyright{popets}
\copyrightyear{YYYY}

\acmYear{YYYY}
\acmVolume{YYYY}
\acmNumber{X}
\acmDOI{XXXXXXX.XXXXXXX}
\acmISBN{}
\acmConference{Proceedings on Privacy Enhancing Technologies}
\usepackage{multirow}
\usepackage{xcolor}
\usepackage{csquotes}

\usepackage{amsmath}

\usepackage{amssymb}
\usepackage{array}
\usepackage{graphicx}
\usepackage{svg}

\usepackage[ruled,vlined,linesnumbered]{algorithm2e}

\usepackage{subcaption}
\usepackage{amsthm}
\usepackage[symbol]{footmisc}

\usepackage{tikz}
\usepackage{enumitem}
\newcommand{\RNum}[1]{\uppercase\expandafter{\romannumeral #1\relax}}

\newcommand{\MSE}{\operatorname{MSE}}

\newtheoremstyle{customtheorem}%
  {3pt}%
  {3pt}%
  {\itshape}%
  {}%
  {\bfseries}%
  {}%
  { }%
  {\thmname{#1}\thmnumber{ #2.}\thmnote{ (#3)}}%

\theoremstyle{customtheorem}
\newtheorem{theorem}{Theorem}
\newtheorem{proposition}{\textbf{Proposition}}

\begin{document}

\title{Frequency Estimation of Correlated Multi-attribute Data under Local Differential Privacy}

\author{Shafizur Rahman Seeam}
\orcid{0000-0003-3350-0047}
\affiliation{%
  \institution{Rochester Institute of Technology }
  \city{Rochester}
  \state{NY}
  \country{USA}}
\email{ss6365@rit.edu}

\author{Ye Zheng}
\orcid{0000-0003-0623-9613}
\affiliation{%
  \institution{Rochester Institute of Technology }
  \city{Rochester}
  \state{NY}
  \country{USA}}
\email{ye.zheng@mail.rit.edu}

\author{Yidan Hu}
\orcid{0000-0002-9443-8411}
\affiliation{%
  \institution{Rochester Institute of Technology }
  \city{Rochester}
  \state{NY}
  \country{USA}}
\email{yidan.hu@rit.edu}

\begin{abstract} 
Large-scale data collection, from national censuses to IoT-enabled smart homes, routinely gathers dozens of attributes per individual. These multi-attribute datasets are crucial for analytics but pose significant privacy risks. Local Differential Privacy (LDP) is a powerful tool for protecting user privacy by allowing users to locally perturb their records before releasing them to an untrusted data aggregator. However, existing LDP mechanisms either split the privacy budget across all attributes or treat each attribute independently, thereby ignoring natural inter-attribute correlations. This leads to excessive noise and, consequently, significant utility loss, particularly in high-dimensional datasets.

We introduce a two-phase LDP framework that overcomes these limitations by privately learning and exploiting inter-attribute dependencies. In Phase~I, a small subset of users applies a standard per-attribute LDP mechanism, enabling the aggregator to derive dependency information from the privatized data. In Phase~II, each remaining user perturbs a single randomly chosen attribute with the full privacy budget, while the unreported attributes are reconstructed using Phase~I statistics, incurring no additional privacy cost. As a concrete instantiation, we develop Correlated Randomized Response (Corr-RR), which employs correlation-aware probabilistic mappings to substantially improve estimation accuracy. We prove that Corr-RR satisfies $\epsilon$-LDP, and demonstrate through extensive experiments on synthetic and real-world datasets that it consistently outperforms state-of-the-art baselines, with the largest gains in high-dimensional and strongly correlated datasets.
\end{abstract}

\maketitle

\section{Introduction}

Large-scale multi-attribute datasets, where each record contains multiple features such as age, education, and income, are crucial for evidence-based policymaking, healthcare analytics, and data-driven decision-making. A core primitive in analyzing such data is \emph{frequency estimation}, which quantifies how often attribute values occur and serves as a building block for tasks such as heavy-hitter detection~\cite{bun2019heavy}, key–value aggregation~\cite{ye2019privkv}, and time-sensitive analytics~\cite{wang2020towards}. However, these tasks typically rely on raw user data that may contain sensitive personal information. In multi-attribute settings, combinations of attributes (e.g., age + ZIP code + gender) can uniquely identify individuals, heightening privacy risks. Such risks are particularly acute in regulated domains such as healthcare, finance, and location-based services~\cite{seeam2025privar}, where frameworks such as GDPR~\cite{GDPR} and HIPAA~\cite{HIPAA} mandate strict privacy requirements.

Local Differential Privacy (LDP) offers strong user-level protection without relying on a trusted curator~\cite{kasiviswanathan2011can, duchi2013local}, as each user perturbs their data locally before transmission. The parameter $\epsilon$ controls the utility-privacy trade-off, where smaller $\epsilon$ enforces stronger privacy, but adds more noise, reducing utility. LDP has been adopted in practice by Google~\cite{erlingsson2014rappor}, Apple~\cite{team2017learning}, and Microsoft~\cite{ding2017collecting}, underscoring its real-world viability.

Frequency estimation under LDP is straightforward for single attributes, but extending it to multi-attribute data remains challenging. Existing solutions that require complete reports of all $d$ attributes\footnote{We exclude Random Sampling, which reports only one attribute per user with the full privacy budget $\epsilon$.} fall into two main categories. The first, Split Budget (SPL), evenly divides the privacy budget $\epsilon$ across all attributes, assigning only $\epsilon/d$ to each attribute. As $d$ grows, the per-attribute budget decreases rapidly, leading to high noise and poor accuracy~\cite{wang2019collecting, arcolezi2024improving}. The second allocates the full privacy budget to a single attribute and imputes the remaining $d-1$ attributes with synthetic values. One such method, Random Sampling plus Fake Data (RS+FD), perturbs one attribute with $\epsilon$ and fills the others uniformly at random~\cite{arcolezi2021random, arcolezi2024improving}, which introduces systematic bias because the imputed values do not reflect realistic distributions. A refinement, Random Sampling plus Realistic Fake Data (RS+RFD), samples unreported attributes from prior distributions~\cite{arcolezi2022risks}. However, its utility is highly sensitive to the quality of priors, which may be biased, outdated, or unavailable due to privacy constraints. In both cases, attributes are treated as independent, overlooking the natural correlations present in real-world data, an omission that leads to unnecessary noise and loss of utility. Consequently, there remains a pressing need for advanced LDP mechanisms that strike a better balance between accuracy and privacy in multi-attribute frequency estimation.

In this work, we exploit the \emph{inter-attribute correlations} naturally present in real-world datasets for enhanced utility-privacy trade-off, a direction that remains underexplored in the context of LDP. Many attributes are strongly dependent; for example, educational attainment, employment status, and income level are closely linked. Explicitly leveraging such correlations can reduce noise and improve accuracy under LDP without weakening privacy guarantees. Intuitively, consider two attributes $X_1$ and $X_2$ that are perfectly correlated ($X_1 = X_2$). 
Instead of perturbing both independently, a user could perturb one attribute via an LDP mechanism with the full privacy budget: $Y_1 \leftarrow \mathcal{M}_{\epsilon}(X_1)$, and report the same perturbed value for $X_2$, i.e., $Y_2=Y_1$. By doing so, users can generate reports for unselected attributes with improved utility without incurring any additional privacy cost. In practice, however, two challenges arise: (i) real-world attribute correlations are rarely perfect, varying in both strength and form, and (ii) inter-attribute correlations are often unknown, and strict privacy regulations (e.g., GDPR, HIPAA) prohibit direct access to raw data for correlation estimation. These challenges motivate our central question:

\begin{quote}
\emph{How can we leverage inherent yet unknown \underline{inter-attribute} \underline{correlations} to improve the utility of multi-attribute frequency estimation under LDP without compromising privacy guarantees?}
\end{quote}

We address this question by proposing a \textbf{two-phase LDP framework} consisting of: (i) \textit{Phase I: Dependency Learning}, which estimates inter-attribute correlations directly from privatized data under LDP, without requiring access to raw user data; and (ii) \textit{Phase II: Dependency-Aware Reporting}, which leverages the learned correlations to improve the accuracy of multi-attribute frequency estimation while maintaining $\epsilon$-LDP. Specifically, in Phase~I, a small subset of users applies SPL, perturbing all $d$ attributes independently with budget $\epsilon/d$. The data collector then aggregates these noisy reports to privately infer inter-attribute correlations, leveraging the statistical relationship between correlations in the original data and those in the perturbed values—without requiring access to raw data.
In Phase~II, each remaining user randomly selects one pivot attribute and perturbs it using an LDP mechanism with the full budget $\epsilon$. The remaining $d-1$ attributes are then reconstructed \emph{indirectly}, using the learned inter-attribute correlations from Phase~I for an enhanced utility-privacy trade-off. Because reconstruction is purely post-processing on already-perturbed values, the overall mechanism continues to satisfy $\epsilon$-LDP.

As a concrete instantiation of our framework, we propose \textbf{Correlated Randomized Response (Corr-RR)}. Corr-RR randomly selects one attribute and perturbs it with the full privacy budget. For each unselected attribute, instead of deterministically reusing the pivot's report (i.e., $Y_2 = Y_1$ under perfect correlation), it introduces a parameter $p_y \in [0,1]$ that specifies the probability of reusing the pivot's perturbed value. Intuitively, a larger $p_y$ is favored when attributes are strongly correlated, while a smaller $p_y$ corresponds to weaker dependence. To determine these parameters, Corr-RR leverages marginal estimates obtained in Phase~I and chooses $p_y$ values that minimize a closed-form approximation of the mean squared error (MSE) of the frequency estimators. Below, we summarize our contributions:

\begin{itemize}
\item \textbf{Two-phase LDP framework.} We introduce the first framework that separates \emph{dependency learning} from \emph{dependency-aware reporting}, deriving correlation-aware parameters from privatized data to improve multi-attribute frequency estimation.
\item \textbf{Concrete mechanism: Corr-RR.} We propose \emph{Corr-RR}, a novel instantiation that perturbs one randomly chosen attribute with the full privacy budget and synthesizes the remaining attributes using correlation-guided probabilistic mappings.
\item \textbf{Rigorous privacy guarantee.} We formally prove that Corr-RR satisfies $\epsilon$-LDP.
\item \textbf{Extensive evaluation.} Experiments on both synthetic and real-world datasets demonstrate that Corr-RR consistently outperforms state-of-the-art baselines, with the largest accuracy gains in high-dimensional and strongly correlated settings.
\end{itemize}

\textbf{Roadmap.} The remainder of this paper is organized as follows: Section~\ref{background} introduces the relevant background and the problem statement. Section~\ref{sec:solutiona} details our proposed solution. Section~\ref{sec:evaluation} presents experimental evaluations. Section~\ref{sec:related} discusses the existing literature, and we conclude the work in Section~\ref{sec:conclusion}.

\section{Preliminaries}\label{background}

\subsection{Local Differential Privacy (LDP)} 
Local Differential Privacy (LDP) requires that each user perturb their data locally before transmission, thereby ensuring robust privacy protection in distributed data collection settings~\cite{kasiviswanathan2011can}. 

\begin{definition}[$\epsilon$-LDP]
A randomized mechanism $\mathcal{M}:\mathcal{X}\rightarrow\mathcal{Y}$ satisfies $\epsilon$-LDP, where $\epsilon\geq 0$, if for any two inputs $x,x'\in\mathcal{X}$ and any output $y\in\mathcal{Y}$, the following holds:
\[
\Pr[\mathcal{M}(x)=y] \;\leq\; e^{\epsilon}\cdot \Pr[\mathcal{M}(x')=y].
\]
\end{definition}

The parameter $\epsilon$, called the \emph{privacy budget}, controls the trade-off between privacy and utility: smaller $\epsilon$ provides stronger privacy but introduces more noise, while larger $\epsilon$ yields better accuracy but weaker privacy. Prior work on multi-attribute data typically explores $\epsilon\in[0.1,10]$~\cite{arcolezi2021random, wang2021local}.

LDP inherits several key properties from Differential Privacy (DP), including composability~\cite{kairouz2015composition, mcsherry2009privacy} and post-processing immunity~\cite{TCS-042}. We restate them here:

\begin{theorem}[Sequential Composition]~\cite{zhang2023trajectory}\label{compositiontheorem}
If mechanisms $\{\mathcal{M}_i\}_{i=1}^n$ each satisfy $\epsilon_i$-LDP, then their sequential application satisfies $\epsilon$-LDP with $\epsilon=\sum_{i=1}^n \epsilon_i$.
\end{theorem}

\begin{theorem}[Parallel Composition]~\cite{mcsherry2009privacy}\label{Paralleltheorem}
If mechanisms $\{\mathcal{M}_i:\mathcal{X}_i\rightarrow\mathcal{Y}_i\}$ each satisfy $\epsilon_i$-LDP on disjoint input subsets $\{\mathcal{X}_i\}$, then the combined mechanism $\mathcal{M}=(\mathcal{M}_1,\ldots,\mathcal{M}_n)$ satisfies $\max_i \epsilon_i$-LDP.
\end{theorem}

\begin{theorem}[Post-Processing]~\cite{zhang2023trajectory, li2024privacy}\label{Post-processing}
If $\mathcal{M}:\mathcal{X}\rightarrow\mathcal{Y}$ satisfies $\epsilon$-LDP and $\mathcal{F}:\mathcal{Y}\rightarrow\mathcal{Y}'$ is any randomized mapping, then the composed mechanism $\mathcal{F}\circ\mathcal{M}:\mathcal{X}\rightarrow\mathcal{Y}'$ also satisfies $\epsilon$-LDP.
\end{theorem}

\subsection{Generalized Randomized Response (GRR)}\label{grr}
Generalized Randomized Response (GRR) extends the classic Randomized Response (RR) technique~\cite{warner1965randomized} to categorical domains of size $k=|\mathcal{D}|\geq 2$, while satisfying $\epsilon$-LDP~\cite{erlingsson2014rappor}.  
In the single-attribute setting, each user $u_i$ holds a private value $x_i \in \mathcal{D}$.  
Given $x_i$, the reported value is drawn as follows. For any $y\in\mathcal{D}$,
\[
\Pr[\mathcal{M}_{\mathrm{GRR}(\epsilon,k)}(x_i)=y]=
\begin{cases}
p=\dfrac{e^{\epsilon}}{e^{\epsilon}+k-1}, & \text{if } y=x_i,\\[4pt]
q=\dfrac{1}{e^{\epsilon}+k-1}, & \text{if } y\neq x_i.
\end{cases}
\]
The privacy guarantee follows from $p/q = e^{\epsilon}$.  

Let $n$ be the number of users and $c_v$ the count of reports equal to $v$, where $v \in \mathcal{D}$. The unbiased estimator of the true frequency
\(
f(v) = \frac{1}{n}\sum_{i=1}^n \mathbb{I}(x_i = v)
\)
is:
\[
\hat{f}(v) = \frac{c_v/n - q}{p-q}.
\]

The estimation variance grows linearly with the domain size $k$, leading to degraded accuracy in high-cardinality domains~\cite{wang2017locally}:
\[
\mathrm{Var}[\hat{f}(v)] = \frac{e^{\epsilon}+k-2}{n(e^{\epsilon}-1)^2}.
\]

\subsection{Problem Statement}
We now extend to the multi-attribute setting.  
Each user $u_i$ holds a private record
\(
\mathbf{x}_i = (x_{i,1},x_{i,2},\ldots,x_{i,d}),
\)
where each $x_{i,j}$ denotes the value of attribute $X_j$ for user $i$, drawn from a \emph{common domain} $\mathcal{D}$ of size $k=|\mathcal{D}|$.  
To protect privacy, each user applies a local randomization mechanism $\mathcal{M}$ to their entire record, producing
\(
\mathbf{y}_i= \mathcal{M}(\mathbf{x}_i)=  (y_{i,1},y_{i,2},\ldots,y_{i,d}), 
\; y_{i,j}\in\mathcal{D}.
\)
Users must report their full perturbed record rather than selectively reporting a subset of attributes, since selective disclosure significantly increases re-identification risk~\cite{arcolezi2021random}.  
Given the collection of reports $\{\mathbf{y}_1,\ldots,\mathbf{y}_n\}$, the data collector aims to estimate the marginal distribution of each attribute.  
For attribute $X_j$ and value $v\in\mathcal{D}$, the true marginal frequency is:
\[
f_j(v) = \frac{1}{n}\sum_{i=1}^n \mathbb{I}(x_{i,j}=v),
\]
where $\mathbb{I}(\cdot)$ is the indicator function.  
The goal is to design $\mathcal{M}$ such that the estimated marginals $\hat{f}_j(v)$ are as accurate as possible while ensuring $\epsilon$-LDP.

\subsection{Existing Solutions and Their Limitations}\label{sec:existingmethod}

We now review existing LDP mechanisms for multi-attribute frequency estimation and outline their limitations.  

\subsubsection{Split Budget (SPL)} 
In SPL, each user perturbs all $d$ attributes independently using GRR with per-attribute budget $\epsilon/d$ and reports the full vector $\mathbf{y}_i$. The aggregator then debiases each attribute separately to obtain marginal estimates. Because the privacy budget is divided across all attributes, each is perturbed with very low $\epsilon/d$, introducing substantial noise. The estimation error grows as $O(d\sqrt{\log d}/(\epsilon\sqrt{n}))$~\cite{wang2019collecting, arcolezi2024improving}, which quickly becomes prohibitive in high-dimensional settings.  

\subsubsection{Random Sampling with Fake Data (RS+FD)~\cite{arcolezi2021random}}  
In RS+FD, each user selects one attribute uniformly at random, perturbs it with the full budget $\epsilon$, and fills the remaining $d-1$ attributes with uniformly generated fake values. This reduces noise for the sampled attribute, but the majority of reported attributes carry no real information. As $d$ increases, these fake values dominate the dataset, leading to significant accuracy degradation.  

\subsubsection{Random Sampling with Realistic Fake Data (RS+RFD)~\cite{arcolezi2022risks}}  
This improves upon RS+FD by sampling non-selected attributes from external prior distributions (e.g., Census statistics) rather than uniformly. While this yields more realistic reports and partially alleviates RS+FD's bias, its accuracy is highly sensitive to the quality of priors. In practice, priors may be outdated, biased, or unavailable for sensitive attributes due to privacy regulations such as GDPR~\cite{GDPR} or HIPAA~\cite{HIPAA}, limiting the robustness of this approach.

\noindent\textbf{Discussion.}  
Although SPL, RS+FD, and RS+RFD adopt different strategies, they share a common limitation: all three treat attributes as independent and fail to exploit the correlations naturally present in real-world data. As a result, they either suffer from excessive noise due to budget splitting (SPL) or introduce bias and reliance on external priors (RS+FD/RS+RFD). This gap motivates the need for correlation-aware mechanisms that can improve accuracy without weakening privacy guarantees.

\section{Two-Phase Privacy Framework}\label{sec:solutiona}

We now present a general two-phase privacy framework for multi-attribute data collection, followed by its concrete instantiation, Correlated Randomized Response (Corr-RR).

\begin{figure*}[ht!]
  \centering
\includegraphics[width=0.98\textwidth]{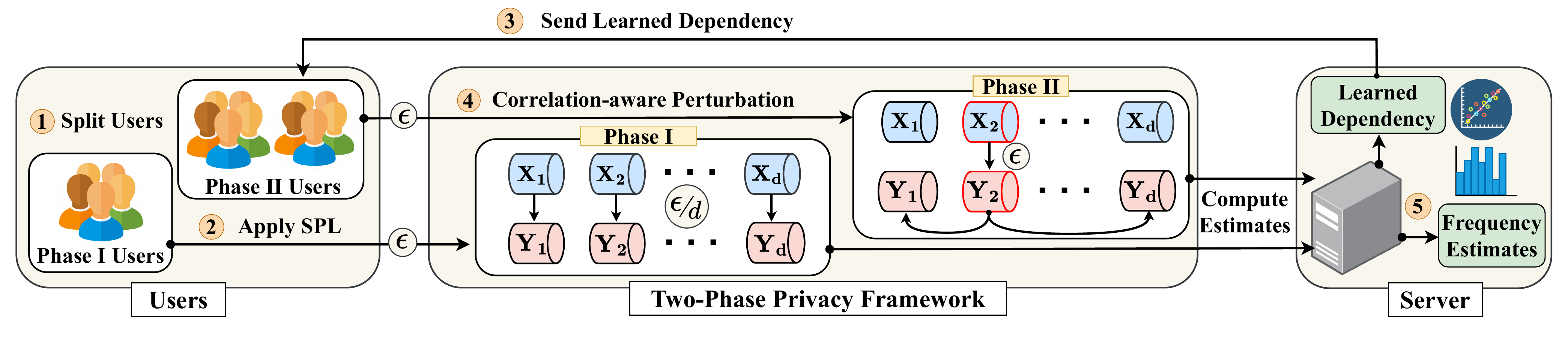}
   \caption{Two-phase privacy framework. Phase~I users apply SPL to enable private estimation of inter-attribute dependencies. Phase~II users perturb one attribute with the full privacy budget and reconstruct the rest using these dependencies.}
   \Description{Block diagram showing the two-phase LDP framework. Phase I users apply split budget perturbation to learn inter-attribute dependencies, which are sent to the server. Phase II users perturb one pivot attribute with full privacy budget and reconstruct remaining attributes using the learned dependencies.}

  \label{fig:framework}
\end{figure*}

\subsection{Overview}
The proposed framework is built on two key observations. First, inspired by RS+FD and RS+RFD, we observed that perturbing a single \emph{pivot} attribute with the full privacy budget $\epsilon$ and generating the remaining attributes through some randomization by indirect perturbation of this privatized value can significantly improve the utility-privacy trade-off compared to conventional budget-splitting schemes like SPL. The key idea can be formulated as
1) assigning the entire budget to a randomly selected attribute,
\(
Y_1 \sim \mathcal{M}_{\epsilon}(X_1),
\)
and 2) generating the remaining attribute reports as
\(
Y_2 \sim \mathcal{G}(Y_1;\theta),
\)
where $\mathcal{G}$ is a randomized mapping defined only on the already privatized pivot. By the post-processing property, the overall mechanism satisfies $\epsilon$-LDP. However, designing an effective mapping $\mathcal{G}$ is non-trivial. Naive approaches, as in RS+FD and RS+RFD, often introduce excessive noise, degrading utility.

We also observe that multi-attribute data often exhibit strong inter-attribute dependencies. For example, educational attainment, employment status, and income level are tightly correlated. These dependencies can be exploited to guide the design of $\mathcal{G}$, enabling more accurate reporting through post-processing of the privatized pivot. Consider two attributes with perfect correlation (e.g., $X_1=X_2$). Once $X_1$ is privatized to $Y_1$, we can set $Y_2=Y_1$ without incurring additional privacy cost while enhancing data utility. For multi-attribute data with partial correlation, there are several ways to incorporate correlation into the perturbation process to generate higher-utility reports. For instance, if $X_1$ and $X_2$ are strongly but not perfectly correlated, and $X_1$ is selected as the pivot that generates $Y_1$ through direct perturbation, then the indirect report $Y_2$ could be equal to $Y_1$ with high probability and differ otherwise. A key challenge is that true correlations cannot be computed from raw user data, since strict privacy regulations (e.g., GDPR, HIPAA) prohibit direct access to sensitive attributes. This motivates our \textbf{two-phase privacy framework}, which first learns coarse dependency parameters privately and then leverages them for dependency-aware reporting. The two-phase framework (Figure~\ref{fig:framework}) operates as follows:

\begin{itemize}
\item \textbf{Phase I: Dependency Learning.} A small subset of $n_1$ users perturbs all $d$ attributes using SPL with per-attribute budget $\epsilon/d$. From these privatized reports the server can recover unbiased \emph{marginal frequencies} and, if needed, approximate \emph{joint statistics} across attributes. These privatized statistics retain enough signal to capture coarse inter-attribute dependencies. The server then encodes these dependencies into \emph{learned parameters} (e.g., probabilities for reuse, conditional distributions, or other mapping rules) that guide how unreported attributes should be reconstructed in Phase~II.

\item \textbf{Phase II: Dependency-Aware Reporting.} Each of the remaining users randomly selects one pivot attribute and perturbs it with the full privacy budget $\epsilon$. The remaining $d-1$ attributes are generated through a randomized transformation of the pivot's privatized value, guided by the dependency parameters learned in Phase~I. Since these reconstructions depend only on already privatized outputs, they constitute pure post-processing and do not increase the privacy cost, ensuring that the overall mechanism remains $\epsilon$-LDP.

\end{itemize}

This design provides a general template for correlation-aware multi-attribute frequency estimation under LDP. In the following subsection, we instantiate it with \emph{Correlated Randomized Response}.

\begin{table}[t]
\centering
\caption{Notation for Corr-RR Mechanism}
\label{tab:notation}
\begin{tabular}{lp{0.65\linewidth}}
\toprule
\textbf{Symbol} & \textbf{Meaning} \\
\midrule
$n$         & Total number of users \\
$n_1$       & Number of Phase~I users \\
$n_2$       & Number of Phase~II users ($n_2=n-n_1$) \\
$d$         & Number of attributes per user \\
$\mathbf{x}_i$ & Raw attribute vector of user $i$ \\
$\mathbf{y}_i$ & Perturbed attribute vector of user $i$ \\
$X_j$       & R.V. for the $j$-th raw attribute \\
$Y_j$       & R.V. for the $j$-th perturbed attribute\\
$x_{i,j}$   & Raw value of attribute $j$ for user $i$ \\
$y_{i,j}$   & Perturbed value of attribute $j$ for user $i$ \\
$k=|\mathcal{D}|$ & Size of the common attribute domain \\
$\epsilon$  & Privacy budget \\
$p_1,q_1$   & GRR probabilities in Phase~I \\
$p_{2},q_{2}$ & GRR probabilities in Phase~II \\
$\hat f_j^{I}(v)$ & Frequency est.\ for $X_j=v$ from Phase~I \\
$\hat f_j^{II}(v)$ & Frequency est.\ for $X_j=v$ from Phase~II \\
$\hat f_j(v)$ & Final combined estimate for $X_j=v$ \\
$p_y$       & Probability parameter for $(X_j,X_s)$ \\
\bottomrule
\end{tabular}
\end{table}

\subsection{Correlated Randomized Response (Corr-RR)}

Under the two-phase framework, we now present a concrete instantiation, \textbf{Corr-RR}, for privacy-preserving multi-attribute data collection. Corr-RR introduces probability parameters $p_{y} \in [0,1]$ derived from privatized marginals via an optimization procedure, to guide the probabilistic mapping used in Phase~II for reconstructing non-pivot attributes based on the perturbed pivot. For readability, we denote these parameters generically as $p_y$.  Consider two attributes $X_1$ and $X_2$ that are perfectly positively correlated. In this case, producing the same output for both attributes (i.e., $Y_1=Y_2$) corresponds to copying the perturbed value $Y_1$ as the report for $Y_2$ with probability $p_y=1$, and reporting a different value with probability $1-p_y$. More generally,
\[
\Pr(Y_2 = v \mid Y_1 = v) = p_y,
\qquad
\Pr(Y_2 \neq v \mid Y_1 = v) = 1 - p_y.
\]
Thus $p_y$ controls the likelihood that the non-pivot attribute aligns with the pivot's privatized value, with its value chosen to minimize the expected mean squared error (MSE). 
In Corr-RR, the MSE is expressed as a function of marginal distributions. Those distributions could serve as proxies for inter-attribute correlations: if two attributes are strongly positively correlated, their marginals tend to be similar, leading the optimizer to favor larger $p_y$; conversely, weaker similarity in marginals yields smaller $p_y$. Although marginals do not capture exact correlations, they preserve enough statistical signal for effective parameter estimation under LDP. Accordingly, Corr-RR first estimates marginal distributions for each attribute in Phase~I and uses them to compute $p_y$. In Phase~II, Corr-RR randomly selects one attribute for perturbation and generates reports for the remaining attributes by reusing the perturbed value with probability $p_y$. We detail the design as follows.

\subsubsection{Detailed Design.} 
Let $d\ge 2$ denote the number of attributes, and $\mathbf{x}_i=(x_{i,1},\ldots,x_{i,d})$ the raw attribute vector of user $i$, where each $x_{i,j}$ is drawn from a common domain $\mathcal{D}$ of size $k=|\mathcal{D}|$. We denote by $n_1$ and $n_2=n-n_1$ the number of users participating in Phase~I and Phase~II, respectively. Corr-RR operates in the following two phases.

\paragraph{\textbf{Phase~I: Dependency Learning.}}
A subset of $n_1 \ll n$ users applies SPL: each user perturbs all $d$ attributes using GRR with per-attribute budget $\epsilon/d$. For attribute $j\in[d]$, the sanitized report $y_{i,j}$ is generated as:
\[
\Pr(y_{i,j}=v' \mid x_{i,j}=v)=
\begin{cases}
p_1=\dfrac{e^{\epsilon/d}}{e^{\epsilon/d}+k-1}, & v'=v,\\[6pt]
q_1=\dfrac{1}{e^{\epsilon/d}+k-1}, & v'\neq v,
\end{cases}
\]
where $v,v'\in\mathcal{D}$ and $p_1+(k-1)q_1=1$. Aggregating the $n_1$ reports yields unbiased marginal estimates:
\[
\hat f_j^{I}(v)=\frac{I_j^{I}(v)-n_1 q_1}{n_1(p_1-q_1)}, 
\quad 
I_j^{I}(v)=\sum_{i=1}^{n_1}\mathbb{I}(y_{i,j}=v).
\]

These marginals preserve statistical signals of inter-attribute dependencies, though marginals alone do not fully capture correlation. Next, we use these marginals to infer the pairwise parameter $p_{j\leftrightarrow s}\in[0,1]$ for each attribute pair $(j,s)$, according to the closed-form expression derived from the MSE minimization procedure in Section~\ref{py_determination}. For brevity, we denote this parameter as $p_y$.

\paragraph{\textbf{Phase~II: Correlation-Aware Perturbation.}} 
Each of the $n_2$ remaining users, $u_i$, selects one pivot attribute uniformly at random and perturbs it with full budget $\epsilon$:
\[
\Pr(y_{i,s}=v' \mid x_{i,s}=v)=
\begin{cases}
p_2=\dfrac{e^{\epsilon}}{e^{\epsilon}+k-1}, & v'=v,\\[6pt]
q_2=\dfrac{1}{e^{\epsilon}+k-1}, & v'\neq v,
\end{cases}
\]
where $p_2+(k-1)q_2=1$.

For every non-pivot attribute $X_j$ ($j\neq s$), the user reports:
\[
\Pr(y_{i,j}=v' \mid y_{i,s}=v)=
\begin{cases}
p_y, & v'=v,\\
\frac{1-p_y}{k-1}, & v'\neq v.
\end{cases}
\]

That is, with probability 
$p_y$, the mechanism copies the pivot's privatized value; otherwise, sample uniformly from the remaining domain values. This operates only on privatized values and thus constitutes post-processing, incurring no additional privacy loss.

For each attribute $j\in[d]$, the server forms estimates:
\[
\hat f_j^{II}(v)=\frac{I_j^{II}(v)-n_2 q_2}{n_2(p_2-q_2)},
\quad 
I_j^{II}(v)=\sum_{i=n_1+1}^{n}\mathbb{I}(y_{i,j}=v).
\]

The final estimator combines both phases:
\[
\hat f_j(v)=\frac{n_1 \hat f_j^{I}(v)+n_2 \hat f_j^{II}(v)}{n}.
\]
Since $n_1 \ll n_2$, the combined estimate is dominated by Phase~II.

\textbf{Bias Note.} The Phase~I GRR estimator of Corr-RR is unbiased because each attribute is perturbed independently, as established in~\cite{erlingsson2014rappor}. In contrast, during Phase~II, Corr-RR generates non-pivot attribute values conditionally on the privatized pivot value rather than perturbing them independently. This conditional generation violates the independence assumption underlying the unbiasedness of the estimator in standard randomized response, thereby introducing a potential bias in the estimated frequencies.
Nevertheless, empirical results show that, despite this bias, Corr-RR consistently achieves lower MSE than unbiased baselines, while preserving the same $\epsilon$-LDP guarantee (see~\ref{syntheticresults}). This indicates that the accuracy gain from exploiting inter-attribute correlations outweighs the effect of the induced bias. Notably, Corr-RR is explicitly designed to prioritize minimizing overall estimation error rather than ensuring strict unbiasedness of the estimator.

\subsubsection{Determination of $p_y$}\label{py_determination} Since the Phase II estimator is biased, we use its MSE to quantify the achieved data utility and then determine the parameter, $p_y$ for each attribute pair $(X_j,X_s)$ by minimizing the MSE.  For categorical value $v\in\{0,\ldots,k-1\}$,  we first compute the MSE of $\hat f_j^{II}(v)$.

\begin{theorem}\label{thm:MSE_categorical_general}
For categorical value $v\in\mathcal{D}$, the MSE of the Phase~II estimator $\hat f_j^{II}(v)$ is

\[
\mathrm{MSE}\big[\widehat f^{II}_j(v)\big]
= \frac{\bigl(q_2+\Delta\,\mu_v\bigr)\bigl(1-q_2-\Delta\,\mu_v\bigr)}{n_2\,\Delta^2}
\;+\; \Big(\mu_v - f_j(v)\Big)^2,
\]
where $\mu_v
= \tfrac12\![f_j(v)+\frac{1-p_y}{k-1}+a\,f_s(v)]$, $a=\frac{k p_y-1}{k-1}$, and $\Delta=\dfrac{e^\epsilon-1}{e^\epsilon+k-1}$.  
$f_j(v)$ is the true frequency of targeted attribute $j$ with value $v$, and
$f_s(v)$ denotes the true frequency of $v$ for the selected attribute $s$ used to conditionally generate attribute $j$.
\end{theorem}  
\begin{proof}
   (See Appendix~\ref{appendix:thm:mse} for the full proof of Theorem~\ref{thm:MSE_categorical_general}.) 
\end{proof}

As a result, the average MSE across all possible $k$ values can be rewritten as a function of $p_y$, which is given by: 
\begin{align}
\MSE_{\mathrm{avg}}(p_y)
&=\frac{1}{k}\sum_{v=0}^{k-1}\MSE[\hat f_j^{II}(v)]
\end{align}

Finally, we determine $p_y$ by minimizing $\MSE_{\mathrm{avg}}(p_y)$. Differentiating the average MSE and setting the derivative to zero yields the following closed-form solution.

\begin{proposition}\label{thm:py_categorical_general}
Using the notation above, the unconstrained minimizer of $\MSE_{\mathrm{avg}}(p_y)$ is
\[
  p^\ast_y
  \;=\;
  -\frac{
\displaystyle \sum_{v} C_1(v)\Bigl[\,2\alpha\,C_0(v) + \beta_v\,\Bigr]
}{
2\alpha\,\displaystyle \sum_{v} C_1(v)^2
}
\]
where $\alpha=1-\frac{1}{n_2}$, $\beta_v=-2f_j(v)+\frac{1-2q_2}{n_2\Delta}$, $C_0(v)=\tfrac12\![f_j(v)+\frac{1-f_s(v)}{k-1}]$, and $C_1(v)=\tfrac12\![\frac{k f_s(v)-1}{k-1}]$.
\end{proposition}
(See Appendix~\ref{appendix:thm:py} for the full derivation of Proposition~\ref{thm:py_categorical_general}.)

\textbf{Procedure for Determining $p_y$ in Practice.}
Since the true frequencies $f_j(v)$ and $f_s(v)$ are not available in practice, we use the Phase~I estimates $\hat{f}^I_j(v)$ and $\hat{f}^I_s(v)$ to approximate them. We then compute $p^\ast_y$ using Proposition~\ref{thm:py_categorical_general} with these marginal estimates.
Note that $p_y \in [0,1]$ rather than the unconstrained domain $\mathbb{R}$. 
To account for boundary solutions, we also evaluate the MSE at $p_y = 0$ and $p_y = 1$ and compare these values with the MSE achieved at $p^\ast_y$.
The final value of $p_y$ is chosen as the one that minimizes $\MSE_{\mathrm{avg}}(p_y)$.
This procedure enables Corr\text{-}RR to adaptively select dependency-aware parameters without accessing raw attribute values. Figure~\ref{fig:py-all} illustrates how $p_y$ varies as a function of $(\hat{f}^I_j, \hat{f}^I_s)$. When the marginals are similar, the optimizer selects a larger $p_y$, increasing the likelihood that non-pivot reports follow the pivot's value. As the marginals diverge, the optimizer reduces $p_y$.

Phase~I's marginal estimates $\hat{f}^I_j(v)$ and $\hat{f}^I_s(v)$ contain noise that propagates into the computation of $p_y^*$. When Phase~I includes too few users or a small privacy budget, this noise becomes substantial, producing unstable $p_y$ values. Conversely, with a larger Phase~I subset or moderate $\epsilon$, the marginals better approximate their true distributions, yielding more stable $p_y$ values and improved utility. Empirically, allocating about 10--20\% of users to Phase~I (for moderate $\epsilon$), depending on attribute number and domain size, provides a good balance between reliable $p_y$ determination and preserving Phase~II accuracy.

\textbf{Extension to Multiple Attribute Pairs.} When users have multiple attribute pairs ($d\geq 3$), Corr-RR naturally extends its pairwise design as follows. In Phase~I, the server computes pair-specific dependency parameter $p_{j\leftrightarrow s}\in[0,1]$ for every ordered attribute pair $(j,s)$ with $j\neq s$. Each parameter is derived from privatized marginals via the MSE-minimization procedure in Section~\ref{py_determination}, so that every potential pivot and non-pivot relationship is equipped with its own optimized mapping probability. In Phase~II, each user selects a pivot attribute $X_j$ uniformly at random and perturbs it with the full budget $\epsilon$ to obtain $y_{i,j}$. The remaining attributes $\{X_s: s\neq j\}$ are then reconstructed independently from $y_{i,j}$ using their pairwise parameters relative to the chosen pivot: for each $s$, the privatized output $y_{i,s}$ is sampled by copying $y_{i,j}$ with probability $p_{j\leftrightarrow s}$ and, otherwise, sampling a value uniformly from $\mathcal{D}\setminus\{y_{i,j}\}$. For instance, with four attributes $(X_1,X_2,X_3,X_4)$, if the pivot is $X_1$, the user generates $(y_{i,2},y_{i,3},y_{i,4})$ using $(p_{1\leftrightarrow 2}, p_{1\leftrightarrow 3}, p_{1\leftrightarrow 4})$, respectively.

\begin{figure}[t]
    \centering
    \includegraphics[width=0.95\linewidth]{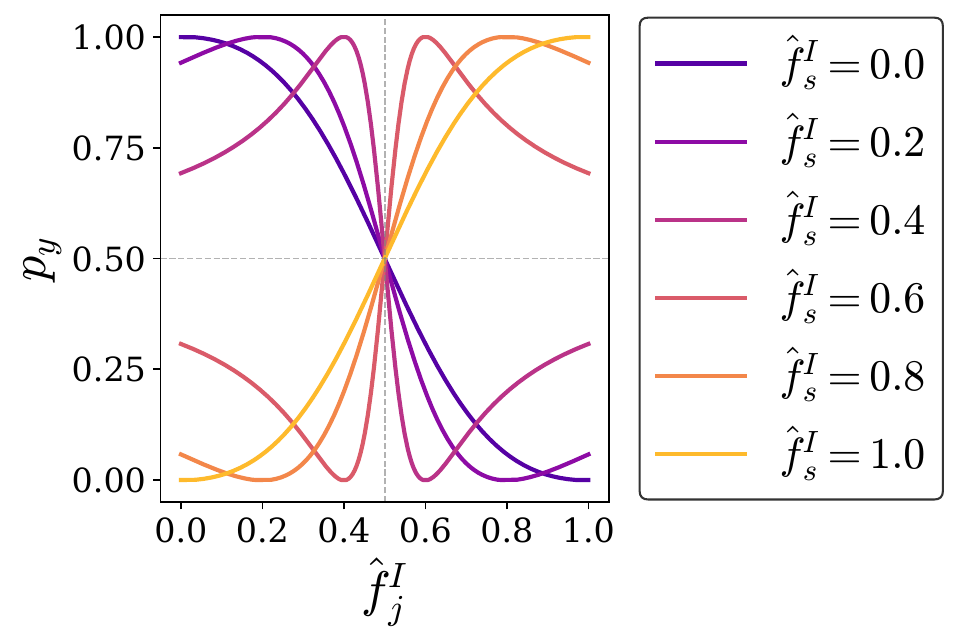}
    \caption{Optimal correlation‑aware probability \(p_y\) as a function of
             the marginal \(\hat{f}^I_j\) for representative values of
             \(\hat{f}^I_s\).}
    \Description{Line plot showing the optimal correlation-aware probability p_y as a function of the estimated marginal frequency of one attribute, with separate curves for different marginal values of a second attribute.}

    \label{fig:py-all}
\end{figure}

\subsubsection{Privacy Analysis}
\begin{theorem}
\label{thm:Corr-RR_proof_d}
\textnormal{Corr-RR} satisfies $\epsilon$-LDP.
\end{theorem}

\begin{proof}
    The detailed proof is provided in Appendix~\ref{appendix:thm:Corr-RR_proof_d}
\end{proof}

\begin{algorithm}[t]
\caption{Correlated Randomized Response (Corr-RR)}
\label{alg:corr-rr-multi}
\DontPrintSemicolon
\SetAlgoVlined
\KwIn{$n,\epsilon,n_1,d,k$;\quad each user $u_i$ holds $\mathbf{x}_i\in\mathcal{D}^d$}
\KwOut{Marginal estimates $\{\hat f_j(v)\}_{j\in[d],\,v\in\mathcal{D}}$}

\textbf{Phase~\MakeUppercase{\romannumeral1} (users $1,\dots,n_1$): Private dependency learning}\;
\ForEach{user $u_i$}{
  \For{$j\gets 1$ \KwTo $d$}{
     $y_{i,j}\leftarrow \text{GRR}(x_{i,j},\,\epsilon/d)$\;
  }
  Send $\mathbf{y}_i$ to server\;
}
Server: compute debiased marginals $\hat f_j^{I}(v)$ for all $j\in[d],\,v\in\mathcal{D}$\;
Server: for every ordered pair $(j,s)$ with $j\neq s$, compute $p_{j\leftrightarrow s}\in[0,1]$ by minimizing $\MSE_{\mathrm{avg}}$ per Section~\ref{py_determination}; store matrix $P=[p_{j\leftrightarrow s}]$\;

\BlankLine
\textbf{Phase~\MakeUppercase{\romannumeral2} (users $n_1+1,\dots,n$): Correlation-aware reporting}\;
\ForEach{user $u_i$}{
  Sample pivot $j \sim \mathrm{Unif}([d])$\;
  $y_{i,j}\leftarrow \text{GRR}(x_{i,j},\,\epsilon)$ \tcp*{full-budget privatization}
  \For{$s\gets 1$ \KwTo $d$ \textbf{where} $s\neq j$}{
      With prob.\ $p_{j\leftrightarrow s}$ set $y_{i,s}\leftarrow y_{i,j}$;\;
      Otherwise set $y_{i,s}\leftarrow \text{rand}(\mathcal{D}\setminus\{y_{i,j}\})$\;
  }
  Send $\mathbf{y}_i$ to server\;
}
Server: compute $\hat f_j^{II}(v)$ for all $j\in[d],\,v\in\mathcal{D}$ using GRR debiasing with budget $\epsilon$\;
\BlankLine
\textbf{Output (aggregation)}:\quad $\hat f_j(v)=\dfrac{n_1\,\hat f_j^{I}(v)+n_2\,\hat f_j^{II}(v)}{n}$ for all $j\in[d],\,v\in\mathcal{D}$\;
\end{algorithm}

\subsection{Discussion}
We now discuss key properties and potential extensions of our two-phase framework.

\subsubsection{Privacy Amplification}
In Phase~II of Corr-RR, each user perturbs only one randomly chosen attribute out of $d$, corresponding to an attribute-level sampling rate $\beta=1/d$.  
By the standard privacy amplification by sampling argument~\cite{li2012sampling, arcolezi2021random}, an $\epsilon$-LDP mechanism under such sampling is equivalent to an \emph{amplified} guarantee with effective parameter:
\[
\epsilon' = \ln\!\left( \tfrac{1}{\beta}(e^{\epsilon}-1) + 1 \right),
\qquad \beta = \tfrac{1}{d}.
\]
Although Phase~I does not involve sampling (as each of the $n_1$ users perturbs all attributes with $\epsilon/d$), Phase~II dominates when $n_2 \gg n_1$, so in practice Corr-RR can be interpreted as enjoying an amplified privacy bound.  

\subsubsection{Iterative Refinement}
Corr-RR relies on dependency parameters ($p_y$ values) estimated from a small Phase~I subset. Over time, correlations may shift or estimation errors may accumulate. As a result, the $p_y$ values determined from the initial marginal estimates can gradually diverge from the true underlying correlations, reducing the accuracy of frequency estimation.
A natural enhancement is to periodically refresh Phase~I with additional samples, re-estimating marginals and updating $p_y$ values. This gradual recalibration reduces estimation noise, ensuring that the learned dependencies, $p_y$, remain aligned with the evolving data distribution, thereby improving Phase~II reconstructions and overall frequency estimation accuracy.

\subsubsection{Grouping Attributes}
In high-dimensional datasets, not all attributes are equally correlated. A more scalable extension is to partition attributes into $t$ groups based on estimated correlation strength. Each group receives a sub-budget $\epsilon/t$. Within a group, Phase~II operates as in Corr-RR: one pivot attribute is perturbed with GRR and the remaining attributes are reconstructed via $p_y$-guided mappings. Grouping reduces noise for strongly correlated subsets without wasting privacy budget on weakly related attributes.

\subsubsection{Extension to Attributes with Varying Domain Sizes}
While the above discussion of Corr-RR focuses on attribute pairs with identical categorical domains (i.e., $\mathcal{D}_j = \mathcal{D}_k$), the core idea naturally extends to attributes whose domains differ in type or size. A straightforward approach is to convert these attributes into categorical variables with aligned domain sizes through temporary grouping and mapping.
For example, suppose $X_1 \in \{yes, no\}$ and $X_2 \in \{0,1,2,3,4\}$. We can align their domains by grouping and mapping $\{0,1\} \mapsto yes$ and $\{2,3\} \mapsto no$. After this alignment, Corr-RR is applied to generate perturbed users' values based on the mapped categories. If the \emph{perturbed (mapped)} value of $X_2$ is ``yes'', the user randomly selects a value from $\{0,1\}$ as the final reported value; similarly, if it is ``no'', the user randomly selects a value from $\{2,3\}$ as the report. 

\subsubsection{Extension to Higher-order Dependencies}
Corr\text{-}RR currently models pairwise dependencies between attributes, which is computationally efficient and already captures substantial dependency structure in typical datasets. However, real data may exhibit higher-order dependencies in which an attribute is jointly influenced by multiple others. Corr\text{-}RR can be extended to incorporate such higher-order relationships to further improve estimation accuracy.
In Phase~I, rather than estimating only marginal distributions, we can estimate joint distributions over multiple pivot attributes (e.g., the joint distribution of $X_{s_1}$ and $X_{s_2}$). In Phase~II, these pivot attributes jointly determine the reporting behavior of each non-pivot attribute through $\Pr(y_{i,j}=v' \mid y_{i,s_1}=v_1, y_{i,s_2}=v_2)$, with dependency parameters specified for each combination of pivot values. We then derive the corresponding MSE expression and determine the dependency parameters by minimizing the MSE, obtained by taking the partial derivative of the MSE with respect to each parameter.

\subsubsection{Alternative Mechanism under the Two-phase Framework}

Our two-phase framework naturally supports other instantiations. 
One such example is \textbf{Conditional Randomized Response (Cond-RR)}, which extends the idea of correlation-aware synthesis by using \emph{conditional distributions} rather than pairwise probabilities. 

\textit{Phase I: Conditional Modeling.}  
As in Corr-RR, a subset of $n_1$ users perturbs all $d$ attributes with per-attribute budget $\epsilon/d$. Using these reports, the server learns an approximate joint distribution over the attributes. From this joint distribution, the server then derives conditional distributions that summarize how attributes relate to one another in a privacy-preserving way. 

\textit{Phase II: Conditional Synthesis.}  
Each remaining user randomly selects one attribute and perturbs it using their full privacy budget $\epsilon$. The remaining unselected attributes are then indirectly inferred by sampling from the conditional distributions learned in Phase I, conditioned on the user's perturbed attribute. Because this reconstruction depends only on the privatized pivot and publicly released statistics, it is pure post-processing and incurs no additional privacy cost and satisfies $\epsilon$-LDP.

Cond-RR demonstrates that our two-phase design is not limited to correlation-based mappings but can also support dependency-aware synthesis guided by conditional distributions.

\section{Performance Evaluation}\label{sec:evaluation} 
This section evaluates our proposed method against three baseline solutions on synthetic and real-world datasets.

\subsection{Evaluation Metrics}

\subsubsection{Utility Metric} 
We evaluated accuracy using the Mean Squared Error (MSE), a standard metric in the LDP literature~\cite{wang2019collecting, arcolezi2021random}. For each attribute $X_j$ with a common domain $\mathcal{D}$ of size $k$, and each value $v \in \mathcal{D}$, we compute the squared error between the true marginal frequency $f_j(v)$ and its estimate $\hat{f}_j(v)$. Averaging first across values in $\mathcal{D}$ and then across all $d$ attributes yields:
\begin{equation}
\mathrm{MSE} = \frac{1}{d} \sum_{j=1}^{d} 
   \frac{1}{|\mathcal{D}|} \sum_{v \in \mathcal{D}} 
   \bigl(f_j(v) - \hat{f}_j(v)\bigr)^2.
\end{equation}

\subsubsection{Privacy Metric} 
We adopt the privacy parameter $\epsilon$ as our metric, consistent with its standard use in LDP. Smaller $\epsilon$ gives stronger privacy protection but adds more noise. We vary $\epsilon$ in the range $[0.1,0.5]$, which captures practical privacy regimes commonly considered in multi-attribute LDP studies~\cite{arcolezi2021random, wang2021local}.

\subsection{Experimental Setup \& Datasets}

\subsubsection{Environments}
We implement all algorithms in Python~3.10.13 using NumPy~1.23.5 and Pandas~1.5.3. 
To account for randomness, the results are averaged over 100 independent runs. 
All experiments are conducted on a MacBook Pro with an Apple M2 processor and 16\,GB RAM.

\subsubsection{Evaluated Mechanisms}
We evaluate our proposed \textbf{Corr-RR} against standard baselines: \textbf{SPL}~\cite{wang2021local}, \textbf{RS+FD}~\cite{arcolezi2021random}, and \textbf{RS+RFD}~\cite{arcolezi2022risks}. In our evaluation, the methods are instantiated as follows.

SPL and RS+FD operate in single-phase, where all \(n\) users independently perturb their data. The original \emph{RS+RFD} is also single-phase and assumes access to public priors (e.g., perturbed via Laplace) to impute unreported attributes. In our setting, we assume such statistics are \emph{unavailable due to privacy constraints}. For a fair comparison, we therefore use a \emph{two-phase variant} of RS+RFD: in Phase~I, a subset of \(n_1\) users runs SPL to privately estimate marginals, which are normalized and treated as priors; in Phase~II, the remaining \(n_2=n-n_1\) users apply RS+RFD with these privatized priors. Final estimates are obtained via count-weighted aggregation. \emph{Corr-RR (ours)} is inherently two-phase: Phase~I runs SPL on \(n_1\) users to learn privatized dependencies, and Phase~II perturbs a random pivot with the full budget \(\epsilon\) while synthesizing non-pivot attributes using the learned dependency.

\textit{Remark.} 
RS+FD, RS+RFD, and Corr-RR involve subsampling and thus could, in principle, enjoy privacy amplification, yielding an effective privacy level $\epsilon' < \epsilon$. For consistency and fairness, however, we report all results using the same nominal budget $\epsilon$, so that comparisons with SPL remain directly comparable.

\begin{table}[ht!]

\centering

\caption{Pairwise Pearson correlations of SynA and SynB  for three values of $\rho$.} 
\label{tab:synthetic-corr-multi-rho}

\begin{tabular}{lcccccc}
\toprule
\multirow{2}{*}{\textbf{Pair}} 
& \multicolumn{2}{c}{$\rho=0.1$} 
& \multicolumn{2}{c}{$\rho=0.5$}
& \multicolumn{2}{c}{$\rho=0.9$} \\
\cmidrule(lr){2-3} \cmidrule(lr){4-5} \cmidrule(lr){6-7}
& SynA & SynB & SynA & SynB & SynA & SynB \\
\midrule
$(X_1, X_2)$ & 0.100 & 0.100  & 0.491 & 0.491 & 0.899 & 0.899 \\
$(X_1, X_3)$ & 0.094 & 0.012 & 0.495 & 0.241 & 0.900 & 0.897 \\
$(X_1, X_4)$ & 0.100 & 0.011  & 0.508 & 0.128 & 0.903 & 0.812 \\
$(X_2, X_3)$ & 0.007 & 0.098  & 0.231 & 0.490 & 0.808 & 0.804 \\
$(X_2, X_4)$ & 0.009 & 0.006  & 0.247 & 0.243 & 0.810 & 0.730 \\
$(X_3, X_4)$ & 0.007 & 0.099  & 0.251 & 0.495 & 0.811 & 0.904 \\
\midrule
\textbf{Average}
& 0.053 & 0.054
& 0.370 & 0.348

& 0.855 & 0.841

 \\
\bottomrule

\end{tabular}
\end{table}

\subsubsection{Synthetic Datasets} 
We construct two types of synthetic datasets to systematically control inter-attribute dependencies. 
\begin{itemize}
\item \textbf{SynA (Single-Reference).}  
We randomly select one attribute as the pivot (e.g., $X_1$).  
The pivot attribute is sampled independently from the marginal distribution  
$\pi_1 = (0.4, 0.3, 0.2, 0.1)$ over domain $\mathcal{D}$.  
Each remaining attribute copies the value of $X_1$ with probability $\rho$, and otherwise samples a different value from the domain uniformly.  
This produces a clean star-shaped dependency pattern centered at the pivot attribute.
\item \textbf{SynB (Random-Reference).}  
As in SynA, the pivot attribute is randomly selected and generated from $\pi_1$. In contrast, for each subsequent attribute, we randomly select one previously generated attribute as its parent.  
With probability $\rho$, it copies the parent’s value; otherwise, it samples a different value independently.  
This yields a more heterogeneous dependency structure with mixed parent–child relationships.
\end{itemize}
Unless stated otherwise, our default experimental setting uses $n = 20{,}000$ users, $d = 4$ categorical attributes, and domain size $|\mathcal{D}| = 4$.  
Table~\ref{tab:synthetic-corr-multi-rho} summarizes the dataset characteristics (i.e., pairwise Pearson correlations) for SynA and SynB under different values of $\rho$.  
To more thoroughly evaluate Corr\text{-}RR, we additionally generate a variety of SynA and SynB datasets under different choices of $n$ and $d$.

\subsubsection{Real-world Datasets.}
We evaluate our mechanisms on three publicly available datasets: \textit{Clave}~\cite{clave}, \textit{Mushroom}~\cite{mushroom_73}, and \textit{Adult}~\cite{adult_2}. We randomly selected two attributes from the \textit{Clave} dataset, which only contains binary activation vectors. For the \textit{Mushroom} dataset, which consists of categorical attributes describing species characteristics, we first removed attributes with excessively large domains. Next, for the remaining two attributes, whose domain sizes differed, we selected one attribute and merged adjacent categorical values (after encoding) so that its domain size matched that of the other attribute.
For the \textit{Adult} dataset, we selected three categorical attributes, \textit{sex}, \textit{marital-status}, and \textit{relationship}, whose domains are naturally small. Where necessary, we applied the same encoding-based alignment to ensure uniform domain sizes across attributes. All preprocessing steps rely solely on public metadata (e.g., attribute names and declared domain types) and do not use any frequency or correlation information from the raw data, thereby preserving privacy compliance. Dataset characteristics are summarized in Table~\ref{tab:character_table}, while the detailed frequency distributions are provided in Figure~\ref{fig:relative_freq_realdata} in Appendix~\ref{Appendix2}.

\begin{table}[ht!]
\centering
\caption{Characteristics of Real-world Datasets}
\begin{tabular}{l c c c c}
\toprule
\textbf{Dataset} & $d$ & $\boldsymbol{|\mathcal{D}|}$ & \textbf{\# Users} & \textbf{Correlation} \\
\midrule
Clave    & 2 & 2 & 10k & 0.298 \\
Mushroom & 2 & 5 & 8k  & 0.209 \\
Adult    & 3 & 2 & 10k & 0.43--0.95 \\
\bottomrule
\end{tabular}
\label{tab:character_table}
\end{table}

\subsection{Results on Synthetic Data}\label{syntheticresults}

\begin{figure*}[ht!]
  \centering

  \begin{minipage}{.32\textwidth}
    \includegraphics[width=\linewidth]{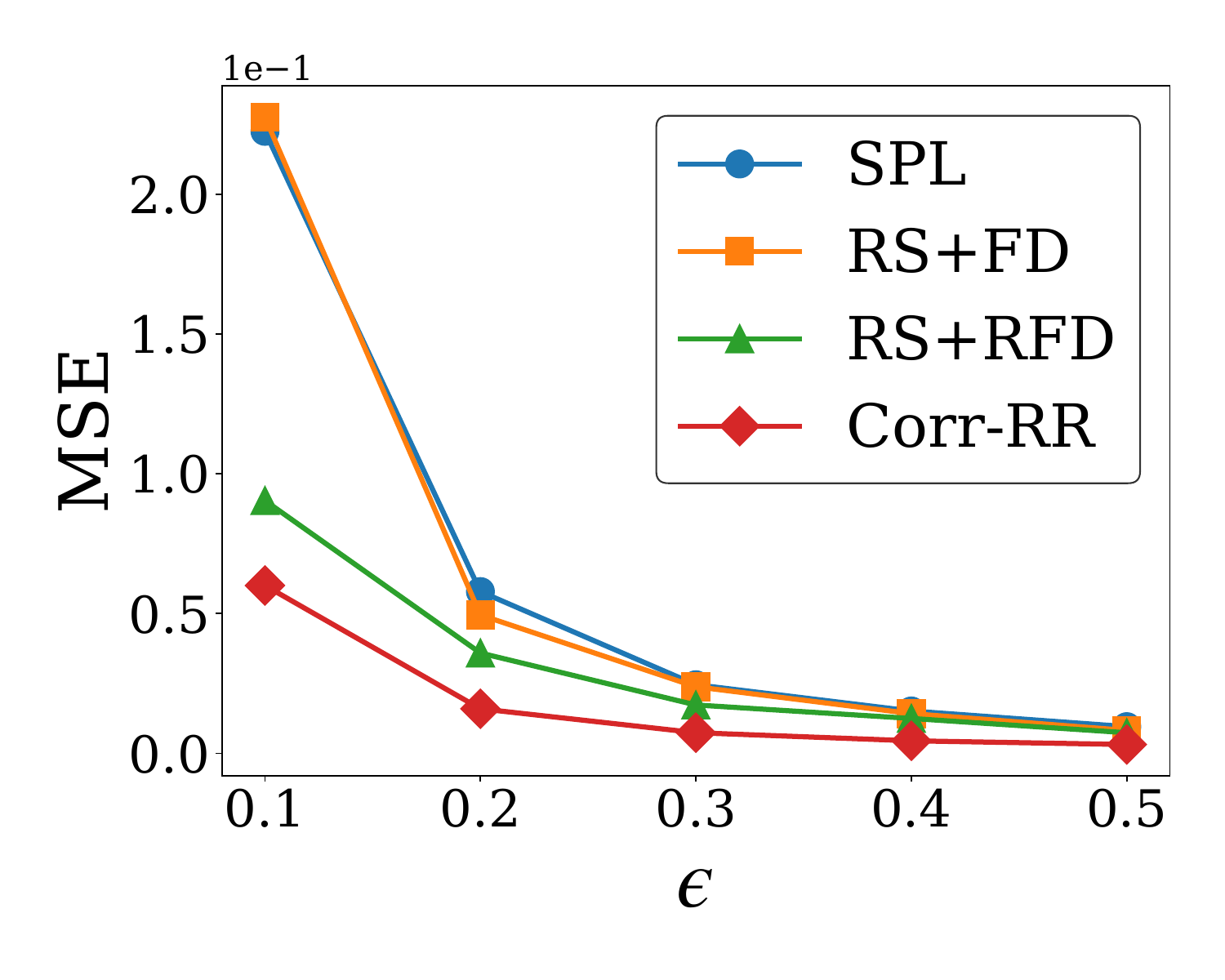}   
    \subcaption{$\rho = 0.1$}
    \label{fig:mse_eps_k4_d4_corr_01}
  \end{minipage}\hfill
  \begin{minipage}{.32\textwidth}
    \includegraphics[width=\linewidth]{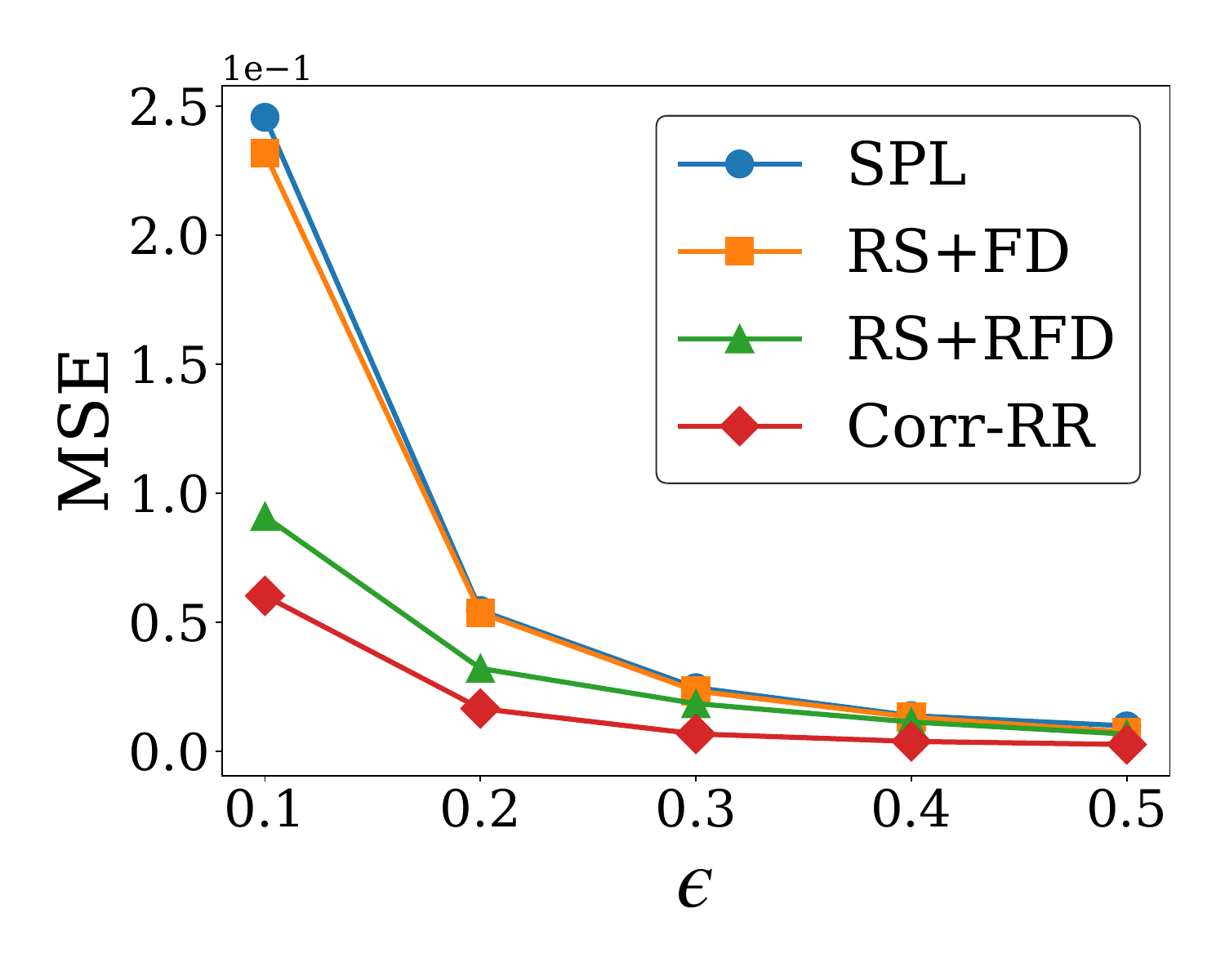}   
    \subcaption{$\rho = 0.5$}
    \label{fig:mse_eps_k4_d4_corr_03}
  \end{minipage}\hfill
  \begin{minipage}{.32\textwidth}
    \includegraphics[width=\linewidth]{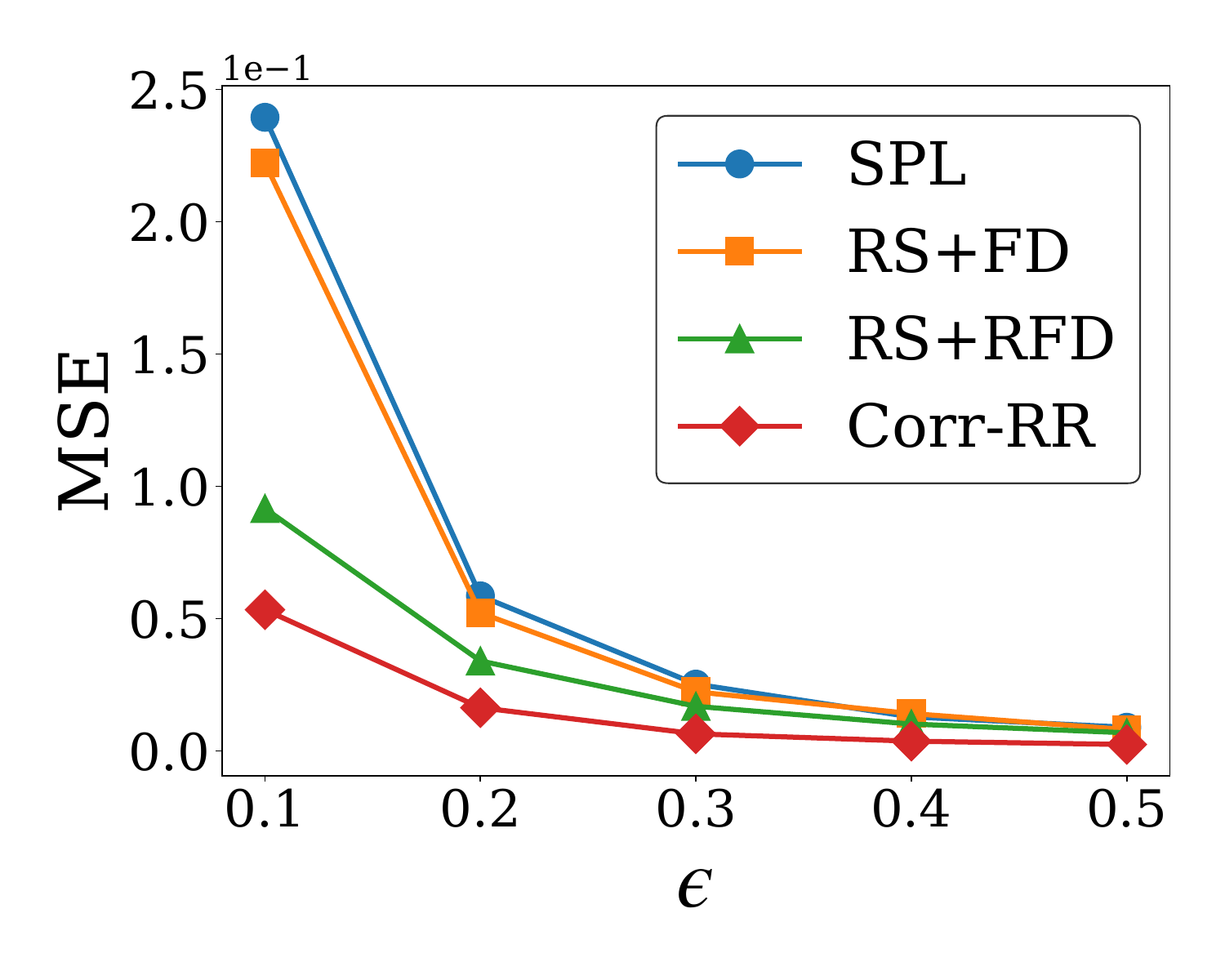}   
    \subcaption{$\rho = 0.9$}
    \label{fig:mse_eps_k4_d4_corr_09}
  \end{minipage}\hfill

\caption{
MSE vs.\ privacy budget $\epsilon$ on SynA with $d=4$ and $|\mathcal{D}|=4$. 
Subplots correspond to different correlations. }
\Description{Three line plots comparing mean squared error versus privacy budget epsilon for four mechanisms: SPL, RS+FD, RS+RFD, and Corr-RR. Each subplot corresponds to a different correlation strength.}

  \label{fig:mse_eps_k4_d4}
\end{figure*}

\begin{figure*}[ht!]

  \centering

  \begin{minipage}{.32\textwidth}
    \includegraphics[width=\linewidth]{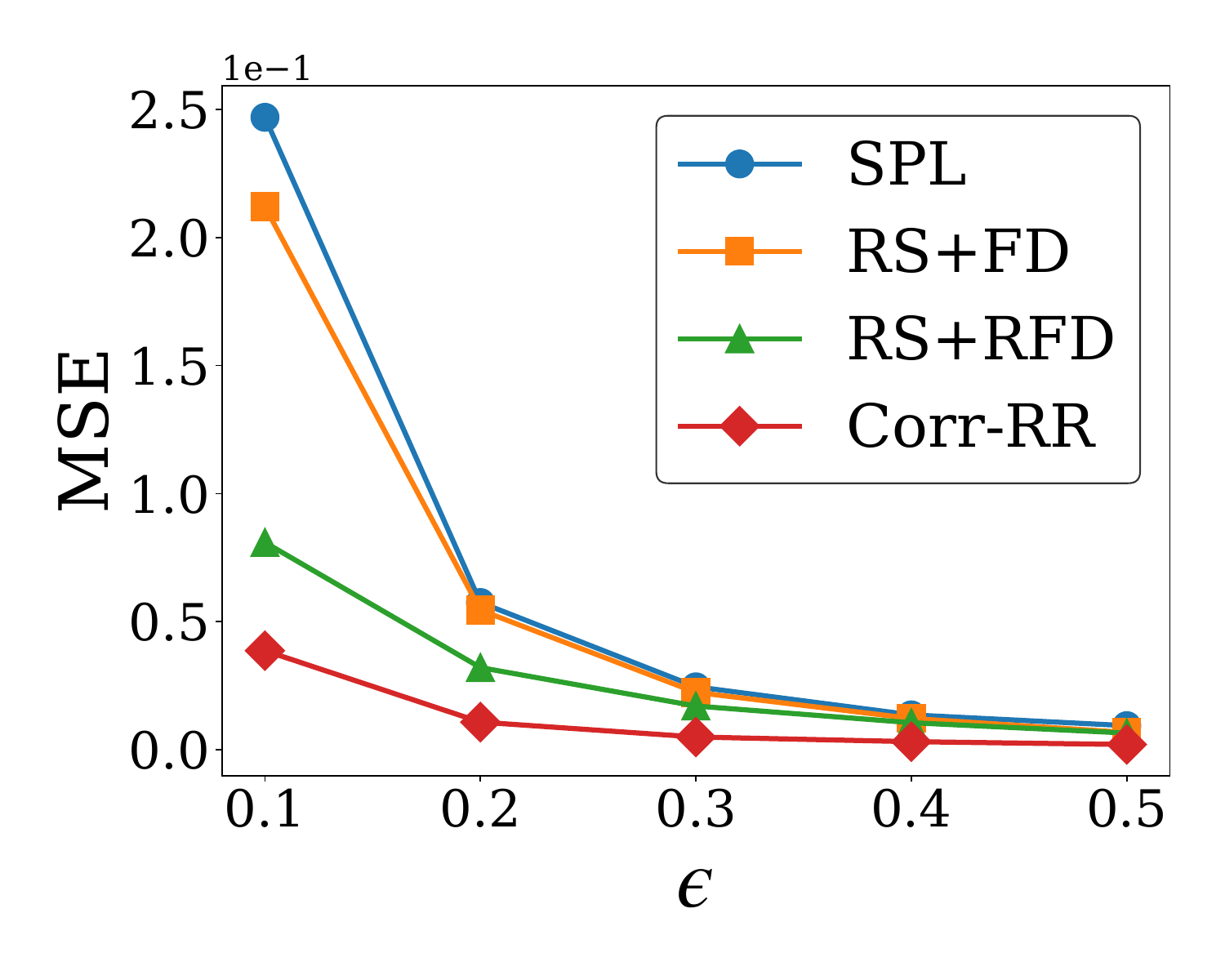}   
    \subcaption{$\rho = 0.1$}
    \label{fig:mse_eps_k4_d4_corr_01_prog}
  \end{minipage}\hfill
  \begin{minipage}{.32\textwidth}
    \includegraphics[width=\linewidth]{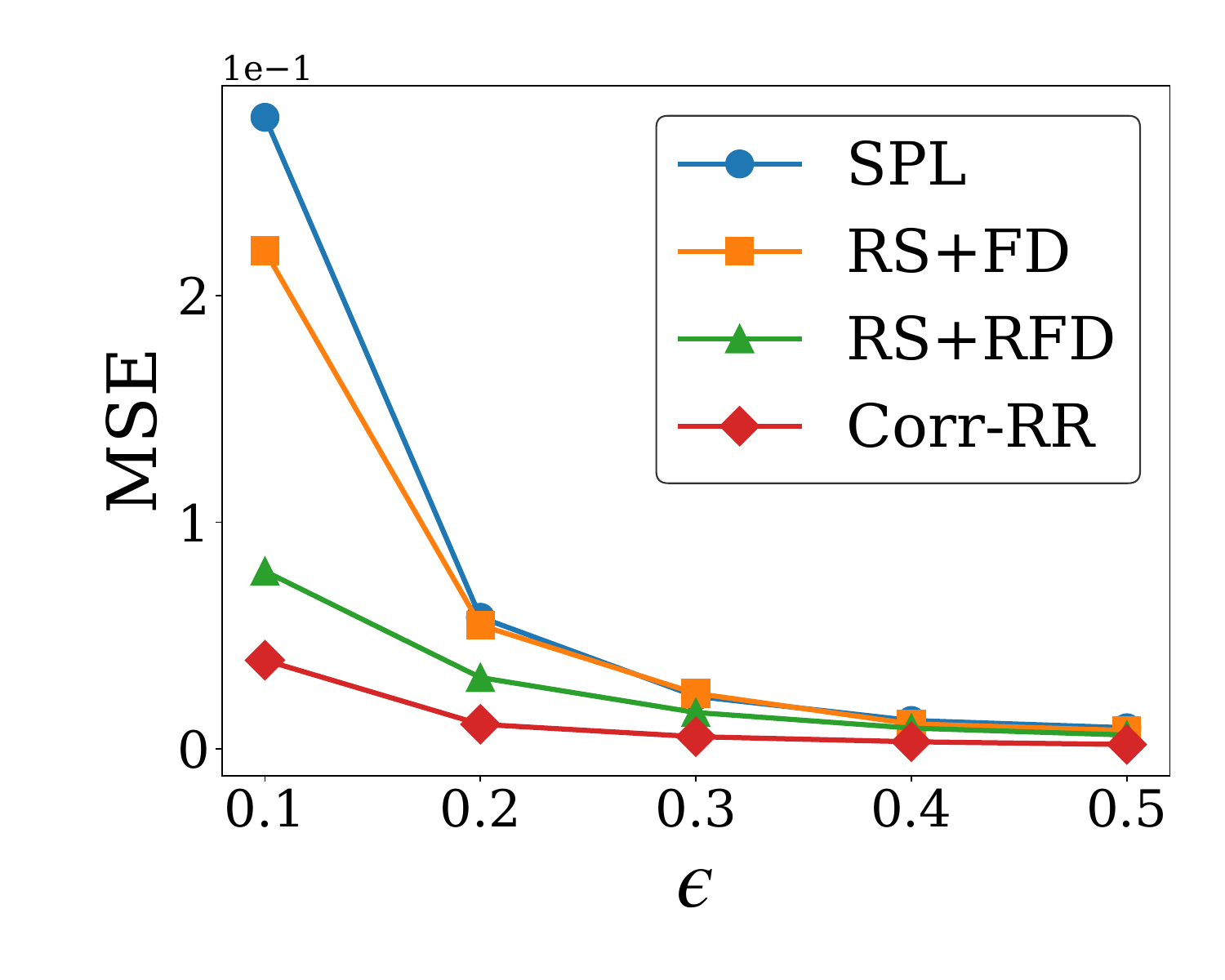}     
    \subcaption{$\rho = 0.5$}
    \label{fig:mse_eps_k4_d4_corr_03_prog}
  \end{minipage}\hfill
  \begin{minipage}{.32\textwidth}
    \includegraphics[width=\linewidth]{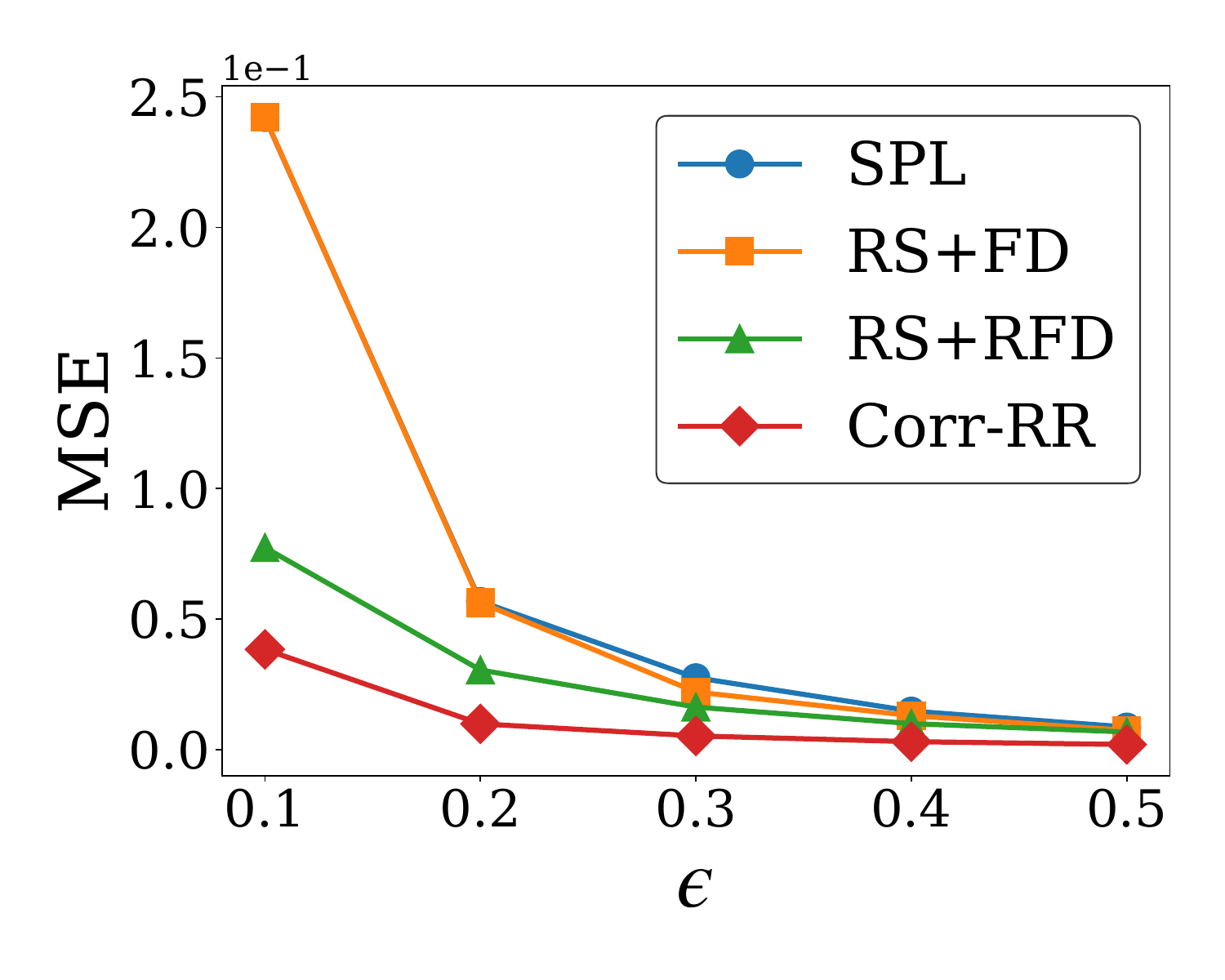}        
    \subcaption{$\rho = 0.9$}
    \label{fig:mse_eps_k4_d4_corr_05_prog}
  \end{minipage}\hfill

\caption{
MSE vs.\ privacy budget $\epsilon$ on SynB with $d=4$ and $|\mathcal{D}|=4$. 
Subplots correspond to different correlations.}
\Description{Three line plots comparing mean squared error versus privacy budget epsilon for four mechanisms: SPL, RS+FD, RS+RFD, and Corr-RR. Each subplot corresponds to a different correlation strength.}

  \label{fig:mse_eps_k4_d4_prog}
\end{figure*}

\subsubsection{Impact of Privacy Budget}
We compare MSE of four LDP mechanisms, SPL, RS+FD, RS+RFD, and Corr-RR, evaluated on two synthetic datasets SynA (Figure~\ref{fig:mse_eps_k4_d4}) and SynB (Figure~\ref{fig:mse_eps_k4_d4_prog}) with domain size $|\mathcal{D}| = 4$, users $n=20{,}000$, and attribute dimensionality $d = 4$.

As expected, MSE decreases as $\epsilon$ increases from 0.1 to 0.5, reflecting the reduced noise at higher privacy budgets. SPL performs worst, since its budget is divided across attributes ($\epsilon/d$ per attribute), which amplifies the error. RS+FD improves upon SPL by allocating the full budget to one attribute, though its fake data limits accuracy. RS+RFD further reduces error by imputing unreported attributes from prior distributions. Corr-RR consistently outperforms all baselines: it concentrates the full budget on a pivot attribute and reconstructs the others using dependency-guided perturbation of the privatized pivot. This design leads to particularly large gains under strong correlations. For example, at $\rho = 0.9$ and $\epsilon = 0.1$, Corr-RR reduces error by more than 70\% compared to SPL, with RS+FD providing almost no improvement and RS+RFD achieving only moderate gains. As $\epsilon$ increases, all mechanisms converge, but Corr-RR remains slightly superior, especially when correlations are high.

\subsubsection{Impact of Number of Attributes}

\begin{figure*}[ht!]
  \centering

  \begin{minipage}{.32\textwidth}
    
    \includegraphics[width=\linewidth]{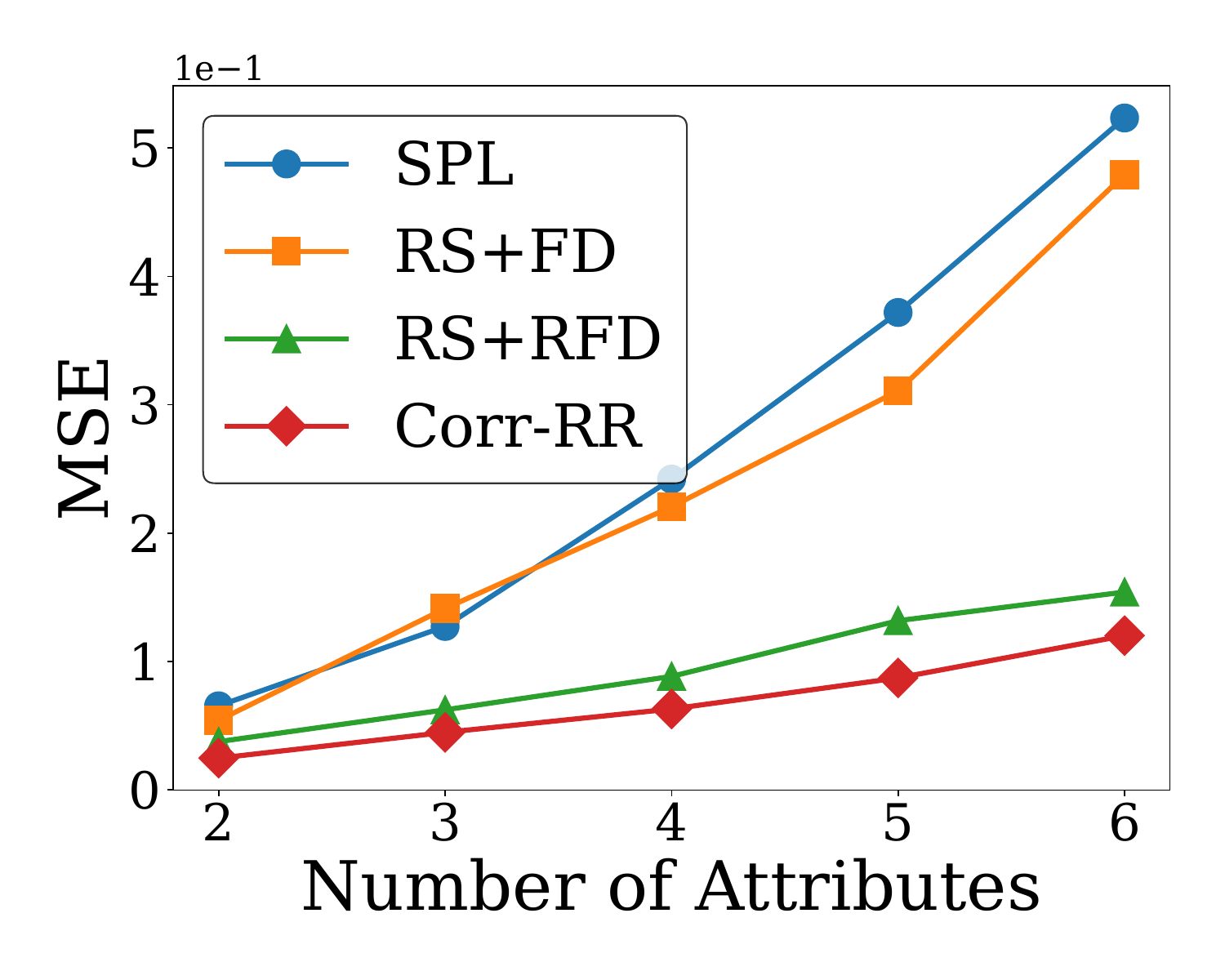}    

    \subcaption{$\epsilon = 0.1$}
    \label{fig:mse_attr_bin_cor09_eps1}
  \end{minipage}\hfill
  \begin{minipage}{.32\textwidth}
   \includegraphics[width=\linewidth]{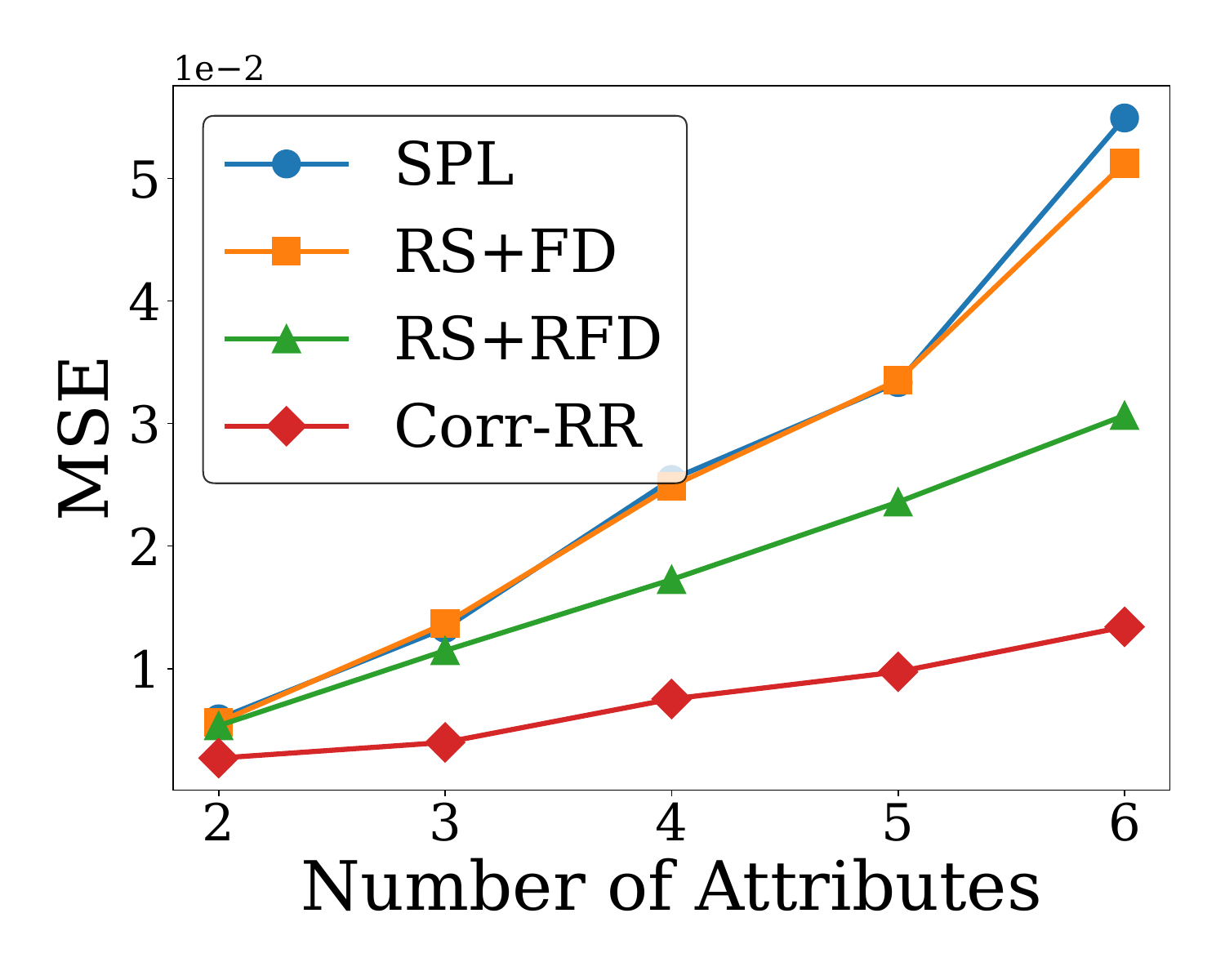}        \subcaption{$\epsilon = 0.3$}
    \label{fig:mse_attr_bin_cor09_eps3}
  \end{minipage}\hfill
  \begin{minipage}{.32\textwidth}
   \includegraphics[width=\linewidth]{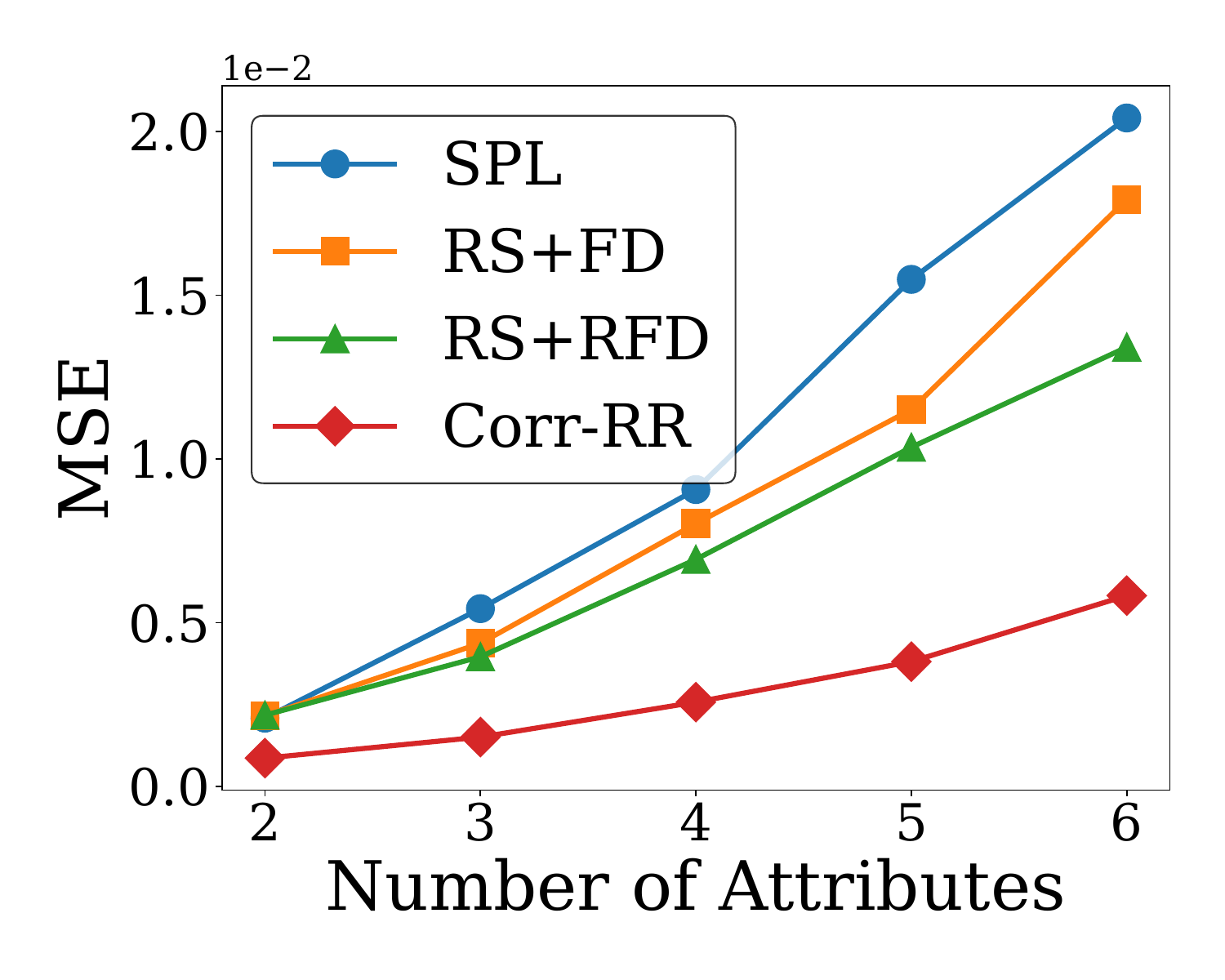}     
    \subcaption{$\epsilon = 0.5$}
    \label{fig:mse_attr_bin_cor09_eps5}
  \end{minipage}\hfill

  \caption{
MSE vs.\ number of attributes on SynA with $\rho=0.9$ and $|\mathcal{D}|=4$. 
Subplots correspond to different privacy budgets.}
\Description{Three line plots showing mean squared error as a function of the number of attributes for different privacy budgets, comparing SPL, RS+FD, RS+RFD, and Corr-RR.}

  \label{fig:mse_vs_attr_bin_cor09}
\end{figure*}


\begin{figure*}[ht!]

  \centering

  \begin{minipage}{.32\textwidth}
    
    \includegraphics[width=\linewidth]{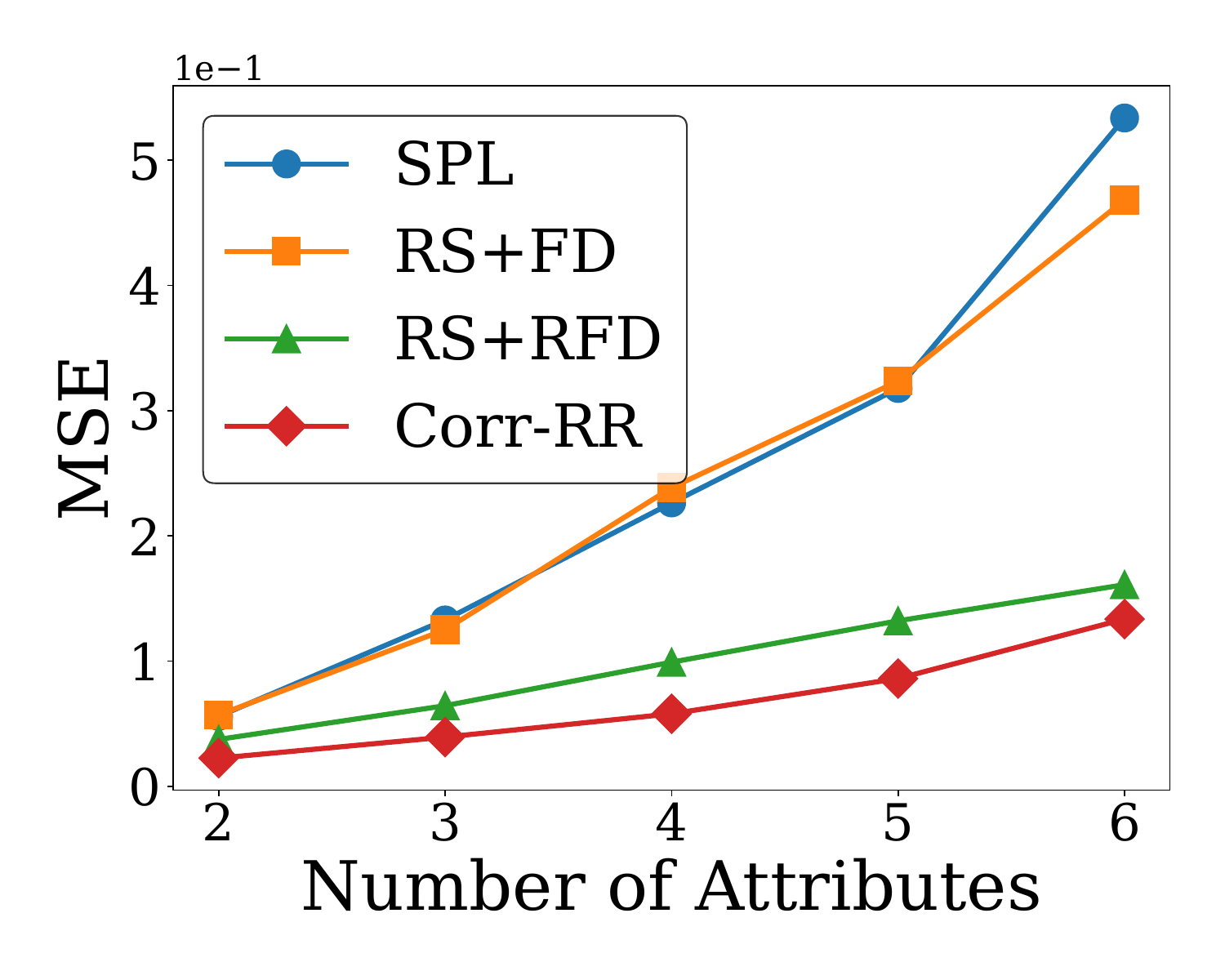}    

    \subcaption{$\epsilon = 0.1$}
    \label{fig:mse_attr_bin_cor09_eps1_prog}
  \end{minipage}\hfill
  \begin{minipage}{.32\textwidth}
   \includegraphics[width=\linewidth]{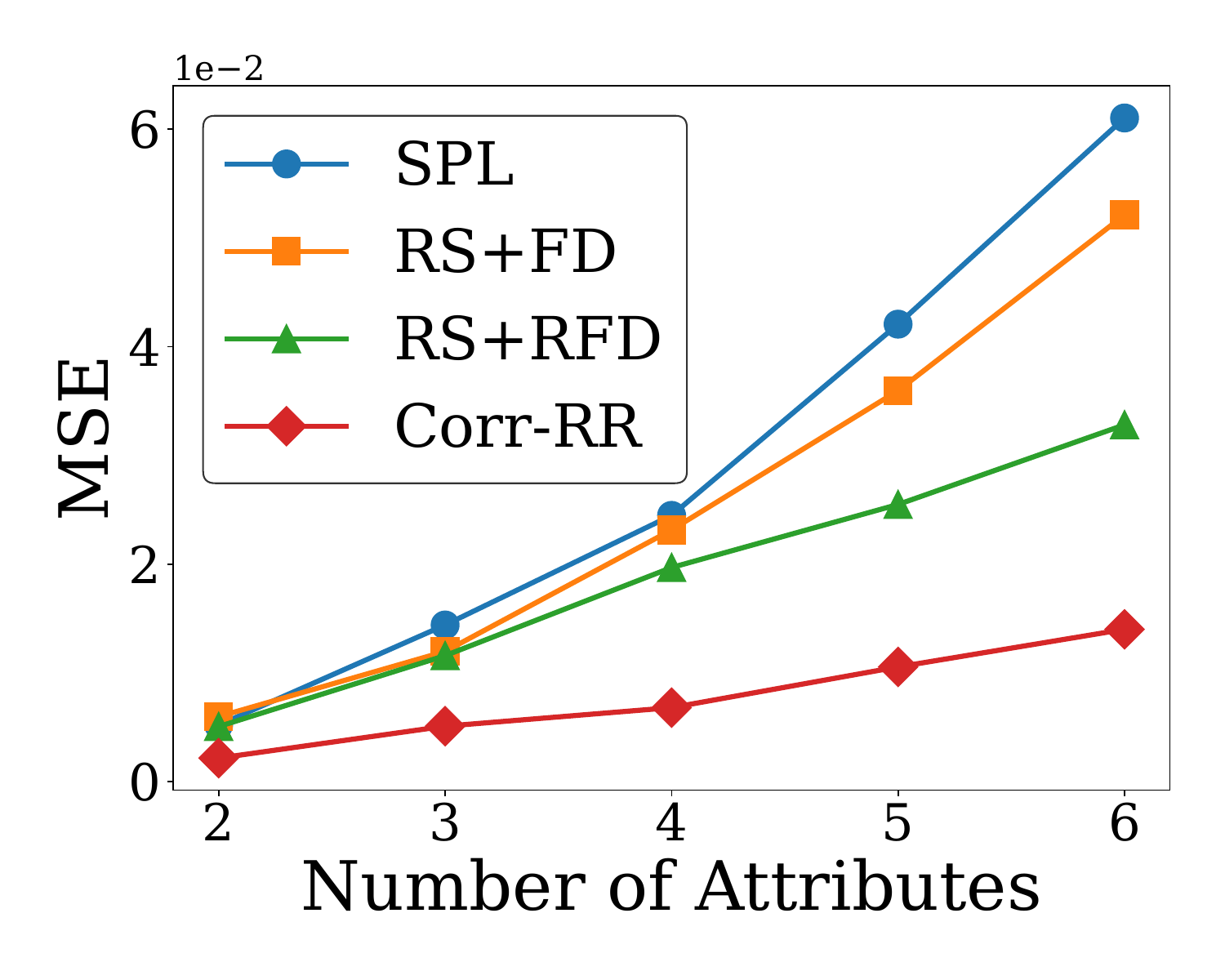}   
   \subcaption{$\epsilon = 0.3$}
    \label{fig:mse_attr_bin_cor09_eps3_prog}
  \end{minipage}\hfill
  \begin{minipage}{.32\textwidth}
 \includegraphics[width=\linewidth]{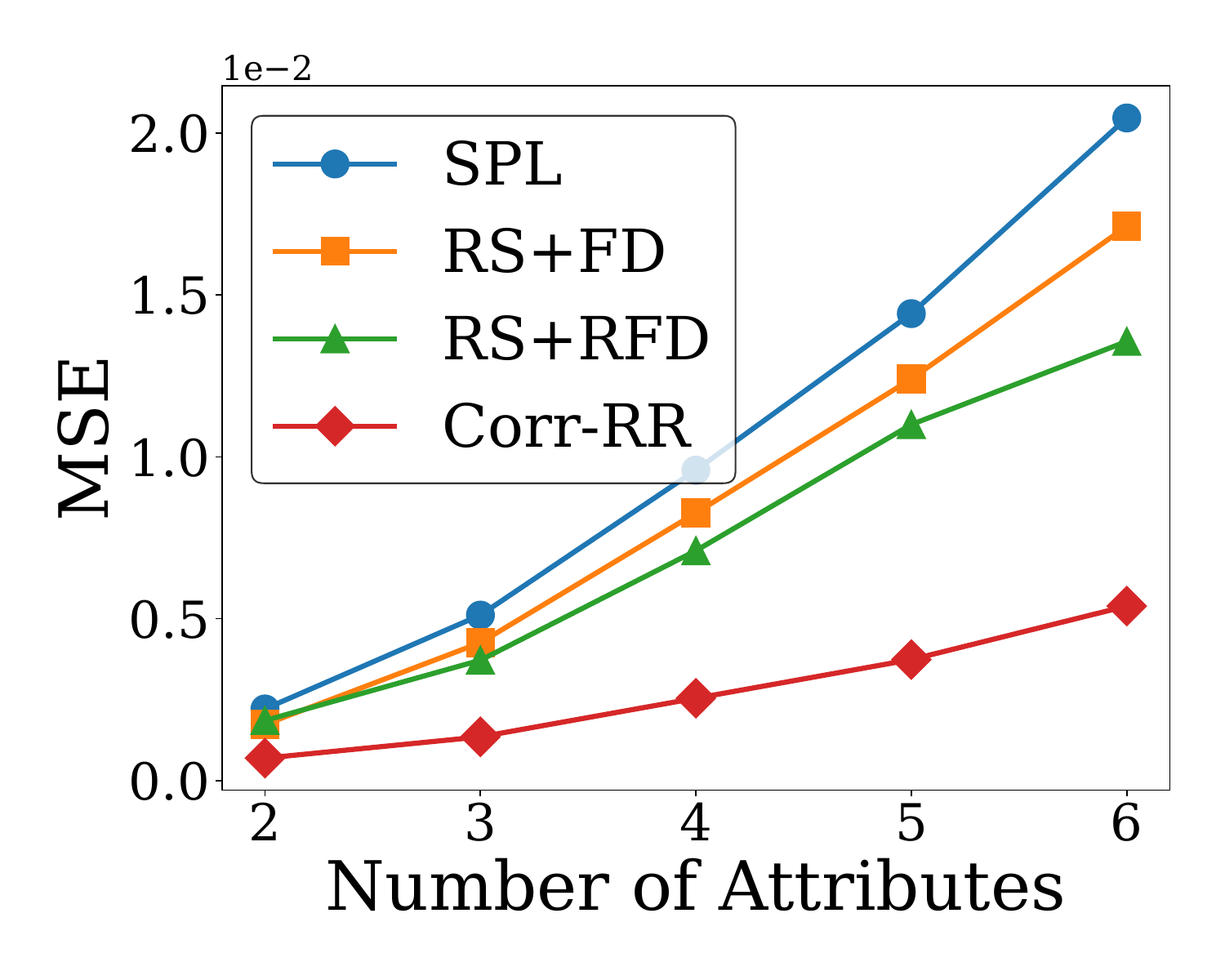}      
    \subcaption{$\epsilon = 0.5$}
    \label{fig:mse_attr_bin_cor09_eps5_prog}
  \end{minipage}\hfill

  \caption{
MSE vs.\ number of attributes on SynB with $\rho=0.9$ and $|\mathcal{D}|=4$. 
Subplots correspond to different privacy budgets.}
\Description{Three line plots showing mean squared error as a function of the number of attributes for different privacy budgets, comparing SPL, RS+FD, RS+RFD, and Corr-RR.}

  \label{fig:mse_vs_attr_bin_cor09_prog}
\end{figure*}

We evaluate how the number of attributes $d \in \{2,\ldots,6\}$ affects estimation accuracy under different privacy budgets. Results for SynA (Figure~\ref{fig:mse_vs_attr_bin_cor09}) and SynB (Figure~\ref{fig:mse_vs_attr_bin_cor09_prog}) are shown for domain size $|\mathcal{D}|=4$, correlation $\rho = 0.9$, and $n = 20{,}000$ users.

We can see that the MSE increases steadily with the number of attributes, reflecting the noise accumulation inherent in high-dimensional local perturbation. Among baselines, SPL consistently suffers the steepest error accumulation, followed by RS+FD, while RS+RFD provides moderate improvement by incorporating prior distributions. In contrast, Corr-RR exhibits markedly lower error accumulation, maintaining the lowest MSE across all $d$ and $\epsilon$ values. The benefits of Corr-RR are particularly pronounced under strong correlation in Figure~\ref{fig:mse_vs_attr_bin_cor09} and Figure~\ref{fig:mse_vs_attr_bin_cor09_prog}. For example, at $\epsilon=0.1$ and $d=6$, Corr-RR reduces error by more than 50\% compared to RS+RFD and by over 4$\times$ relative to SPL. Even at larger budgets, Corr-RR sustains a consistent advantage, highlighting its robustness against dimensionality. 

We also analyze the impact of the number of attributes under a low-correlation setting ($\rho = 0.1$) in Appendix~\ref{Appendix2} for the SynA and SynB datasets, where each subplot presents performance under a specific privacy budget $\epsilon \in \{0.1, 0.3, 0.5\}$.
As shown in Figure~\ref{fig:mse_vs_attr_bin_cor01_syna} (SynA) and Figure~\ref{fig:mse_vs_attr_bin_cor01_synb} (SynB), Corr-RR continues to achieve the best accuracy even under weaker correlation. These results confirm that Corr-RR not only leverages inter-attribute correlation but also scales more gracefully as dimensionality grows, avoiding the severe utility degradation that affects baselines in high-dimensional settings.

\subsubsection{Impact of Correlations}

\begin{figure*}[ht!]
  \centering

  \begin{minipage}{.33\textwidth}
    
      \includegraphics[width=\linewidth]{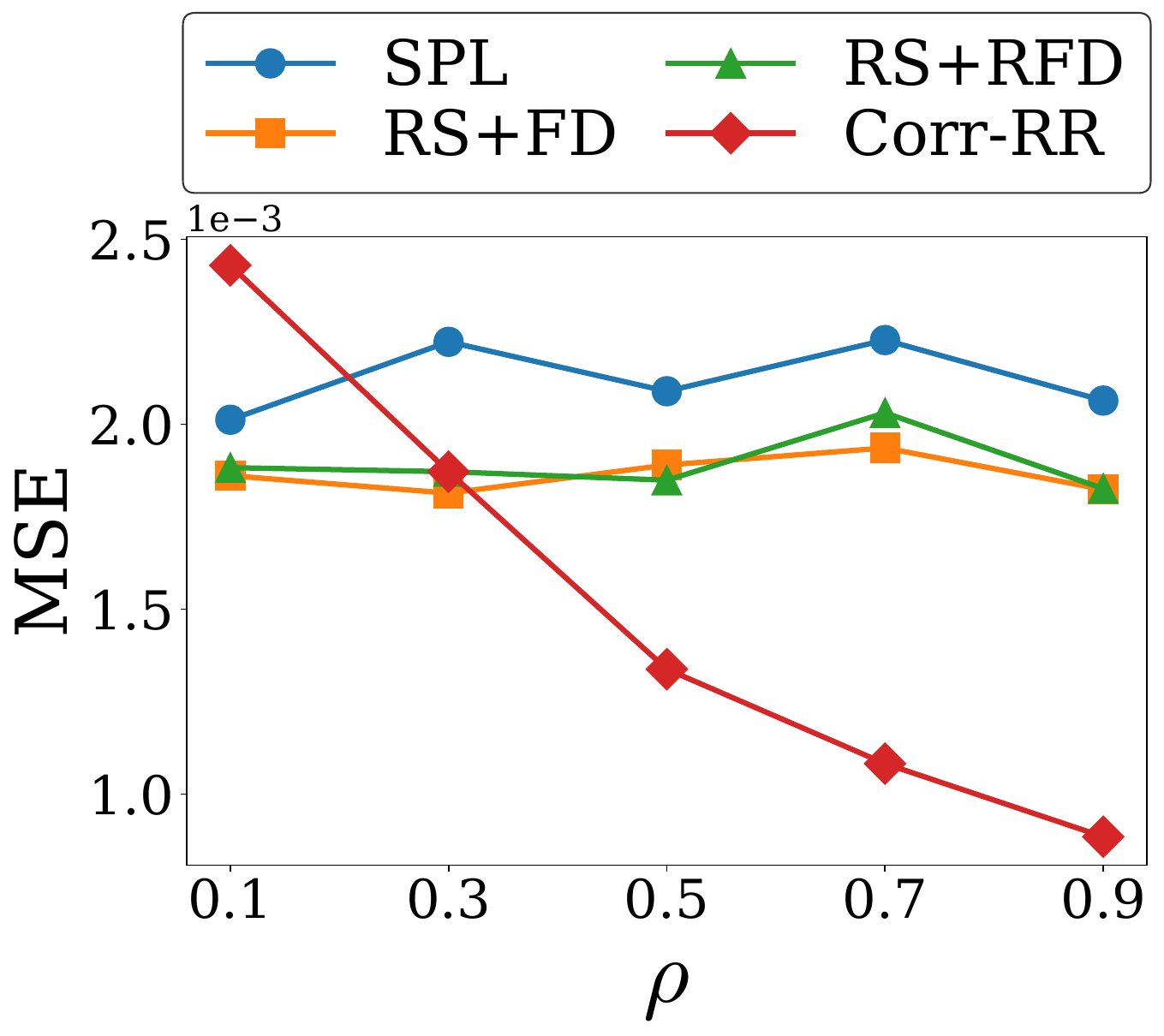}
    \subcaption{$d = 2$}
    \label{fig:correlation_Effect_d2}
  \end{minipage}\hfill
  \begin{minipage}{.33\textwidth}
     \includegraphics[width=\linewidth]{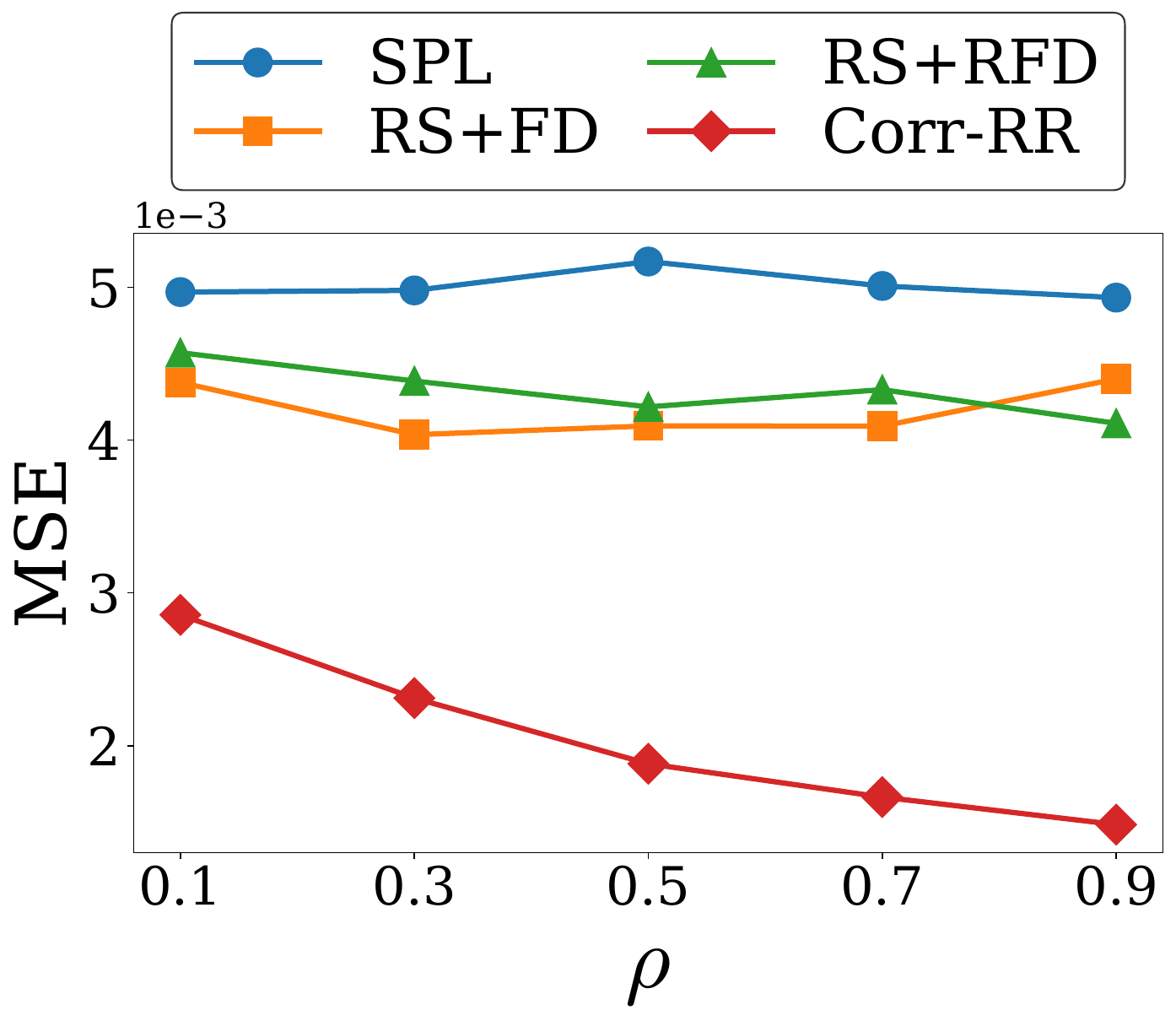}
   
    \subcaption{$d= 3$}
    \label{fig:correlation_Effect_d3}
  \end{minipage}\hfill
  \begin{minipage}{.33\textwidth}
    \includegraphics[width=\linewidth]{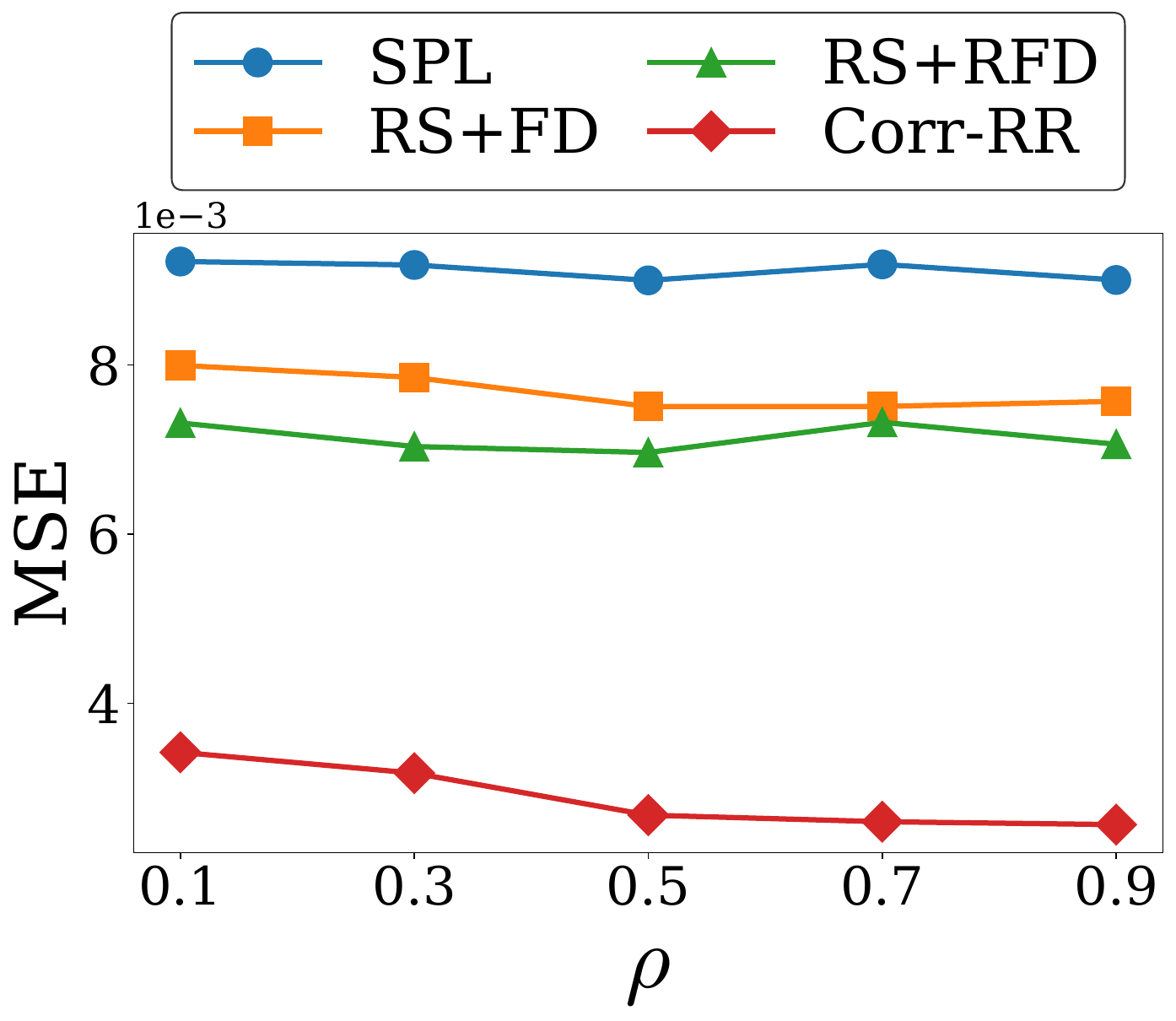}
    
    \subcaption{$d= 4$}
    \label{fig:correlation_Effect_d4}
  \end{minipage}\hfill

  \caption{
MSE vs.\ correlation $\rho$ on SynA with $|\mathcal{D}|=4$ and $\epsilon=0.5$. 
Subplots correspond to different numbers of attributes $d$.
}
  \Description{Line plots showing mean squared error as a function of correlation strength rho for multiple numbers of attributes, comparing SPL, RS+FD, RS+RFD, and Corr-RR.}
  \label{fig:correlation_Effect}
\end{figure*}

\begin{figure*}[ht!]

  \centering

  \begin{minipage}{.32\textwidth}
    
      \includegraphics[width=\linewidth]{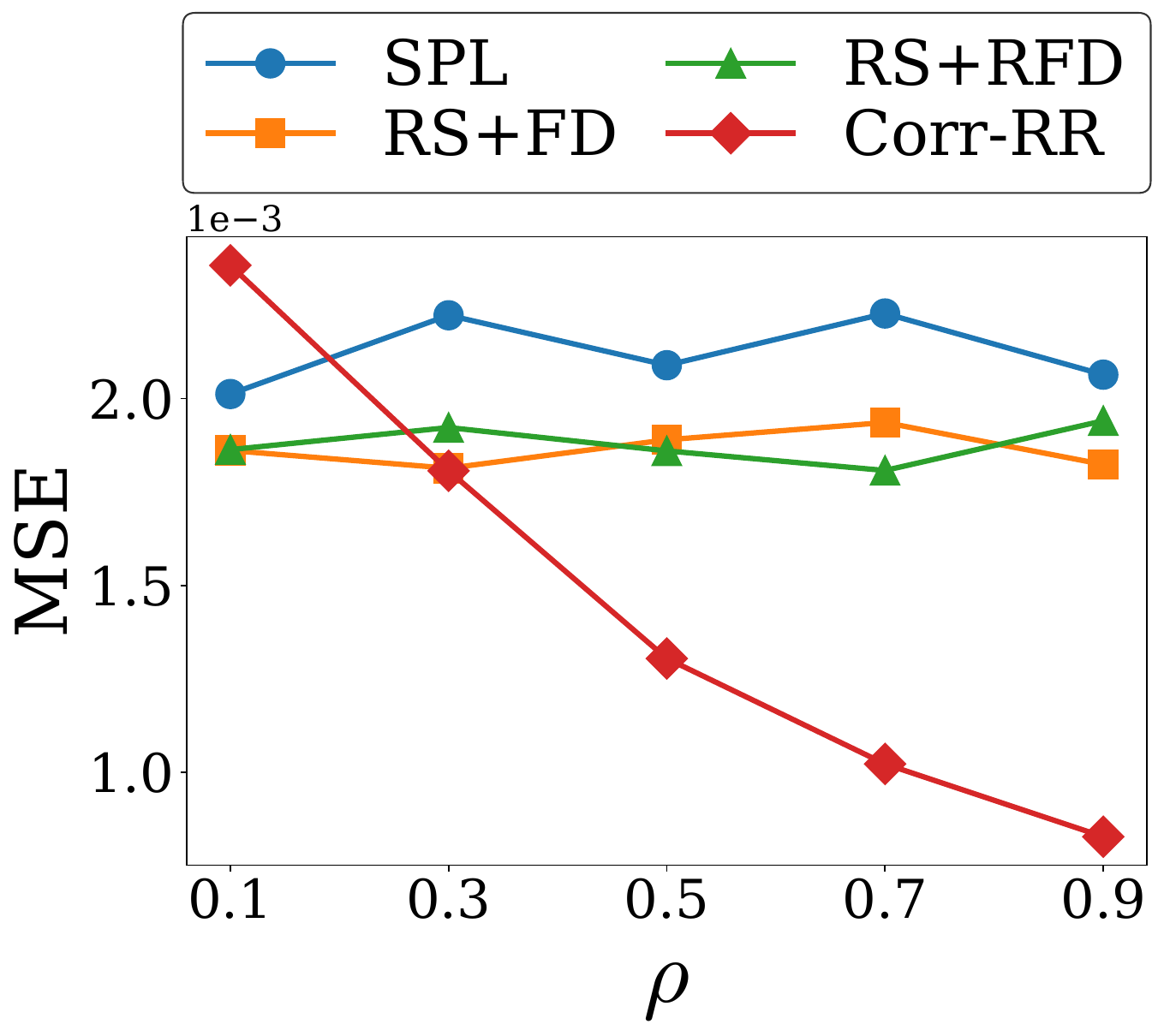}
    \subcaption{$d = 2$}
    \label{fig:correlation_Effect_d2_prog}
  \end{minipage}\hfill
  \begin{minipage}{.32\textwidth}
       \includegraphics[width=\linewidth]{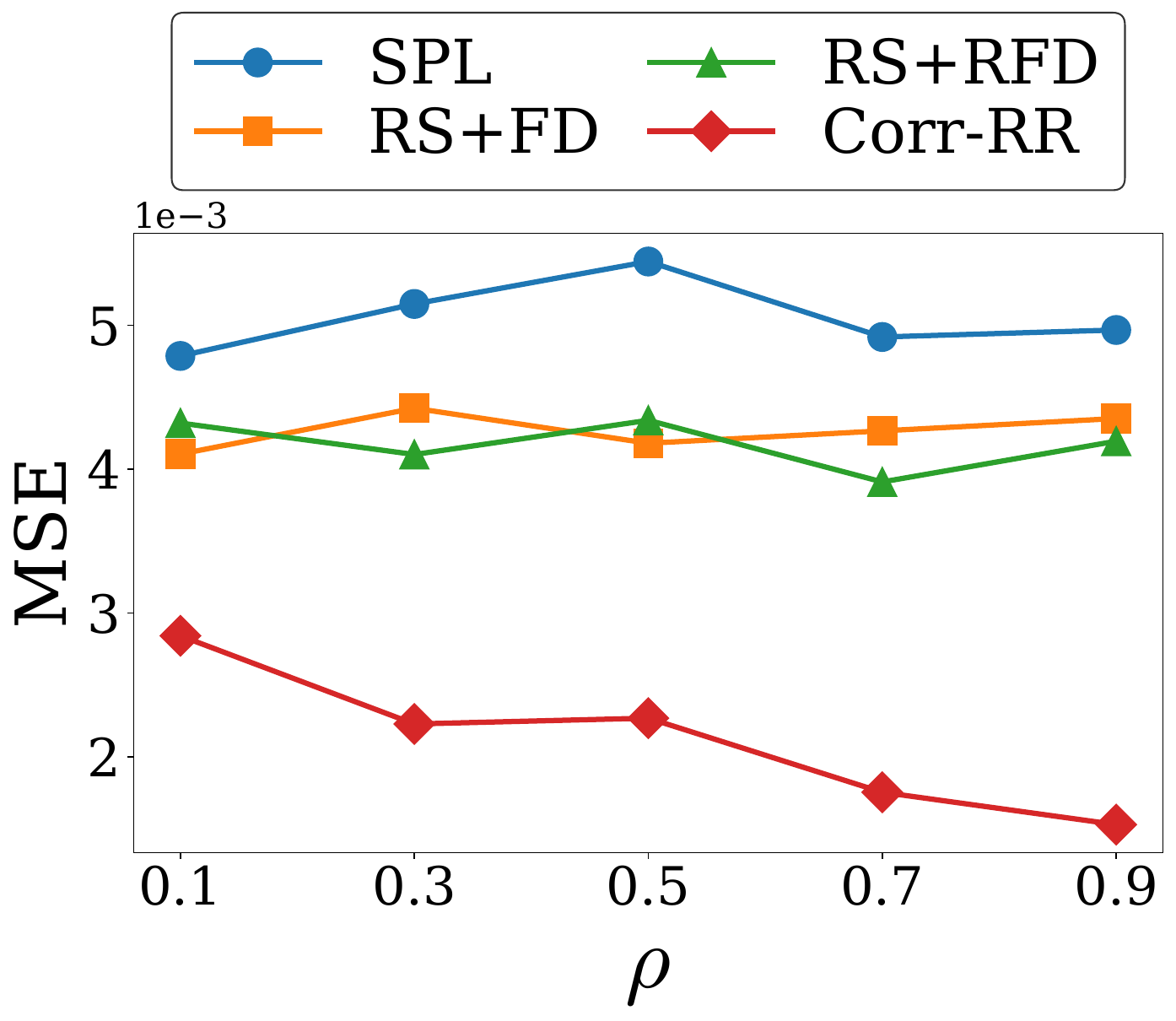} 
    \subcaption{$d= 3$}
    \label{fig:correlation_Effect_d3_prog}
  \end{minipage}\hfill
  \begin{minipage}{.32\textwidth}
      \includegraphics[width=\linewidth]{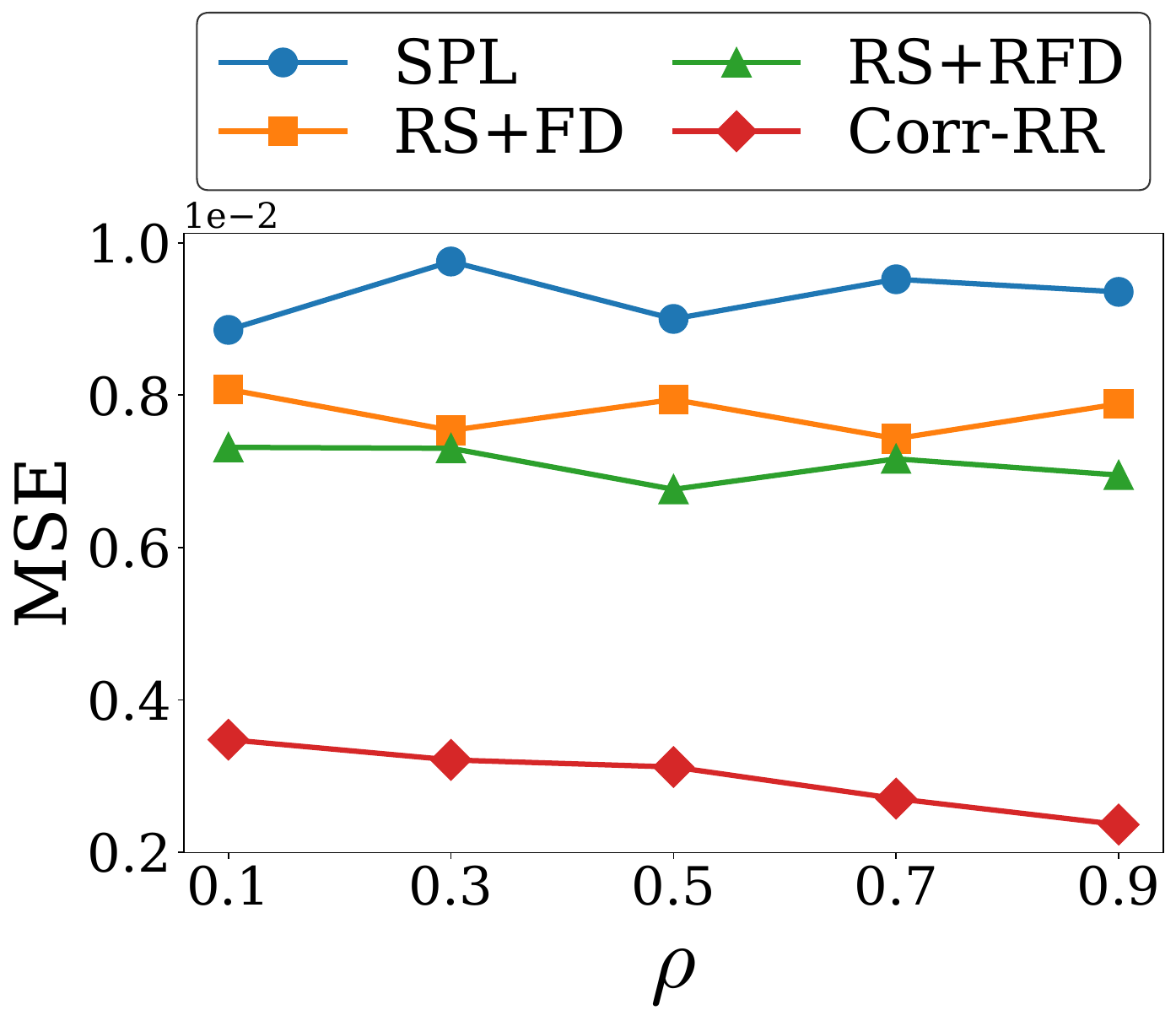}   
    \subcaption{$d= 4$}
    \label{fig:correlation_Effect_d4_prog}
  \end{minipage}\hfill

  \caption{
MSE vs.\ correlation $\rho$ on SynB with $|\mathcal{D}|=4$ and $\epsilon=0.5$. 
Subplots correspond to different numbers of attributes $d$. 
}
\Description{Line plots showing mean squared error as a function of correlation strength rho for multiple numbers of attributes, comparing SPL, RS+FD, RS+RFD, and Corr-RR.}

  \label{fig:correlation_Effect_prog}
\end{figure*}

\begin{figure*}[ht!]
  \centering

  \begin{minipage}{.32\textwidth}
   
    \includegraphics[width=\linewidth]{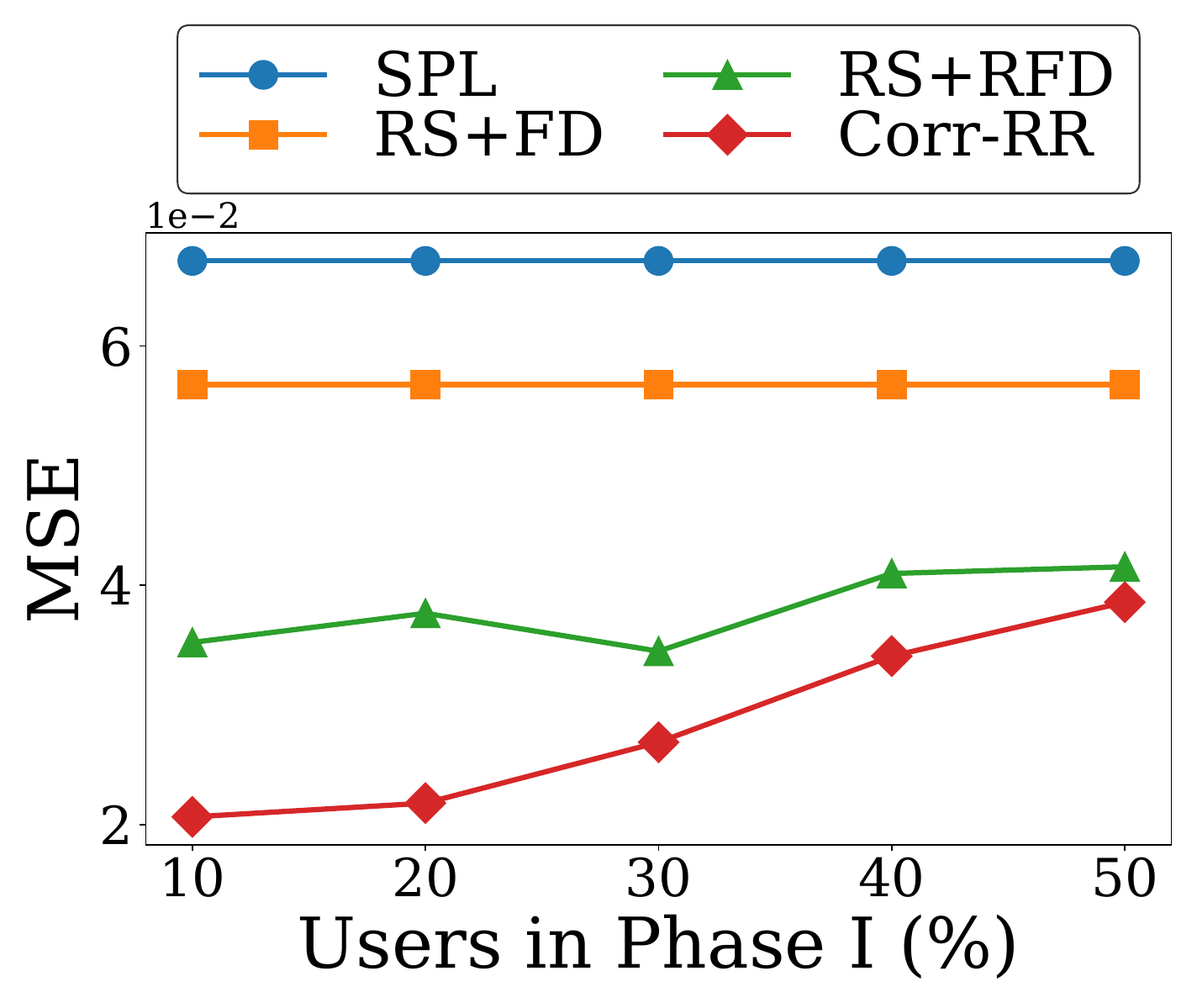} 
    \subcaption{$\epsilon = 0.1$}
    \label{fig:mse_vs_phase_fraction_n_20k_eps01}
  \end{minipage}\hfill
  \begin{minipage}{.32\textwidth}
  \includegraphics[width=\linewidth]{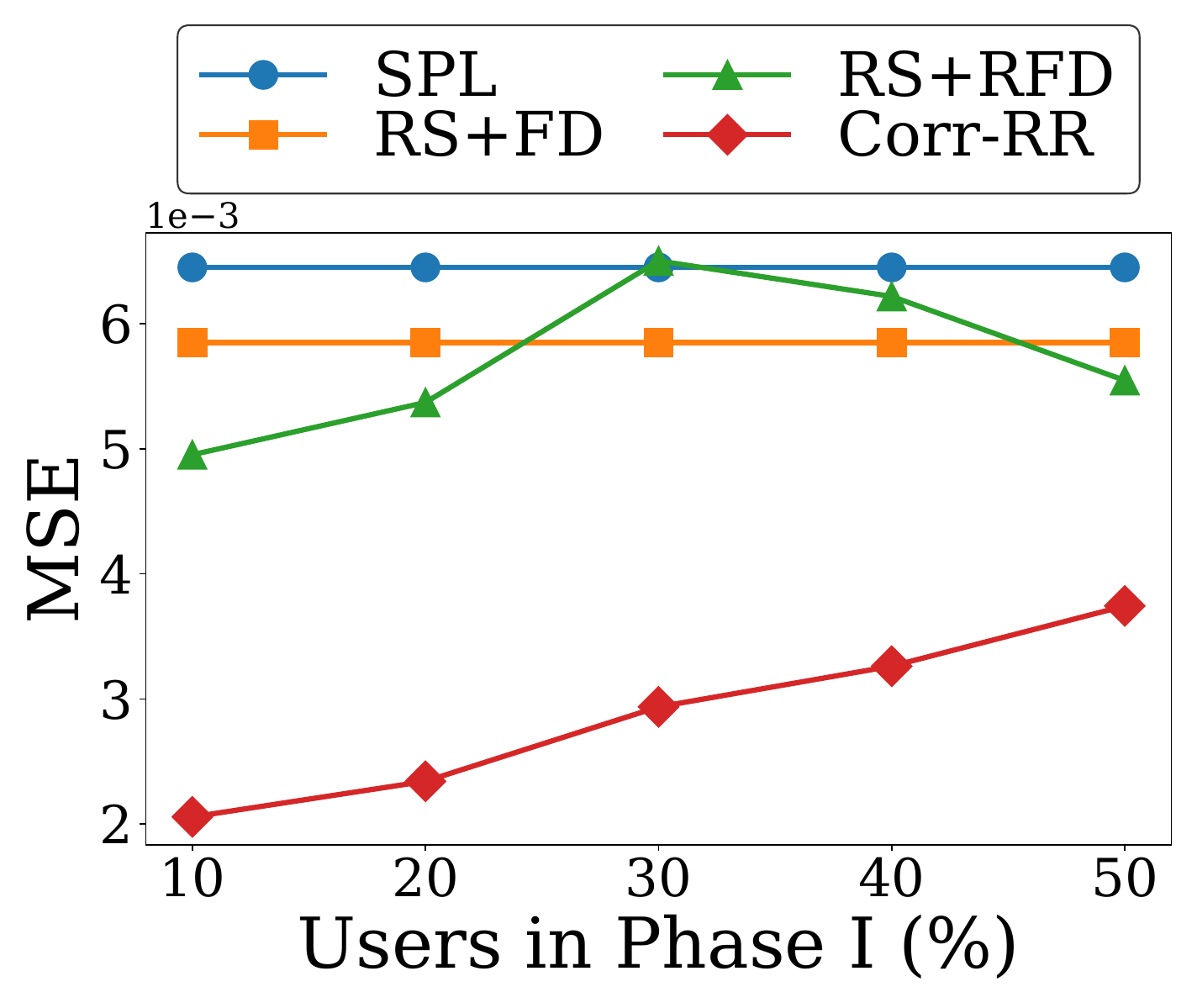}   
    \subcaption{$\epsilon = 0.3$}
    \label{fig:mse_vs_phase_fraction_n_20k_eps03}
  \end{minipage}\hfill
  \begin{minipage}{.32\textwidth}
\includegraphics[width=\linewidth]{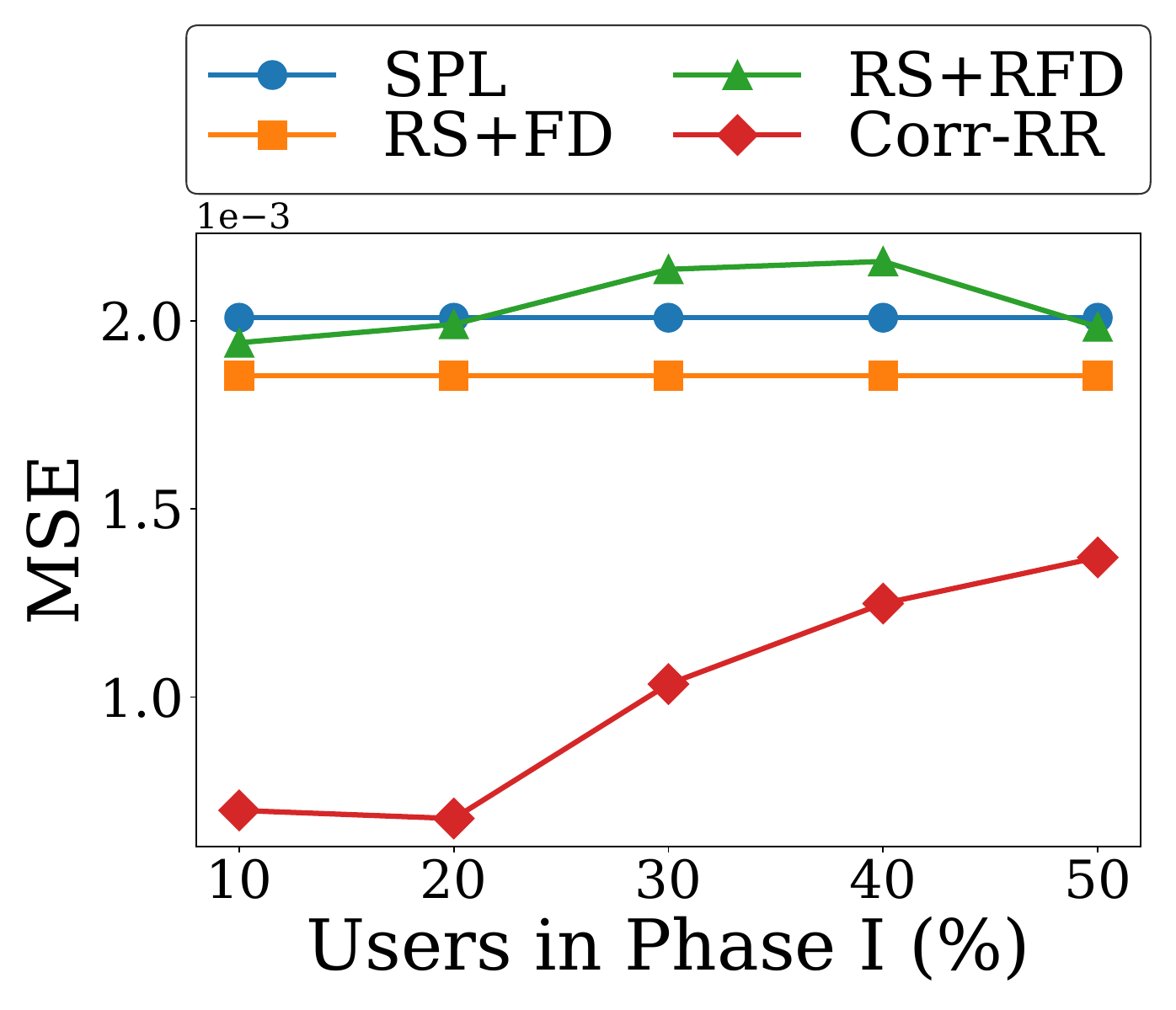}     
    \subcaption{$\epsilon = 0.5$}
    \label{fig:mse_vs_phase_fraction_n_20k_eps05}
  \end{minipage}\hfill

\caption{
MSE vs.\ percentage of Phase~I users ($n_1/n$) on SynA with $d=2$, $|\mathcal{D}|=4$, and $n=20{,}000$. 
Subplots correspond to $\epsilon=0.1$, $0.3$, and $0.5$. 
SPL and RS+FD (single-phase) remain constant across $n_1$, while RS+RFD and Corr-RR (two-phase) vary.
}
\Description{Three line plots showing mean squared error as a function of the fraction of Phase I users n1 over n on the SynA dataset. Each subplot corresponds to a different privacy budget epsilon. Single-phase methods remain constant across n1, while two-phase methods vary with the Phase I allocation.}

  \label{fig:mse_vs_phase_fraction_n_20k}
\end{figure*}

\begin{figure*}[ht!]

  \centering

  \begin{minipage}{.32\textwidth}
   
    \includegraphics[width=\linewidth]{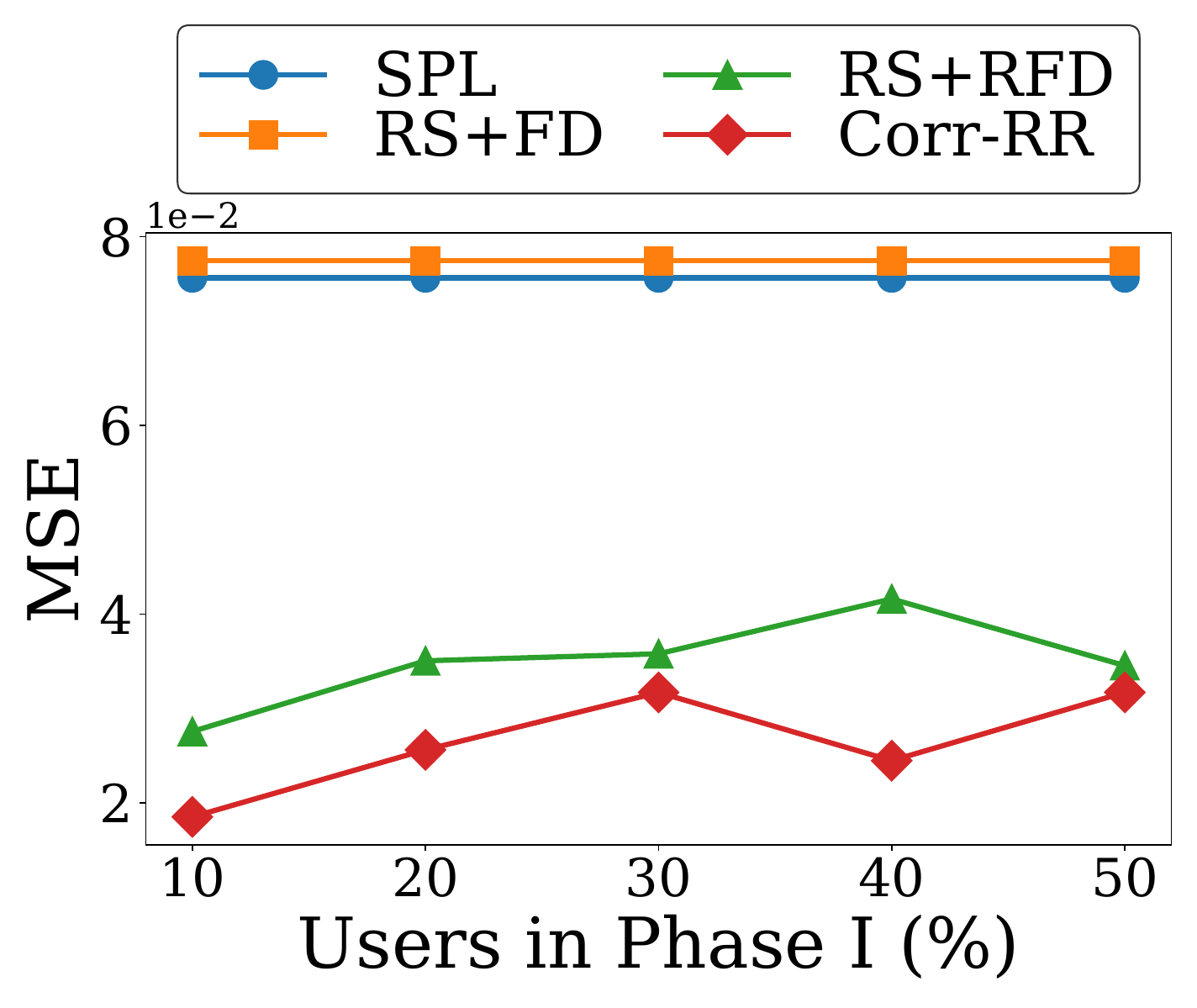} 
    \subcaption{$\epsilon = 0.1$}
    \label{fig:mse_vs_phase_fraction_n_20k_eps01_prog}
  \end{minipage}\hfill
  \begin{minipage}{.32\textwidth}
   \includegraphics[width=\linewidth]{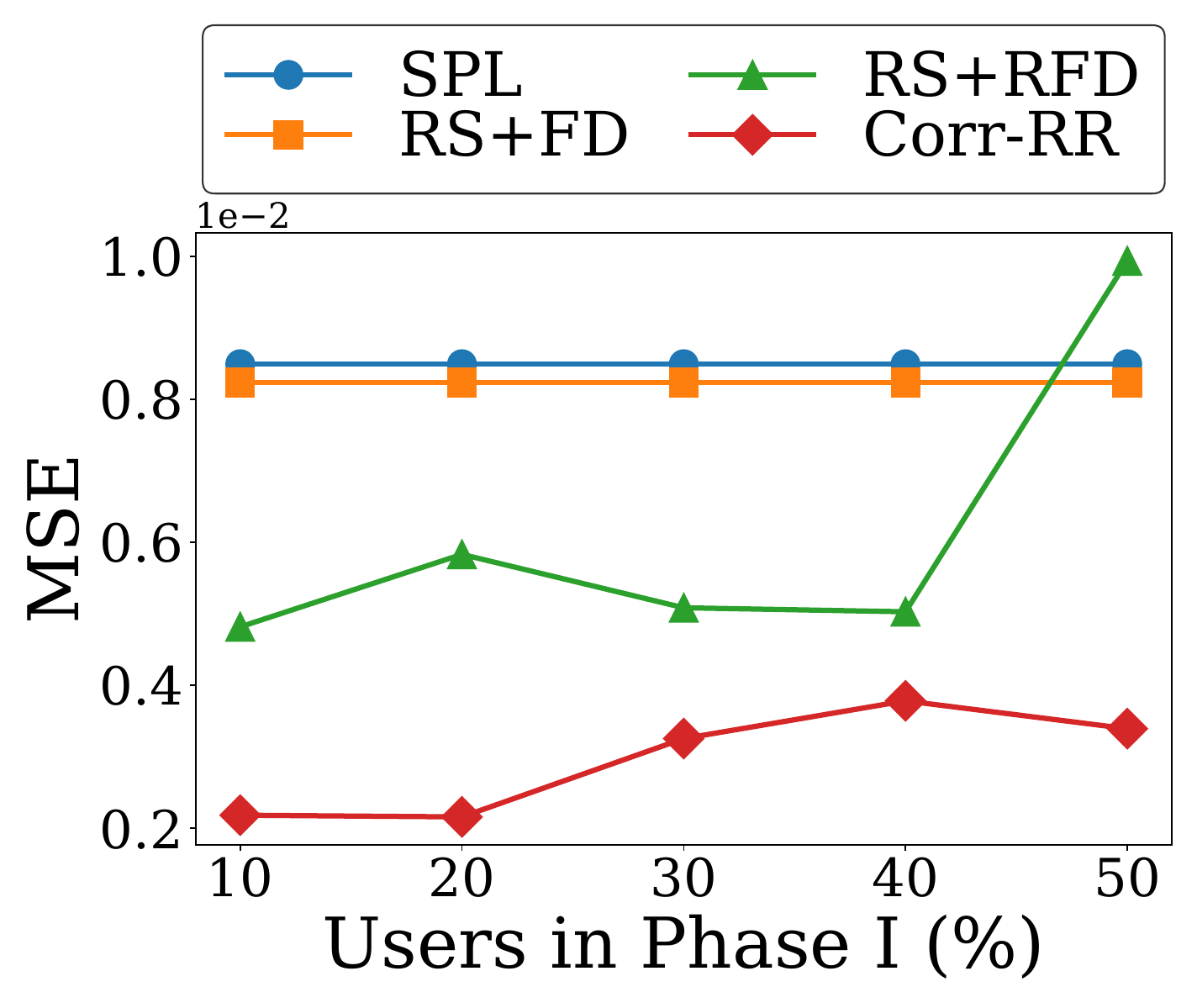} 
    \subcaption{$\epsilon = 0.3$}
    \label{fig:mse_vs_phase_fraction_n_20k_eps03_prog}
  \end{minipage}\hfill
  \begin{minipage}{.32\textwidth}
  \includegraphics[width=\linewidth]{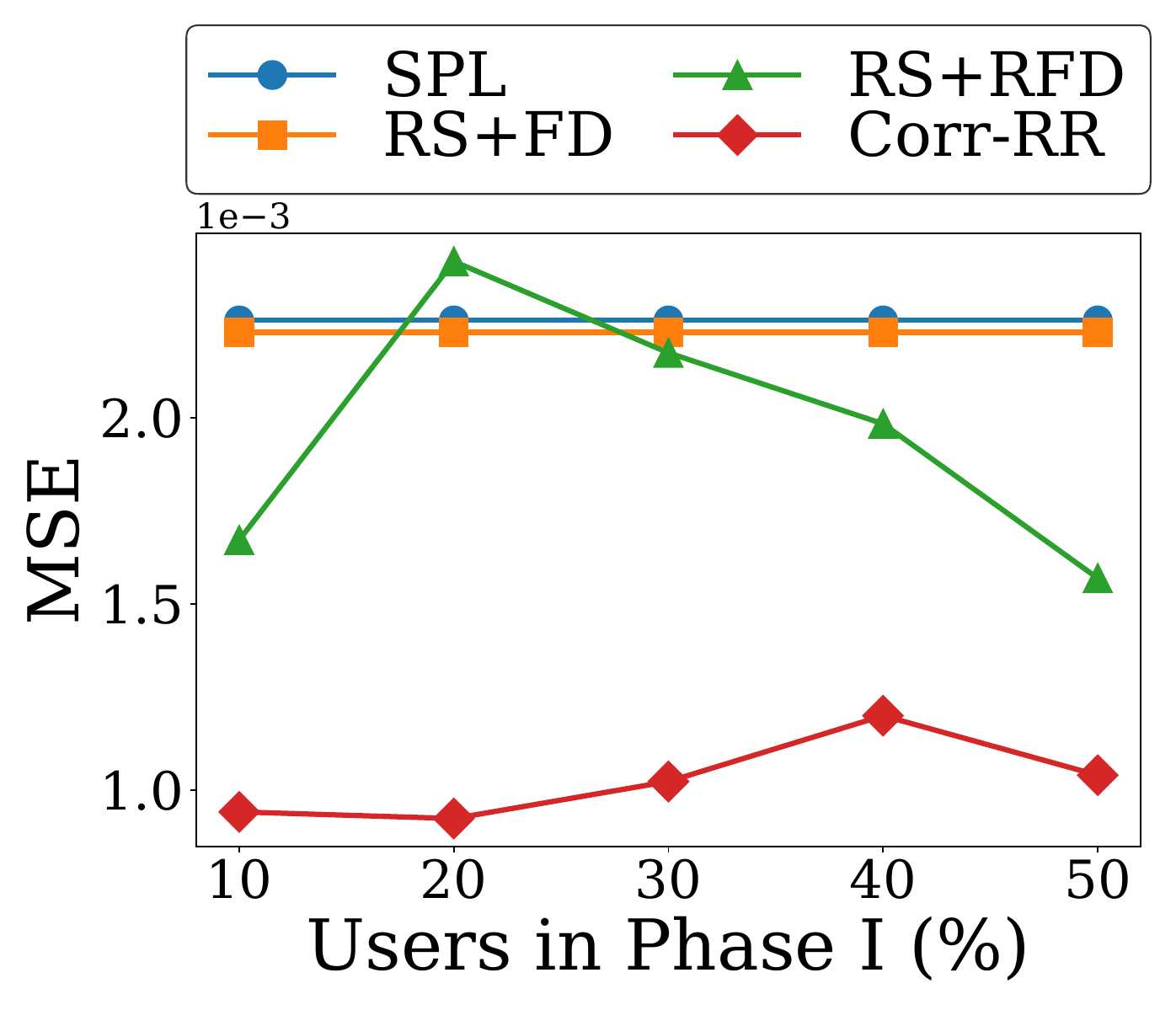}    
    \subcaption{$\epsilon = 0.5$}
    \label{fig:mse_vs_phase_fraction_n_20k_eps05_prog}
  \end{minipage}\hfill

\caption{
MSE vs.\ percentage of Phase~I users ($n_1/n$) on SynB with $d=2$, $|\mathcal{D}|=4$, and $n=20{,}000$. 
Subplots correspond to $\epsilon=0.1$, $0.3$, and $0.5$. 
SPL and RS+FD (single-phase) remain constant across $n_1$, while RS+RFD and Corr-RR (two-phase) vary.
}
\Description{Three line plots showing mean squared error as a function of the fraction of Phase I users n1 over n on the SynB dataset. Each subplot corresponds to a different privacy budget epsilon. Single-phase methods remain constant across n1, while two-phase methods vary with the Phase I allocation.}
  \label{fig:mse_vs_phase_fraction_n_20k_prog}
\end{figure*}

We report the impact of correlation strength $\rho$ on the accuracy of four LDP mechanisms, SPL, RS+FD, RS+RFD, and Corr-RR, on SynA (Figure~\ref{fig:correlation_Effect}) and SynB (Figure~\ref{fig:correlation_Effect_prog}) with domain size $|\mathcal{D}|=4$, $n=20{,}000$, and privacy budget $\epsilon=0.5$. Each subplot corresponds to a different number of attributes $d \in \{2,3,4\}$.

Across all settings, the baseline mechanisms show almost flat curves, with the slight fluctuations arising from the limited number of trials. This confirms that these methods cannot exploit statistical dependencies between attributes: their estimation error is essentially insensitive to correlation. In contrast, Corr-RR demonstrates a clear downward trend in MSE as the correlation increases. This validates the design principle: when attributes are more strongly related, Corr-RR reuses privatized pivot values more effectively, thereby suppressing noise in non-pivot attributes. For example, at $d=4$ and high correlation, Corr-RR achieves more than a 60\% reduction in error relative to the strongest baseline, and even at weak correlation, it provides around a 50\% reduction, demonstrating its robustness in high-dimensional settings. However, Corr-RR is not universally superior. In very low-correlation, low-dimensional scenarios (e.g., $d=2$ and $\rho=0.1$), Corr-RR performs worse than RS+RFD. This arises from how Corr-RR estimates the reuse probability $p_y$. Phase~I perturbs all attributes with budget $\epsilon/d$, which introduces considerable noise at low $\epsilon$. When many values are randomized, the estimated marginals become unreliable, weakening the correlation signal that guides the optimizer. As a result, $p_y$ can be misestimated, leading to underuse or overuse of pivot reuse and inflated error relative to correlation-agnostic baselines. This effect intensifies as $d$ grows, because larger dimensionality reduces the per-attribute Phase~I budget. Once correlations reach moderate strength ($\rho \geq 0.3$), the true dependency structure dominates the Phase~I noise. In this regime, Corr-RR reliably selects effective reuse probabilities and consistently reduces error by substantial margins compared to all baseline mechanisms. Overall, these results underscore two insights: correlation influences utility, and Corr-RR is uniquely capable of leveraging stronger dependencies into significant reductions in error. Its slight disadvantage in the low-correlation regime clarifies that Corr-RR is most impactful in realistic settings where attributes exhibit genuine dependencies and datasets contain multiple attributes.

\subsubsection{Impact of the Size of Phase I Users}\label{split_analysis}

\begin{figure*}[h!]
  \centering

  \begin{minipage}{.32\textwidth}
    
    \includegraphics[width=\linewidth]{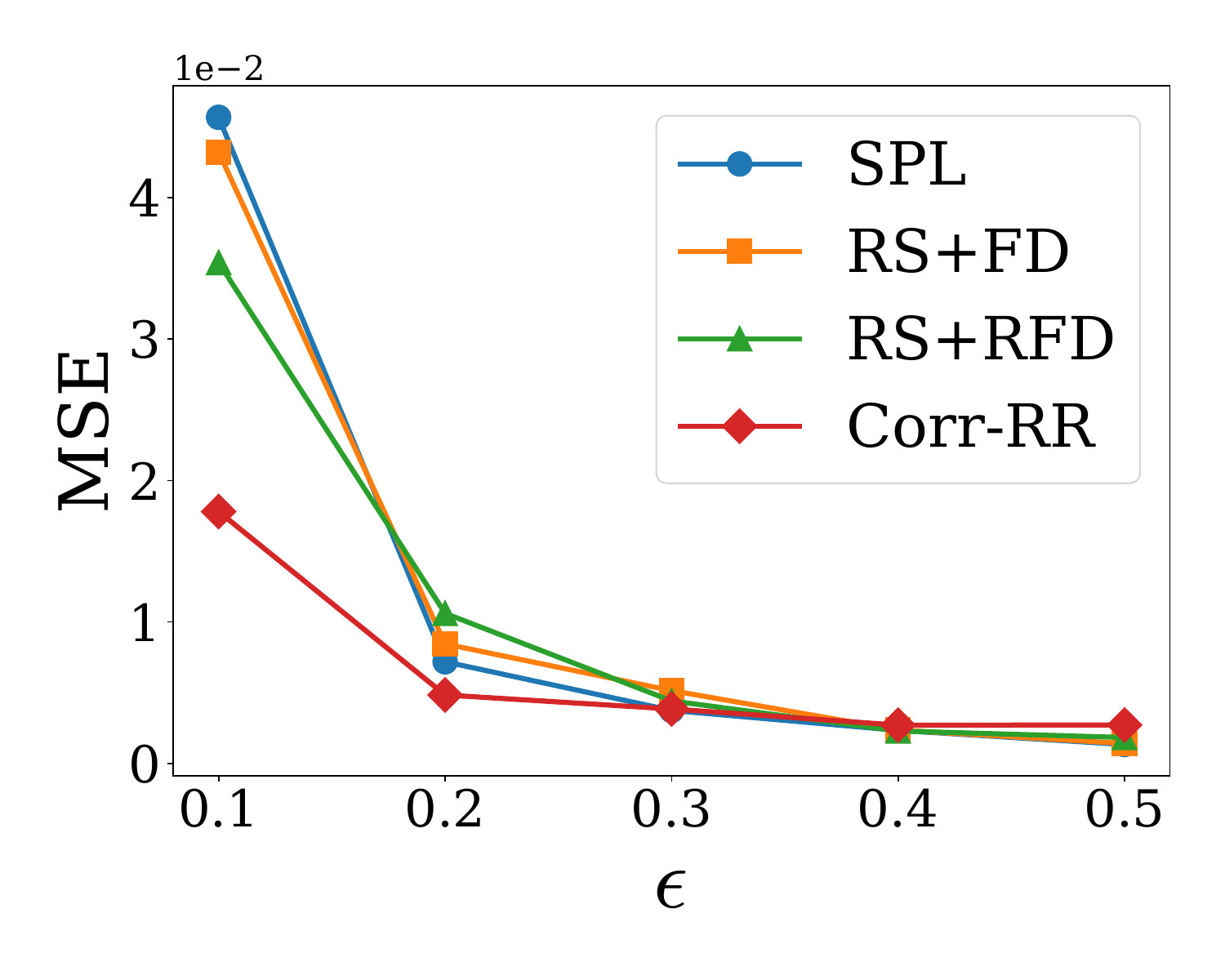}

    \subcaption{Clave}
    \label{fig:mse_eps_clave}
  \end{minipage}\hfill
  \begin{minipage}{.32\textwidth}
   \includegraphics[width=\linewidth]{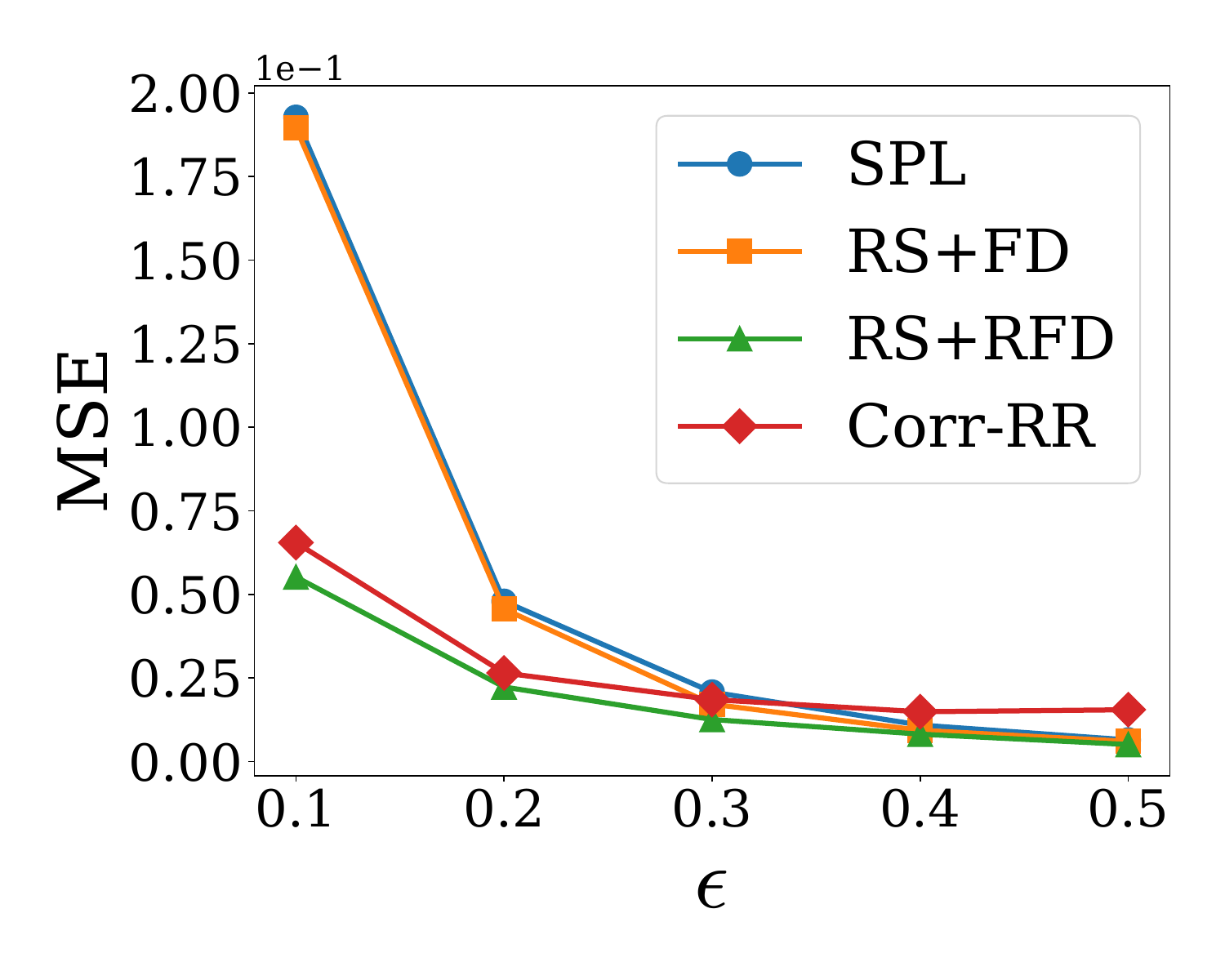}
        \subcaption{Mushroom}
    \label{fig:mse_eps_mushroom}
  \end{minipage}\hfill
  \begin{minipage}{.32\textwidth}
   \includegraphics[width=\linewidth]{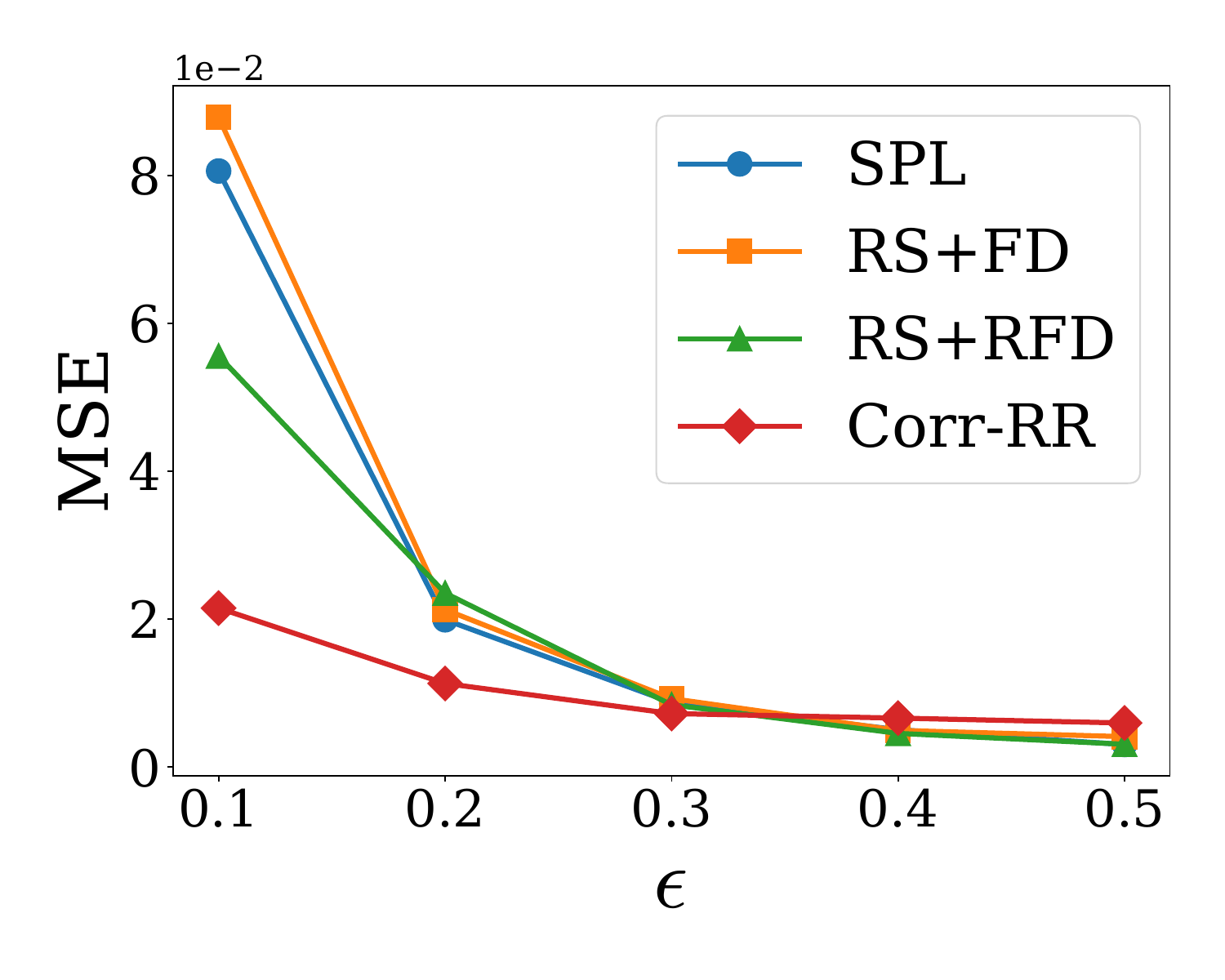}
       
    \subcaption{Adult}
    \label{fig:mse_eps_adult}
  \end{minipage}\hfill

\caption{
MSE vs.\ privacy budget $\epsilon$ for four LDP mechanisms on three real-world datasets: 
(\subref{fig:mse_eps_clave}) \textit{Clave}, 
(\subref{fig:mse_eps_mushroom}) \textit{Mushroom}, 
(\subref{fig:mse_eps_adult}) \textit{Adult}.
}
\Description{Line plots comparing mean squared error across different privacy budgets for multiple mechanisms on a real-world dataset.}

  \label{fig:mse_eps_realworld}
\end{figure*}

In Figures~\ref{fig:mse_vs_phase_fraction_n_20k} and~\ref{fig:mse_vs_phase_fraction_n_20k_prog}, we evaluate how the fraction of Phase~I users ($n_1/n$) affects the estimation accuracy in the two-phase framework on the synthetic datasets \textit{SynA} and \textit{SynB}. Each experiment uses $n = 20{,}000$ users with $d = 2$ attributes and domain size $|\mathcal{D}| = 4$, under privacy budgets $\epsilon \in \{0.1, 0.3, 0.5\}$.

As expected, the single-phase baselines (SPL and RS+FD) remain constant across different $n_1/n$, since they do not rely on a Phase~I/Phase~II split. In contrast, both the two-phase implementation of RS+RFD and our Corr-RR mechanism exhibit varying performance with the fraction of Phase~I users. Corr-RR consistently achieves the lowest MSE in the large-scale setting ($n=20{,}000$) for all privacy budgets and its accuracy is highest when only a small subset of users (approximately 10–20\%) participates in Phase~I. This aligns with the design intuition: Phase~I requires only enough users to obtain coarse but reliable marginal estimates, while allocating the majority of users to Phase~II maximizes the value of full-budget perturbation. At $\epsilon=0.1$, Corr-RR reduces error by more than 40\% relative to RS+RFD and achieves more than a $3\times$ improvement over SPL; this advantage persists for $\epsilon=0.3$ and $\epsilon=0.5$, demonstrating robustness across privacy regimes.

For a more fine-grained analysis, we also include separate tables for RS+RFD and Corr-RR showing detailed results for Phase~I subset sizes ranging from 5\% to 50\% of the users (in 5\% increments), provided in Tables~\ref{tab:phase1_n_200_star}–\ref{tab:phase1_n_20000_prog} in Appendix~\ref{Appendix2}. These results reveal clear and consistent trends across dataset sizes. When the population is large ($n = 20{,}000$; Table~\ref{tab:phase1_n_20000_star} and Table~\ref{tab:phase1_n_20000_prog}), Corr-RR is uniformly the best-performing mechanism for all privacy budgets: its lowest MSE always occurs with only 5–10\% Phase~I users (1,000–2,000 users), and additional Phase~I users do not yield further improvements. With this many users, the privatized marginals produced in Phase~I are already accurate enough for Corr-RR to reliably estimate the dependency parameter $p_y$. At intermediate scale ($n = 2{,}000$; Table~\ref{tab:phase1_n_2000_star} and Tables~\ref{tab:phase1_n_2000_prog}), Corr-RR again performs strongly when the privacy budget is moderate or high ($\epsilon \ge 0.3$), achieving its lowest error also with 5–10\% Phase~I users (100–200 users). Under tight privacy ($\epsilon = 0.1$), however, even the largest Phase~I setting (50\%, i.e., 1,000 users) cannot overcome the GRR noise, and Corr-RR does not outperform RS+RFD. At small scale ($n = 200$; Table~\ref{tab:phase1_n_200_star} and Table~\ref{tab:phase1_n_200_prog}), even when Phase~I includes 50\% of users (i.e., 100 users), the privatized marginals remain too noisy under $\epsilon \in \{0.1, 0.3\}$ for Corr-RR to estimate $p_y$ reliably; only at $\epsilon = 0.5$ does the noise level become low enough in Phase~I for a reliable estimate of $p_y$ for Corr-RR to become competitive. These results collectively demonstrate that the utility of Corr-RR is shaped by the interaction between $n$, $n_1$, and $\epsilon$.

\textit{Guidance on selecting subset size.} The empirical results above provide practical guidance on choosing the Phase~I fraction. In large-scale deployments ($n \ge 20{,}000$), a modest Phase~I subset of 5--10\% consistently provides stable marginal estimates and achieves the best MSE across privacy budgets. At intermediate scales ($n \approx 2{,}000$), this same 5--10\% range remains effective whenever the privacy budget is moderate or high ($\epsilon \ge 0.3$). At very small scales ($n = 200$), the reliability of Phase~I marginals becomes highly sensitive to $\epsilon$, and Corr-RR is effective only when the privacy budget is sufficiently large (e.g., $\epsilon = 0.5$). These guidelines are derived from experiments with $d = 2$ attributes and domain size $|\mathcal{D}| = 4$. For larger attribute dimensionality $d$, the Phase~I privacy budget must be divided across more attributes, increasing noise; similarly, larger domain sizes require more Phase~I users to achieve reliable marginal estimates. Larger domains or higher dimensionality would require proportionally larger Phase~I subsets to obtain sufficiently accurate marginal estimates to determine $p_y$ effectively. Nevertheless, 5--10\% serves as a robust default for large-$n$ or moderate-$\epsilon$ settings when $d=2$ and $|\mathcal{D}|=4$.

\subsection{Results on Real-world Data}

Figures~\ref{fig:mse_eps_clave} to \ref{fig:mse_eps_adult} show the MSE of four LDP mechanisms, SPL, RS+FD, RS+RFD, and Corr-RR, on three real-world datasets, Clave, Mushroom, and Adult, as the privacy budget~$\epsilon$ varies.
Across all datasets, MSE decreases monotonically as the privacy budget $\epsilon$ increases, since a larger $\epsilon$ allows users to report their true values with higher probability, thereby improving data utility and reducing MSE. 

Corr-RR achieves the lowest error on both \textit{Clave} and \textit{Adult}, outperforming baselines, particularly at smaller~$\epsilon$. On the \textit{Mushroom} dataset, however, Corr-RR shows comparatively weaker performance. 
Although it still outperforms SPL and RS+FD under tighter privacy budgets ($\epsilon \leq 0.2$), it falls behind RS+RFD with a slightly larger MSE. This behavior can be attributed to two key factors.
\textit{(1)~Low correlation:} The \textit{Mushroom} dataset shows the weakest inter-attribute correlation ($0.209$) compared to \textit{Clave} ($0.298$) and \textit{Adult} ($0.43$–$0.95$). 
\textit{(2)~Marginal imbalance:} As visualized in Figure~\ref{fig:relative_freq_realdata} in Appendix~\ref{Appendix2}, compared to the \textit{Clave} and \textit{Adult},  the marginal distributions of attributes in \textit{Mushroom} are highly skewed (e.g., X$_1$ is right-skewed, X$_2$ left-skewed) and dominated by one or two frequent category values. Recall that in Corr-RR Phase~II, $p_y$ is determined by minimizing the average MSE rather than MSE of the estimator for each value, making the frequency estimation in Phase~II less accurate for the highly skewed attributes and resulting in a slightly higher MSE. 
Overall, these results suggest that Corr-RR provides the largest gains when Phase~I can capture meaningful dependencies and Phase~II can reuse them without being dominated by a few frequent categories. 
In datasets with only modest correlation and substantial skew, the average-MSE objective can be driven by dominant values, reducing the benefit of dependency-guided synthesis for rarer categories.


\subsection{Summary of Findings}

Our evaluation yields the following key insights:
\begin{itemize}
\item \textbf{Improved utility.} Corr-RR consistently outperforms all baselines under low to moderate privacy budgets, reducing MSE by up to 80\% in some settings. As privacy loosens, the performance gap narrows—since all mechanisms benefit from reduced noise—but Corr-RR retains a steady advantage across the privacy budget $\epsilon$.

\item \textbf{Robust scalability.} Corr-RR alleviates the severe accuracy degradation typically observed as the number of attributes grows. The improvement becomes more striking as dimensionality increases, confirming that Corr-RR is well-suited for modern high-dimensional datasets.

\item \textbf{Correlation-driven gains.} Corr-RR’s advantage grows with stronger inter-attribute dependencies. Even modest correlations lead to measurable improvements under tight privacy, while strong correlations yield the largest reductions in error. This validates Corr-RR’s core principle: leveraging statistical structure to suppress noise and recover more accurate estimates. 
\end{itemize}
These findings establish Corr-RR as a practical and scalable solution for multi-attribute data under LDP.

\section{Related Work}\label{sec:related}

\subsection{Differential Privacy for Correlated Data.} Differential Privacy (DP) protects individual privacy during data analysis by ensuring that query outputs minimally reveal specific individual data, often by integrating noise into the results~\cite{dwork2006calibrating, dwork2006our, TCS-042}. However, the effectiveness of standard DP is challenged by correlated data, which introduces privacy vulnerabilities due to data interdependencies~\cite{kifer2011no}. To address these issues, the Pufferfish privacy model enhances the DP framework to better handle such correlations~\cite{kifer2014pufferfish}. Further innovations, such as the Wasserstein and Markov Quilt Mechanisms within the Pufferfish framework \cite{song2017pufferfish}, and the Blowfish Framework \cite{he2014blowfish}, have advanced DP's application to correlated data contexts. Nonetheless, these advancements are mainly applicable in centralized settings involving a trusted data collector and have not yet been adapted for use in local settings with untrusted collectors.

\subsection{Local Differential Privacy} 
 Recent advancements have seen a notable shift toward Local Differential Privacy (LDP)~\cite{duchi2013local,duchi2014privacy, zheng2025optimal}, a model that grants individuals complete control over their data, eliminating the need to trust the aggregator. Widely adopted in both academic research and industrial applications~\cite{erlingsson2014rappor, team2017learning, wang2019collecting, ding2017collecting, zheng2025locally}, LDP is intrinsically linked to Randomized Response techniques~\cite{warner1965randomized}. Among many other complex tasks (e.g., heavy hitter estimation~\cite{bun2019heavy, qin2016heavy}, estimating marginals~\cite{ren2018textsf, erlingsson2014rappor, peng2019privacy}, time-series data~\cite{wang2020towards}, frequent itemset mining~\cite{wang2018locally, li2012privbasis}, key-value pair analysis~\cite{ye2019privkv,gu2020pckv}), frequency estimation is a fundamental task in LDP and has received considerable attention for a single attribute. A prominent implementation of LDP is Google’s RAPPOR~\cite{erlingsson2014rappor}, which has been successfully integrated into the Chrome browser. RAPPOR is distinguished by its dual-layer defense against windowed attacks and its use of Bloom filters~\cite{bloom1970space}. Additionally, techniques such as Unary Encoding (OUE), Optimal Local Hashing (OLH), and Hadamard Response have been developed to further optimize utility within this framework~\cite{acharya2019hadamard, wang2017locally}. Note that all of the above methods focus on LDP on a single attribute.

\subsection{LDP for Multi-attribute Data}
While most of the works on multi-attribute data focused on numerical data~\cite{nguyen2016collecting, wang2019collecting, wang2021local}, frequency estimation on multi-attribute data is less explored. This is due to the constraints imposed by the composition theorem, as the budget rapidly depletes for multi-attribute datasets. To mitigate the curse of dimensionality, the LoPub algorithm leverages attribute correlation~\cite{ren2018textsf}. Domingo-Ferrer and Soria-Comas proposed a method that groups correlated attributes based on dependencies and categorical combinations, and then applies Randomized Response (RR) collectively to each cluster, improving the accuracy of the estimation ~\cite{domingo2020multi}. Arcolezi et al. have proposed RS+FD and RS+RFD that generate fake data for unsampled attributes, uniformly and nonuniformly, respectively, while the selected attribute receives the full privacy budget~\cite{arcolezi2021random, arcolezi2022risks}. Additionally, a recent study by Du et al. 
introduced a correlation-bounded perturbation mechanism that
quantifies and utilizes inter-attribute correlations and optimizes
partitioning of the privacy budget for each attribute in multi-attribute scenarios~\cite{du2021collecting}. Our approach distinctly diverges from existing methods by leveraging correlations to indirectly perturb attributes, which significantly improves the accuracy of frequency estimation.

\section{Conclusion}\label{sec:conclusion}

We addressed the challenge of multi-attribute frequency estimation under Local Differential Privacy (LDP), where existing solutions either dilute the privacy budget across attributes or rely on synthetic imputations that ignore inter-attribute dependencies. To overcome these limitations, we introduced a general two-phase framework that separates \emph{dependency learning} from \emph{dependency-aware reporting}, enabling correlation-guided reconstructions derived entirely from privatized data. Under the two-phase framework, we developed \textbf{Correlated Randomized Response (Corr-RR)}, which perturbs a single pivot attribute with the full privacy budget and synthesizes the remaining attributes through probabilistic mappings informed by Phase~I statistics. Corr-RR rigorously satisfies $\epsilon$-LDP, and its parameters are chosen to minimize a closed-form approximation of mean squared error, ensuring accuracy while preserving strong privacy guarantees. Under our two-phase framework, we also briefly discuss Conditional Randomized Response (Cond-RR), which extends the idea of correlation-aware synthesis by
using conditional distributions rather than pairwise probabilities. Comprehensive experiments on synthetic and real-world datasets show that Corr-RR consistently improves estimation accuracy over state-of-the-art baselines, with the largest gains observed in high-dimensional and strongly correlated settings. These results highlight the effectiveness of exploiting inter-attribute correlations under LDP and demonstrate the practical viability of our two-phase approach for privacy-preserving data collection.


\begin{acks}
We thank the reviewers and the revision editor for their insightful feedback and guidance, which greatly strengthened this paper. We used GPT-4 as a writing assistance tool for language refinement. This work was partially supported by the U.S. National Science Foundation under grants CNS-2245689 (CRII) and Meta Research Award in Privacy-Enhancing Technologies.
\end{acks}

\bibliographystyle{ACM-Reference-Format}

\bibliography{main}

\section{Appendix }

\subsection{Proof of Theorem~\ref{thm:MSE_categorical_general}}\label{appendix:thm:mse} 
\begin{proof}
Let \(n_2\) be the number of users participated in Phase II and \(d=2\) the number of attributes. Each attribute has categorical value with domain $\{1,2,\cdots,k\}$.
Define GRR with parameters
\[
p=\frac{e^\epsilon}{e^\epsilon+k-1},\qquad
q=\frac{1}{e^\epsilon+k-1},\qquad
\Delta=p-q=\frac{e^\epsilon-1}{e^\epsilon+k-1}.
\]
Also let $s$ be the attribute selected for perturbation using GRR, and $Z$ an indicator random variable of the selected attribute is $j$-th targeted attribute, i.e., $Z=1$ indicates $s=j$, otherwise $Z=0$ indicates $s\neq j$. 

We first calculate probability of an specific attribute \(j\) with value \(v\in\{1,\dots,k\}\) by
 \begin{equation}\label{eq:pra=1}
\begin{split}
        \Pr(Y_j=v)& = \Pr(Z=1)\cdot \Pr(Y_j=v|Z=1)\\ &\quad +\Pr(Z=0)\cdot \Pr(Y_j=v|Z=0) \\
        &=\frac{1}{2} \Pr(Y_j=v|Z=1)+\frac{1}{2} \Pr(Y_j=v|Z=0)
\end{split}
\end{equation}

If $Z=1$, the attribute $j$ is selected for perturbation using GRR, we have 
\[
\Pr(Y_j=v|Z=1)=(pf_j(v)+q(1-f_i(v))=q+\Delta f_j(v)\;,
\]
where $f_j(v)$ is the true frequency of attribute $j$ with value $v$. 

If $Z=0$, the other attribute is selected, i.e., $s\neq i$, the user first perturbs $x_s$ with GRR to get $Y_s=u$; then report $Y_j$ by randomly perturbing $y_s$, described by
\[
\Pr(Y_i=v \mid Y_s=u)=
\begin{cases}
p_y, & u=v,\\[4pt]
\dfrac{1-p_y}{k-1}, & u\neq v,
\end{cases}
\qquad p_y\in[0,1].
\]
It follows that
\begin{equation}\label{pr: y1=v,z1=0}
\begin{split}
&\Pr(Y_j=v|Z=0)\\
&=\Pr(Y_j=v|Z=0,Y_s=v)\Pr(Y_s=v|Z=0)\\
&\quad+\Pr(Y_j=v|Z=0,Y_s\neq v)\Pr(Y_s\neq v|Z=0)\\
\end{split}
\end{equation}

Note that, we have
\[
\begin{split}
&\Pr(Y_j=v|Z=0,Y_s=v)=p_y\\
&\Pr(Y_s=v|Z=0)=pf_s(v)+q(1-f_s(v))=q+\Delta f_s(v)\\
&\Pr(Y_s\neq v|Z=0)=1-\Pr(Y_s=v|Z=0)=p-\Delta f_s(v)\;,\\
\end{split}
\]
and
\[
\begin{split}
    &\Pr(Y_j=v|Z=0,Y_s\neq v)\Pr(Y_s\neq v|Z=0)\\
    &= \sum_{u \neq v} \Pr(Y_j=v|Z=0,Y_s=u)\Pr(Y_s=u|Z=0)\\
    &=\frac{1-p_y}{k-1} \sum_{u \neq v} \Pr(Y_s=u|Z=0)\\
    &=\frac{1-p_y}{k-1} (1-\Pr(Y_s=v|Z=0))
\end{split}
\]
According to Eq.~(\ref{pr: y1=v,z1=0}),
we have 
\[
\begin{split}
    \Pr(Y_j=v|Z=0)&=p_y(q+\Delta f_s(v))+\frac{1-p_y}{k-1}(p-\Delta f_s(v))\\
\end{split}
\]
Substitute the above $\Pr(Y_j=v|Z=0)$ and $\Pr(Y_j=v|Z=1)$ into Eq.~(\ref{eq:pra=1}), we have
\begin{equation}\label{eq:pi-v}
    \Pr(Y_j=v)=\frac{1}{2}[q+\Delta f_j(v)]+\frac{1}{2}[p_y t+\frac{1-p_y}{k-1}(1-t)]
\end{equation}
where $t=q+\Delta f_s(v)$.
Also let 
\[
a:=p_y-\frac{1-p_y}{k-1}=\frac{k p_y-1}{k-1}, \text{ and } \pi_v:=\Pr(Y_j=v).
\]
Using $1-kq=p-q$, we can rewrite \eqref{eq:pi-v} as
\begin{equation}\label{eq:pi-affine}
\pi_v \;=\; q + (p-q)\,\mu_v, \text{ where }
\mu_v
= \tfrac12\![f_j(v)+\frac{1-p_y}{k-1}+a\,f_s(v)].
\end{equation}

\textbf{Estimator, bias, variance, and MSE.}
Over $n_2$ i.i.d.\ users, let
\[
I_a(v):=\sum_{m=1}^{n_2} \mathbf 1\{Y_j^{(m)}=v\}\sim \mathrm{Binomial}\bigl(n_2,\pi_v\bigr).
\]
Consider the standard GRR debiasing estimator
\[
\widehat f_j(v)=\frac{I_a(v)-nq}{n_2(p-q)}=\frac{I_a(v)-n_2q}{n_2\Delta}.
\]

Then
\[
\mathbb{E}\big[\widehat f_i(v)\big]
= \frac{\pi_v-q}{\Delta}
= \mu_v,
\]
\[
\mathrm{Bias}\big[\widehat f_j(v)\big]
= \mu_v - f_j(v)
= \tfrac12\!\left[\frac{1-p_y}{k-1} + a\,f_s(v) - f_j(v)\right].
\]
Moreover,
\[
\mathrm{Var}\big[\widehat f_j(v)\big]
= \frac{\pi_v\bigl(1-\pi_v\bigr)}{n_2\,\Delta^2}
= \frac{\bigl(q+\Delta\,\mu_v\bigr)\bigl(1-q-\Delta\,\mu_v\bigr)}{n_2\,\Delta^2}.
\]
Hence the mean-squared error is
\[
\boxed{\;
\mathrm{MSE}\big[\widehat f_j(v)\big]
= \frac{\bigl(q+\Delta\,\mu_v\bigr)\bigl(1-q-\Delta\,\mu_v\bigr)}{n_2\,\Delta^2}
\;+\; \Big(\mu_v - f_j(v)\Big)^2,
\;}
\]
with $\mu_v$ given in \eqref{eq:pi-affine} and $\Delta=\dfrac{e^\epsilon-1}{e^\epsilon+k-1}$.
\end{proof}

\subsection{The Proof of Proposition \ref{thm:py_categorical_general}}\label{appendix:thm:py}
\begin{proof}

We start by focusing on the per-value MSE for estimating $f^{II}_j(v)$, which is given by
\begin{equation}\label{eq:MSE-mu}
\mathrm{MSE}\big[\widehat f_j(v)\big]
= \frac{(q+\Delta\mu_v)\bigl(1-q-\Delta\mu_v\bigr)}{n\,\Delta^2}
+ \bigl(\mu_v-f_j(v)\bigr)^2.
\end{equation}

We first rewrite $\mathrm{MSE}\big[\widehat f_j(v)\big]$ as a function of $p_y$. Specifically, 

we expand the first term,
\[
\frac{(q+\Delta\mu_v)(1-q-\Delta\mu_v)}{n_2\Delta^2}
= \frac{q(1-q)}{n_2\Delta^2} + \frac{1-2q}{n_2\Delta}\,\mu_v - \frac{1}{n_2}\,\mu_v^2.
\]
and the second term,
\[
\bigl(\mu_v-f_j(v)\bigr)^2=\mu^2_v-2f_j(v)\mu_v+f_j^2(v)
\]
resulting in the quadratic form
\begin{equation}\label{eq:MSEv-quadratic-in-mu}
\mathrm{MSE}_v(\mu_v) = \alpha\,\mu_v^2 + \beta_v\,\mu_v + \gamma_v.
\end{equation}
where $\alpha=1-\frac{1}{n_2}$, $\beta_v=-2f_j(v)+\frac{1-2q}{n_2\Delta}$, and $\gamma_v=f_j(v)^2+\frac{q(1-q)}{n_2\Delta^2}$.

Also note that from Eq.(~\ref{eq:pi-affine}), we have
\[
\mu_v\;=\; \tfrac12\![f_j(v)+\frac{1-p_y}{k-1}+\frac{k p_y-1}{k-1}\,f_s(v)]
\]
which could be rewritten as $\mu_v \;=\; C_0(v)+C_1(v)\,p_y$,
where
\[
C_0(v)=\tfrac12\!\left[f_j(v)+\frac{1-f_s(v)}{k-1}\right], \text{ and }
C_1(v)=\tfrac12\!\left[\frac{k f_s(v)-1}{k-1}\right].
\]

Next, we rewrite MSE as a quadratic in $p_y$ by substituting $\mu_v=C_0(v)+C_1(v)\,p_y$ into \eqref{eq:MSEv-quadratic-in-mu}:
\[
\mathrm{MSE}_v(p_y)
= \alpha\,[C_0(v)+C_1(v) p_y]^2 + \beta_v\,[C_0(v)+C_1(v) p_y] + \gamma_v,
\]
so the coefficients (per $v$) are
\[
\text{coef}(p_y^2)=\alpha\,C_1(v)^2,\quad
\text{coef}(p_y)=2\alpha\,C_0(v)C_1(v)+\beta_v C_1(v),\quad
\]
\[
\text{const}=\alpha C_0(v)^2+\beta_v C_0(v)+\gamma_v.
\]

\paragraph{Sum over $v$ (global objective).}
Define
\[
A \;:=\; \alpha\sum_{v} C_1(v)^2,\qquad
B \;:=\; \sum_{v}\Bigl[2\alpha\,C_0(v)C_1(v)+\beta_v C_1(v)\Bigr].
\]
Then the total MSE across all categories is
\[
\mathrm{MSE}_{\text{tot}}(p_y) \;=\; A\,p_y^2 + B\,p_y + \text{const}.
\]
For $n>1$, $A\ge 0$ with equality iff $C_1(v)=0$ for all $v$ (i.e., $f_s(v)=1/k$).

\paragraph{Closed-form optimizer and feasibility.}
The unconstrained minimizer is
\begin{equation}\label{eq:py-star}
\begin{split}
    &p_y^* \;=\; -\frac{B}{2A}
\;=\;
-\frac{
\displaystyle \sum_{v} C_1(v)\Bigl[\,2\alpha\,C_0(v) + \beta_v\,\Bigr]
}{
2\alpha\,\displaystyle \sum_{v} C_1(v)^2
},\\
& \alpha=1-\tfrac1n,\quad
\beta_v=-2f_i(v)+\frac{1-2q}{n\Delta}.\\
\end{split}
\end{equation}

This completes the proof.

\end{proof}

\subsection{The Proof of Theorem~\ref{thm:Corr-RR_proof_d}}\label{appendix:thm:Corr-RR_proof_d}

\begin{proof}
We establish the per-user guarantee for each phase and then apply parallel composition across disjoint user sets.

\paragraph{Phase~I users.}
Fix a user $i\in A$. The mechanism perturbs each coordinate
$x_{i,j}$ of $\mathbf{x}_i=(x_{i,1},\dots,x_{i,d})$ independently using GRR with budget $\epsilon/d$.
By \emph{Sequential Composition} (Theorem~\ref{compositiontheorem}), the per-user channel
$\mathcal{M}_A$ is $\epsilon$-LDP.

\paragraph{Phase~II users.}
Fix a user $i\in B$. The mechanism samples a pivot index $j\in[d]$ uniformly at random,
independently of $\mathbf{x}_i$. Conditioned on this fixed choice of $j$, the release is obtained by
\[
\mathbf{x}_i \xrightarrow{\ \mathcal{M}_{\epsilon}\ \text{on pivot } j\ } y_{i,j}
\xrightarrow{\ \text{post-process } g(\cdot;\,p_y)\ } \mathbf{y}_i=(y_{i,1},\dots,y_{i,d}),
\]
i.e., $x_{i,j}$ is privatized with budget $\epsilon$ to produce $y_{i,j}\sim\mathcal{M}_{\epsilon}(x_{i,j})$,
and the non-pivot coordinates are generated as
\[
y_{i,k}=g_k\!\bigl(y_{i,j};\,p_y^{(j,k)}\bigr),\qquad k\neq j,
\]
where $g_k$ is a randomized mapping that depends only on the already privatized pivot $y_{i,j}$ and the
\emph{public} reuse probability $p_y^{(j,k)}\in[0,1]$ derived from Phase~I statistics (hence independent of $\mathbf{x}_i$).
By the \emph{Post-Processing Theorem} (Theorem~\ref{Post-processing}), for any fixed pivot $j$ this conditional channel
$\mathcal{M}_B^{(j)}$ is $\epsilon$-LDP. Since the pivot index is chosen independently of the user’s data, the overall Phase~II mechanism is a uniform random
mixture over the $d$ conditional channels. A mixture of $\epsilon$-LDP mechanisms with data-independent mixing
preserves the same privacy guarantee, so $\mathcal{M}_B$ is also $\epsilon$-LDP.

\paragraph{Overall.}
Phases~I and II act on disjoint user subsets $A$ and $B$. By \emph{Parallel Composition}
(Theorem~\ref{Paralleltheorem}), the combined mechanism $\mathcal{M}=(\mathcal{M}_A,\mathcal{M}_B)$ satisfies $\epsilon$-LDP.

\end{proof}

\subsection{Additional Experimental Results}\label{Appendix2}

\subsubsection{Real-World Dataset Characteristics} Figure \ref{fig:relative_freq_realdata} presents the attribute-wise relative-frequency distributions for the Clave, Mushroom, and Adult datasets. The Clave attributes are nearly balanced, with both $X_1$ and $X_2$ showing minimal skewness, indicating that users are distributed fairly uniformly across the two binary categories. In contrast, the Mushroom dataset exhibits pronounced skewness: $X_1$ is highly right-skewed, with most records concentrated in the first few categories, while $X_2$ is slightly left-skewed. This imbalance explains why mechanisms that rely on prior distributions, such as RS+RFD, perform strongly on Mushroom. The Adult dataset lies between these two extremes—its first two attributes ($X_1$, $X_2$) are nearly uniform with negligible skewness, whereas $X_3$ shows moderate right skew.

\begin{figure*}[h!]

  \centering

  \begin{minipage}{.32\textwidth}
    
      \includegraphics[width=\linewidth]{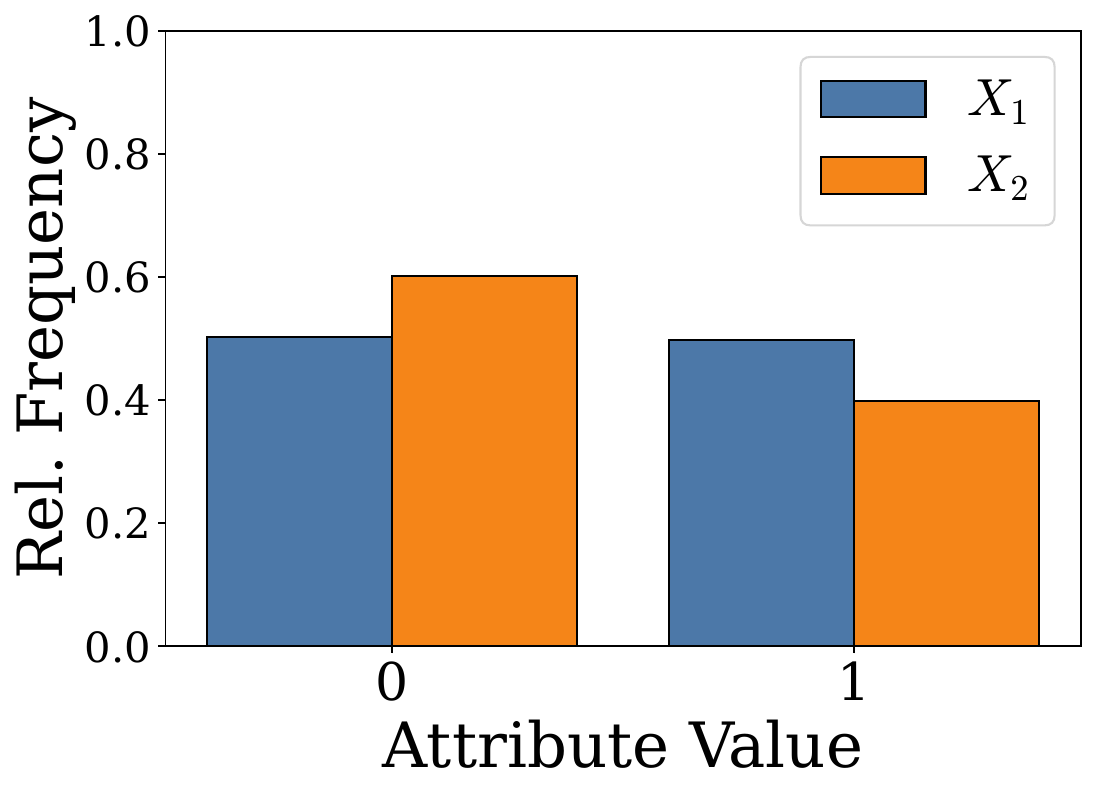}
    \subcaption{Clave}
    \label{fig:clave_freq_dis}
  \end{minipage}\hfill
  \begin{minipage}{.32\textwidth}
          \includegraphics[width=\linewidth]{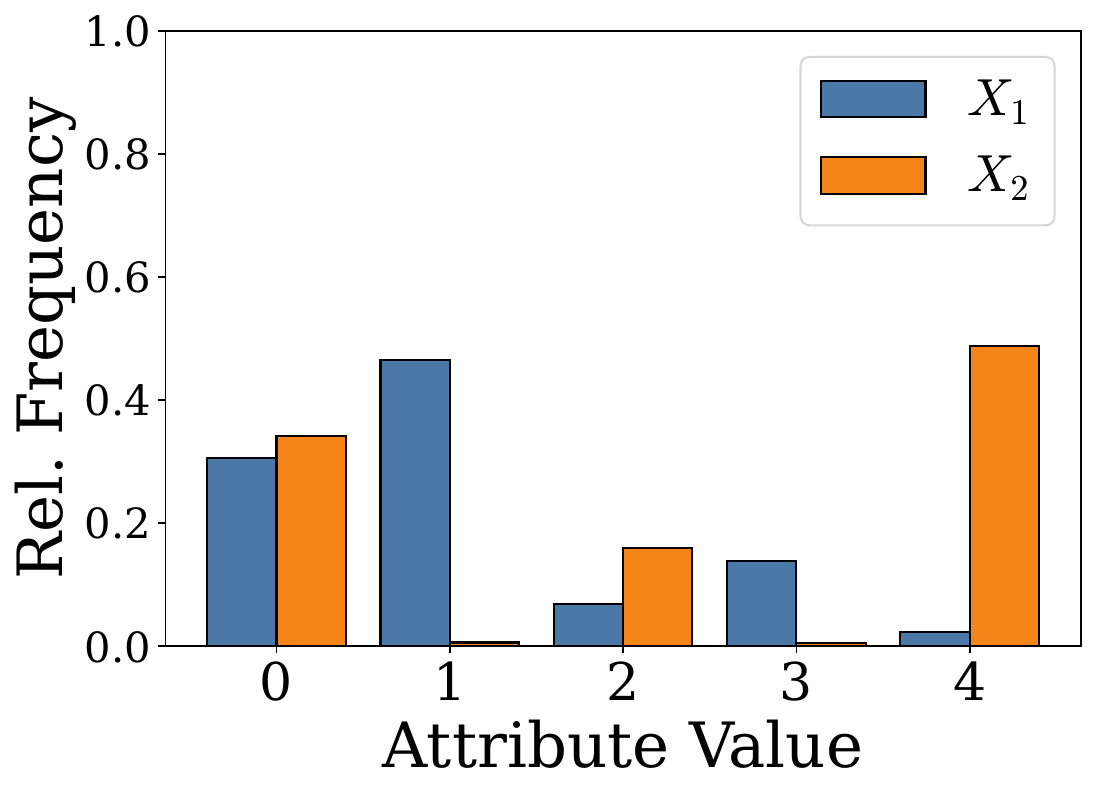}
   
    \subcaption{Mushroom}
    \label{fig:mush_freq_dis}
  \end{minipage}\hfill
  \begin{minipage}{.32\textwidth}
\includegraphics[width=\linewidth]{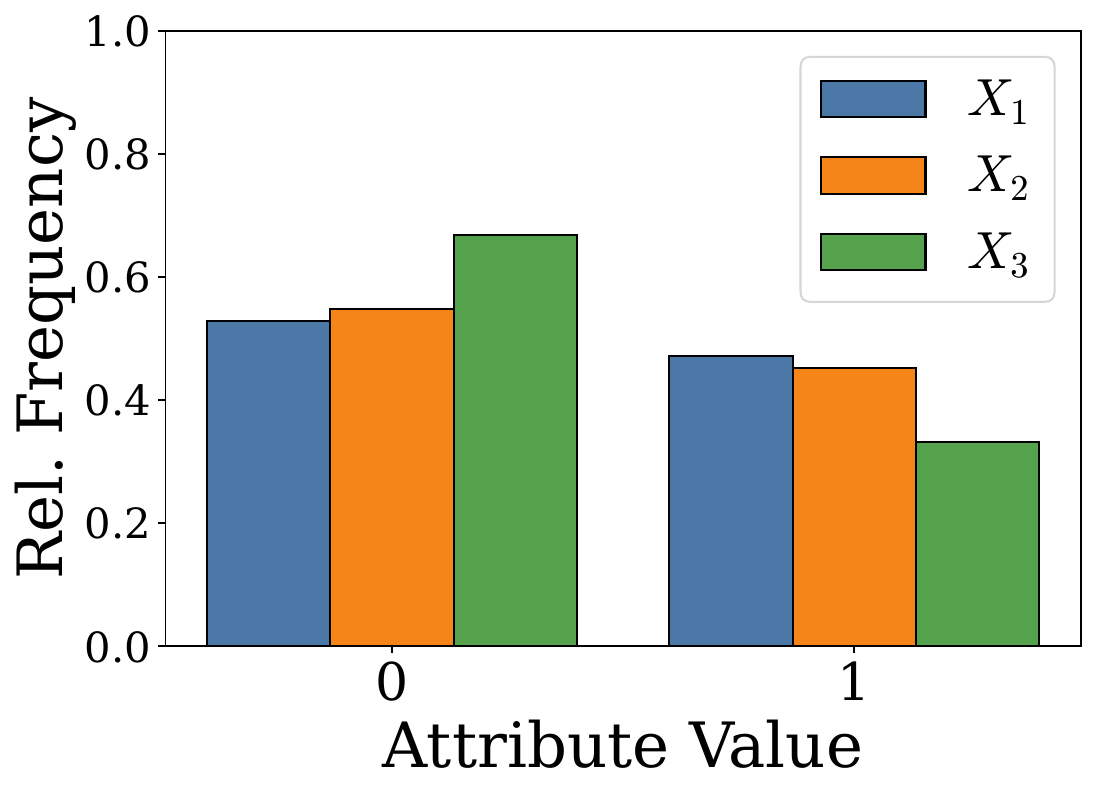}    
    \subcaption{Adult}
    \label{fig:Adult_freq_dis}
  \end{minipage}\hfill

  \caption{Attribute-wise relative-frequency distributions for (a) Clave, (b) Mushroom, and (c) Adult datasets.
Mushroom exhibits highly skewed marginals dominated by one or two values, whereas Clave and Adult are more balanced.}
\Description{Supplementary histogram plots illustrating data-skewness for Clave, Mushroom, and Adult dataset.}

  \label{fig:relative_freq_realdata}
\end{figure*}

\subsubsection{Impact of Number of Attributes under weak correlation~\label{appendix:result:attNo}.} Figure \ref{fig:mse_vs_attr_bin_cor01_syna} illustrate MSE as the number of attributes increases on the SynA dataset under weak correlation ($\rho=0.1$). Each subplot corresponds to a different privacy budget $\epsilon \in \{0.1,0.3,0.5\}$. As expected, error grows with dimensionality for all mechanisms, and Corr-RR maintains the lowest MSE across all settings even though dependencies are minimal. Figure \ref{fig:mse_vs_attr_bin_cor01_synb} presents the same analysis for SynB. Similar trends are observed. 

\begin{figure*}[h!]
  \centering

  \begin{minipage}{.32\textwidth}
  \includegraphics[width=\linewidth]{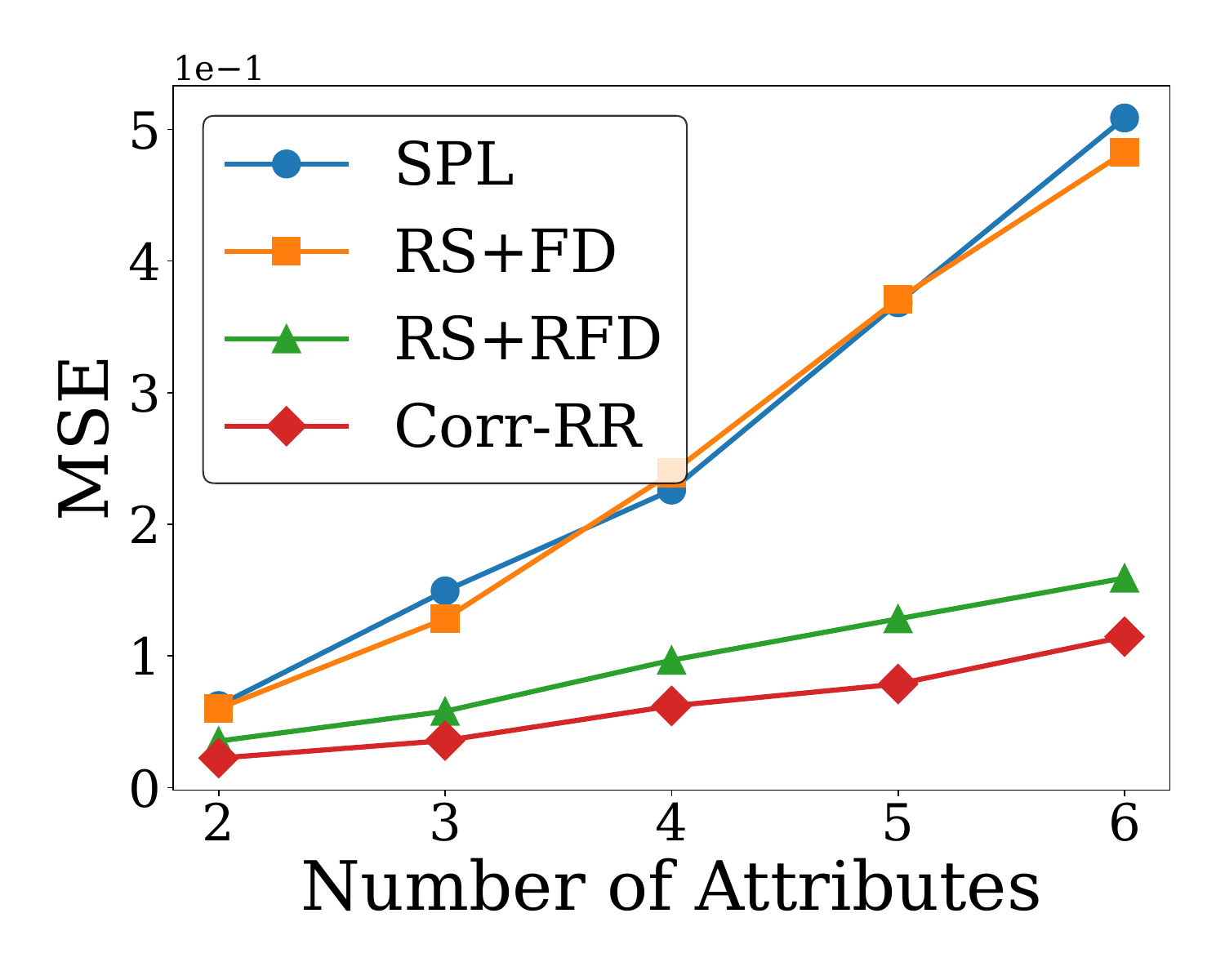}    
    \subcaption{$\epsilon = 0.1$}
    \label{fig:mse_attr_bin_cor01_eps1}
  \end{minipage}\hfill
  \begin{minipage}{.32\textwidth}
 \includegraphics[width=\linewidth]{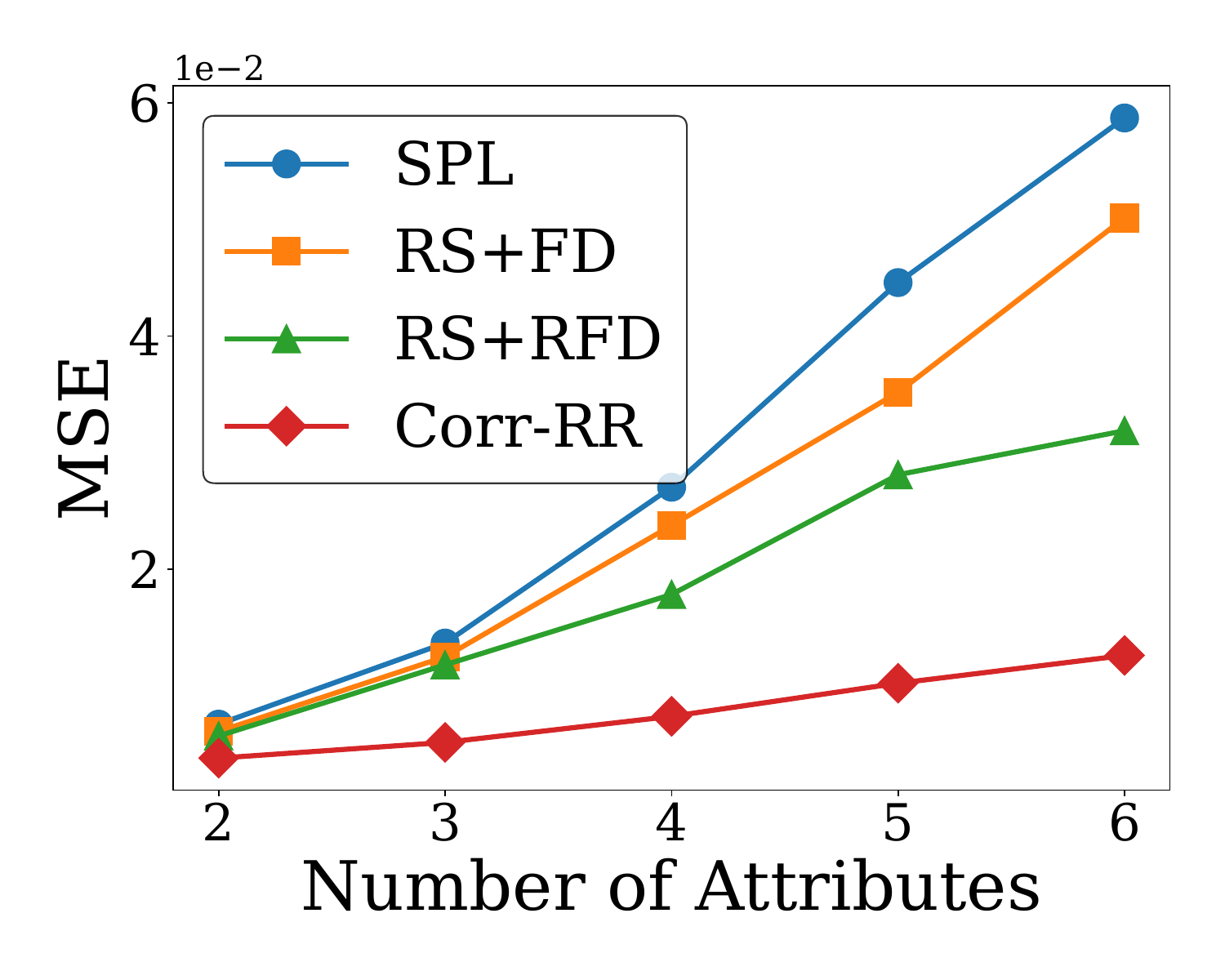}   
    \subcaption{$\epsilon = 0.3$}
    \label{fig:mse_attr_bin_cor01_eps3}
  \end{minipage}\hfill
  \begin{minipage}{.32\textwidth}
  \includegraphics[width=\linewidth]{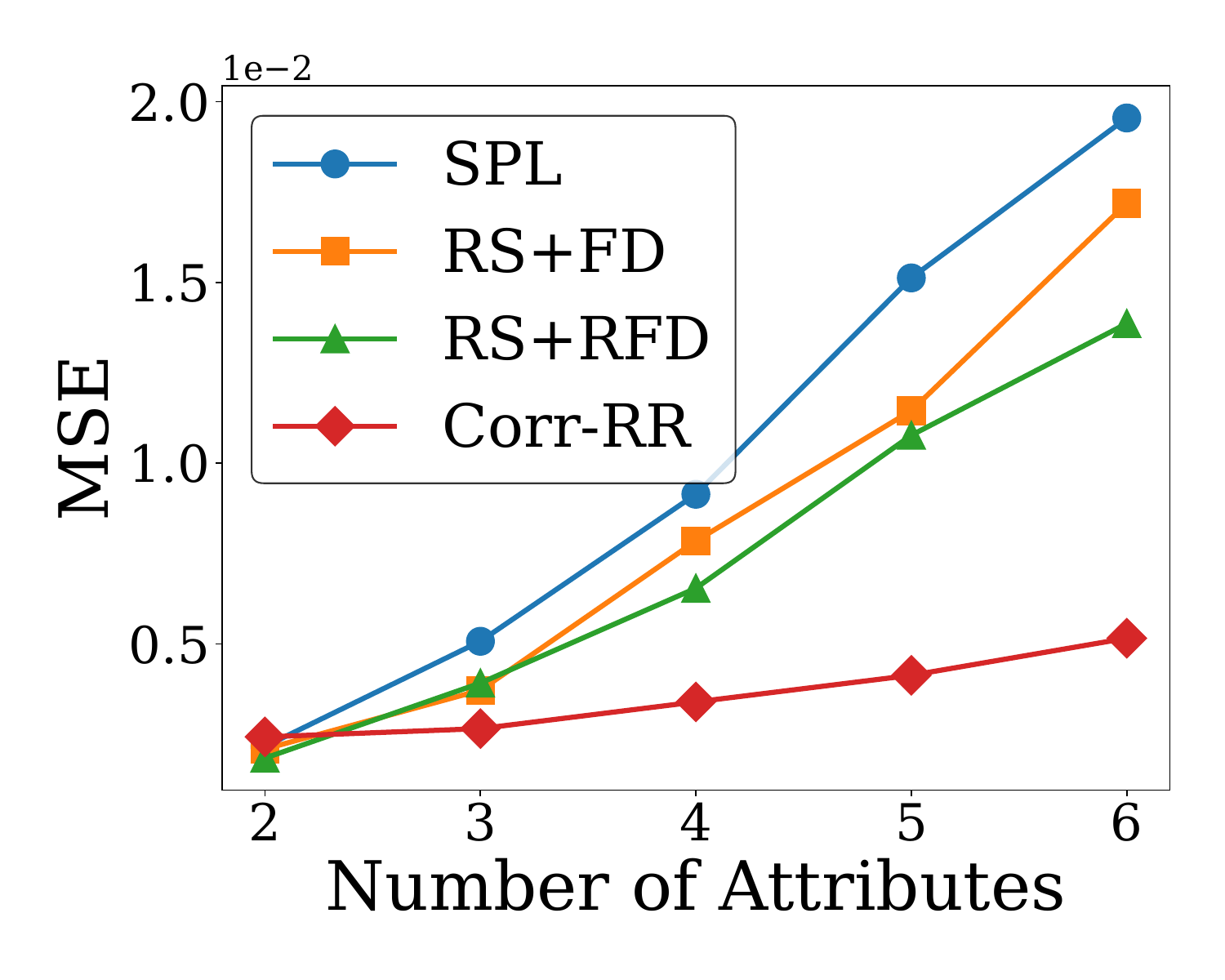}      
    \subcaption{$\epsilon = 0.5$}
    \label{fig:mse_attr_bin_cor01_eps5}
  \end{minipage}\hfill

 \caption{
MSE vs.\ number of attributes on SynA with $\rho=0.1$ and $|\mathcal{D}|=4$. 
Subplots correspond to different privacy budgets.}
\Description{Three line plots showing mean squared error as a function of the number of attributes for different privacy budgets, comparing SPL, RS+FD, RS+RFD, and Corr-RR.}

  \label{fig:mse_vs_attr_bin_cor01_syna}
\end{figure*}

\begin{figure*}[h!]

  \centering

  \begin{minipage}{.32\textwidth}
  \includegraphics[width=\linewidth]{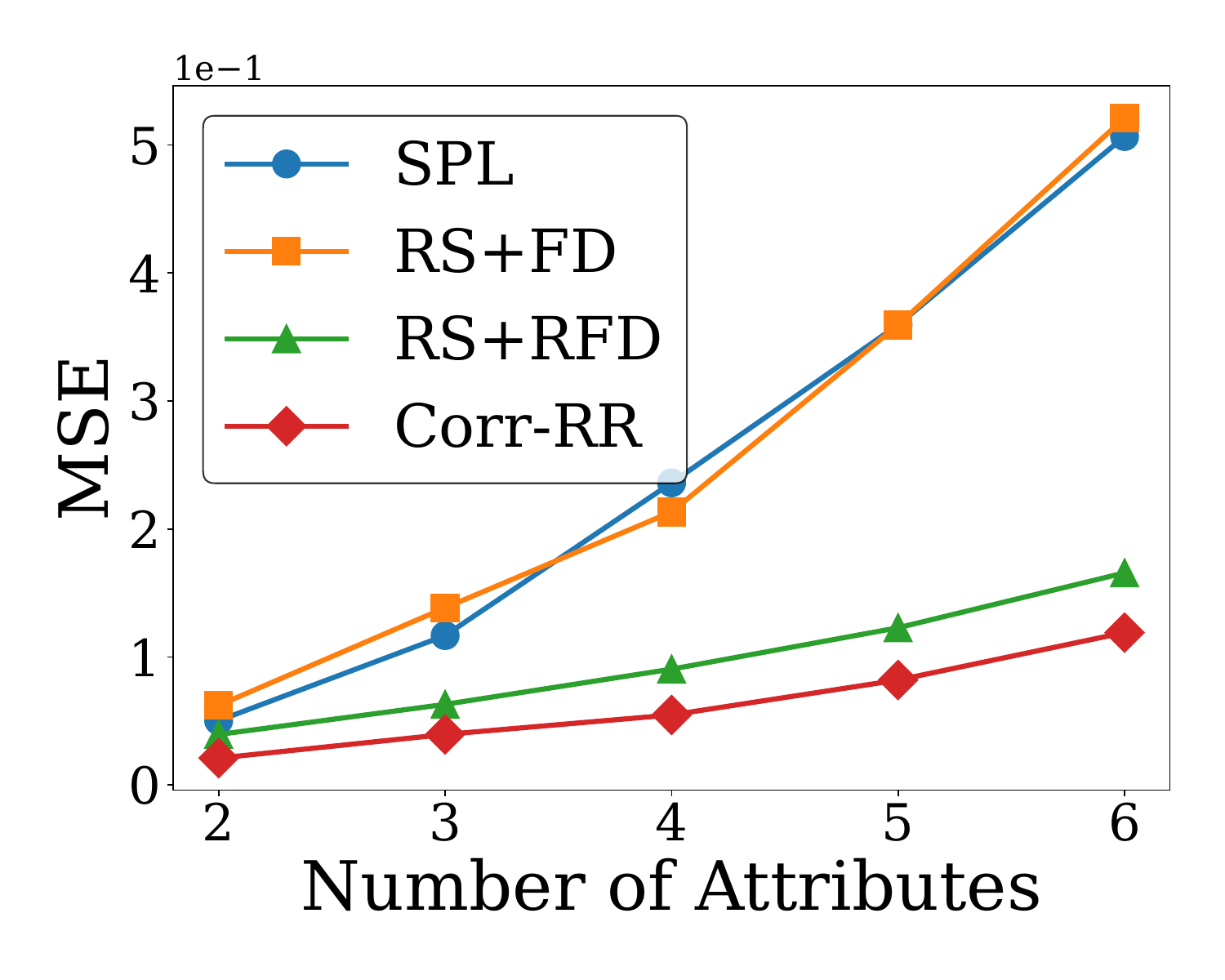}    
    \subcaption{$\epsilon = 0.1$}
    \label{fig:mse_attr_bin_cor01_eps1_synb}
  \end{minipage}\hfill
  \begin{minipage}{.32\textwidth}
  \includegraphics[width=\linewidth]{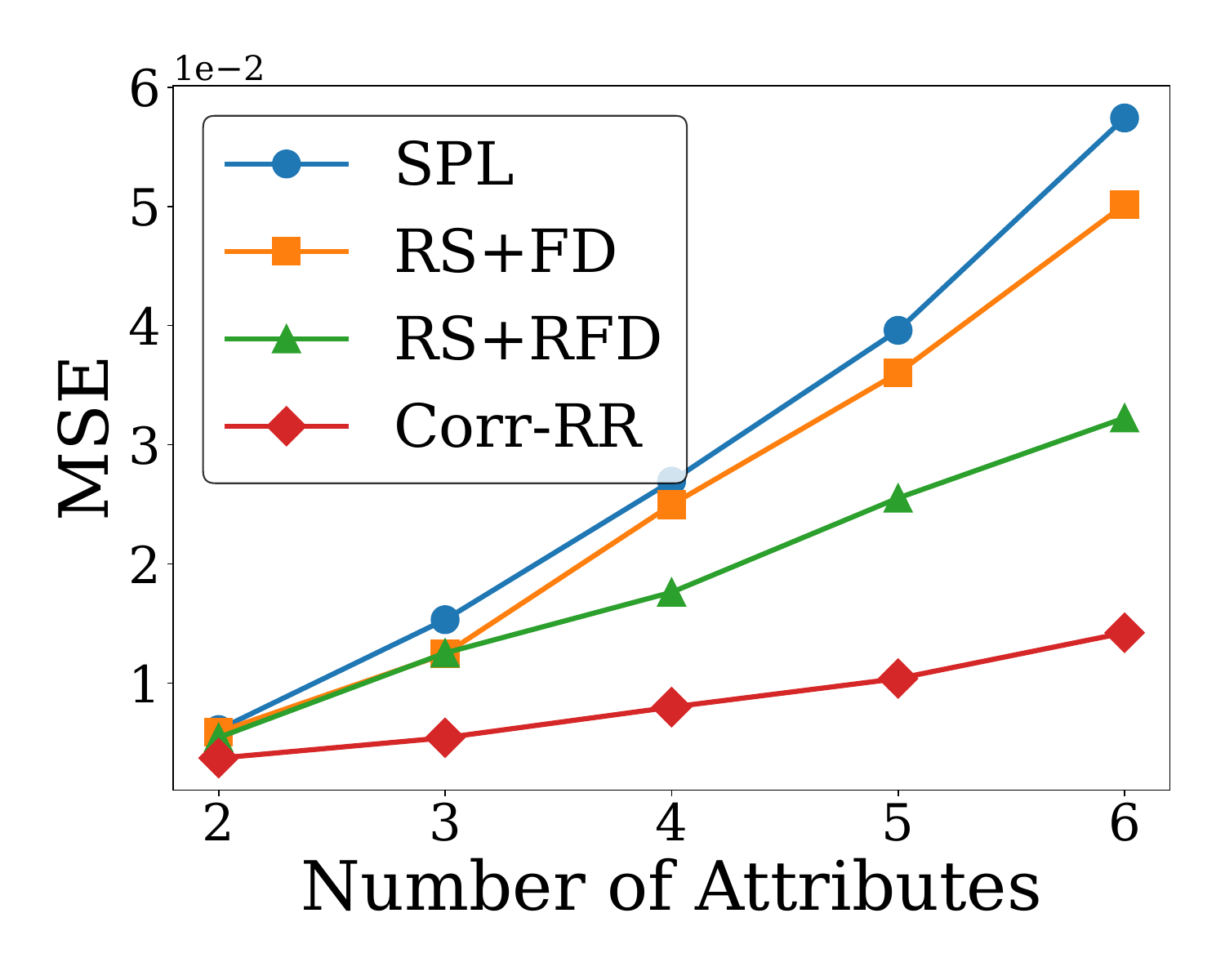}   
    \subcaption{$\epsilon = 0.3$}
    \label{fig:mse_attr_bin_cor01_eps3_synb}
  \end{minipage}\hfill
  \begin{minipage}{.32\textwidth}
\includegraphics[width=\linewidth]{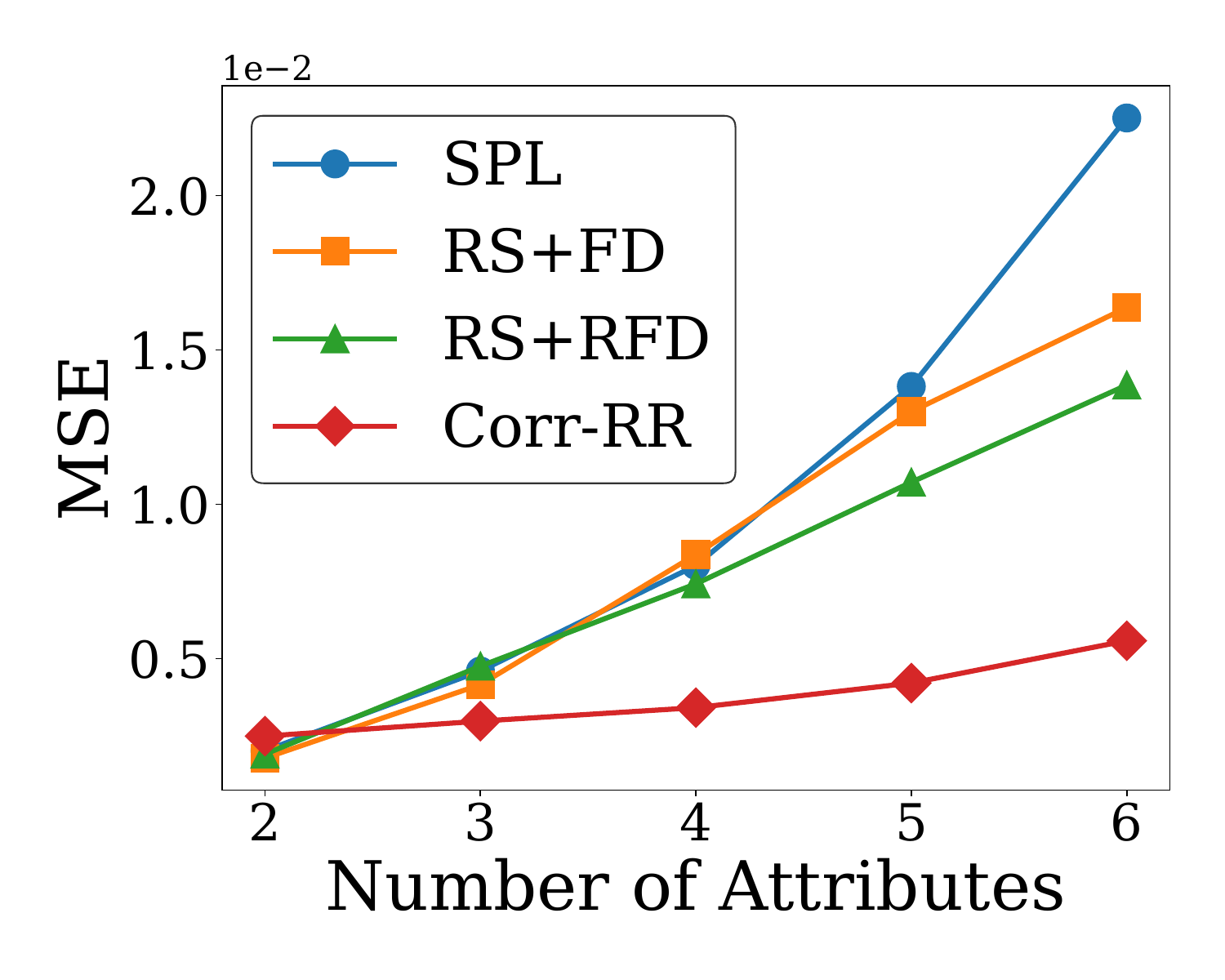}    
    \subcaption{$\epsilon = 0.5$}
    \label{fig:mse_attr_bin_cor01_eps5_synb}
  \end{minipage}\hfill

 \caption{
MSE vs.\ number of attributes on SynB with $\rho=0.1$ and $|\mathcal{D}|=4$. 
Subplots correspond to different privacy budgets.}
\Description{Three line plots showing mean squared error as a function of the number of attributes for different privacy budgets, comparing SPL, RS+FD, RS+RFD, and Corr-RR.}
  \label{fig:mse_vs_attr_bin_cor01_synb}
\end{figure*}

\subsubsection{Impact of the Size of Phase I Users} Tables \ref{tab:phase1_n_200_star}–\ref{tab:phase1_n_20000_prog} report the MSE of RS+RFD and Corr-RR across varying Phase I sizes for both SynA and SynB under low correlation ($\rho=0.1$) and different dataset scales ($n=200$, $2{,}000$, and $20{,}000$). At very small scales ($n=200$), RS+RFD typically achieves the lowest error at tighter privacy budgets due to the instability of Phase I marginal estimates, while Corr-RR becomes competitive only at larger $\epsilon$. As the user population increases to $n=2{,}000$, the trends become smoother, and Corr-RR begins to outperform RS+RFD at moderate and high budgets, though RS+RFD may still hold a slight advantage at $\epsilon=0.1$. At large scale ($n=20{,}000$), Corr-RR consistently attains the lowest MSE across nearly all Phase I sizes and privacy budgets for both SynA and SynB, reflecting the improved reliability of correlation estimation when ample Phase I samples are available. Across all settings, small Phase I fractions often yield the best results, while larger Phase I splits can degrade performance by reducing the number of users contributing full-budget reports in Phase II.

\begin{table*}[t]
\centering
\caption{MSE of RS+RFD and Corr-RR across varying size of Phase~I users on SynA ($n=200$, $\rho=0.1$, $|\mathcal{D}|=4$, $d=2$). The lowest MSE value between both methods across each $\epsilon$ is shown in bold.}
\centering
\small
\begin{tabular}{llcccccccccc}
\toprule
 & & \multicolumn{10}{c}{\textbf{Size of Phase I Users ($n_1$)}} \\
\cmidrule(lr){3-12}
\textbf{Method} & \textbf{$\epsilon$}
& $n_1$=10 & $n_1$=20 & $n_1$=30 & $n_1$=40 & $n_1$=50 
& $n_1$=60 & $n_1$=70 & $n_1$=80 & $n_1$=90 & $n_1$=100 \\
\midrule

\multirow{3}{*}{RS+RFD}

& 0.1
& \textbf{3.853e-01} & 6.937e-01 & 9.096e-01 & 1.323e+00
& 1.643e+00 & 1.648e+00 & 1.989e+00 & 2.522e+00 & 2.523e+00 & 3.012e+00 \\

& 0.3
& \textbf{9.978e-02} & 1.246e-01 & 1.602e-01
& 1.690e-01 & 2.133e-01 & 2.256e-01 & 2.658e-01 & 2.926e-01
& 3.553e-01 & 3.748e-01 \\

& 0.5
& 6.694e-02 & 7.327e-02 & 8.222e-02
& 8.464e-02 & 9.712e-02 & 9.711e-02 & 1.115e-01 & 1.121e-01 
& 1.286e-01 & 1.215e-01 \\
\midrule

\multirow{3}{*}{Corr-RR}

& 0.1
& 1.745e+00 & 1.657e+00 & 1.981e+00
& 2.211e+00 & 2.725e+00 & 2.706e+00 & 2.738e+00 & 3.094e+00 
& 3.492e+00 & 3.356e+00 \\

& 0.3
& 1.701e-01 & 2.250e-01 & 1.987e-01 
& 2.254e-01 & 2.602e-01 & 2.382e-01 & 3.433e-01 & 3.789e-01 
& 3.532e-01 & 4.209e-01 \\

& 0.5
& \textbf{6.023e-02} & 6.807e-02 & 7.651e-02 
& 7.149e-02 & 8.773e-02 & 9.730e-02 & 1.056e-01 & 1.167e-01 
& 1.140e-01 & 1.342e-01 \\
\bottomrule
\end{tabular}
\label{tab:phase1_n_200_star}
\end{table*}

\begin{table*}[t]

\centering

\caption{MSE of RS+RFD and Corr-RR across varying size of Phase~I users on SynA ($n=2{,}000$, $\rho=0.1$, $|D|=4$, $d=2$). The lowest MSE value between both methods across each $\epsilon$ is shown in bold.}
\centering
\small
\begin{tabular}{llcccccccccc}
\toprule
 & & \multicolumn{10}{c}{\textbf{Size of Phase I Users ($n_1$)}} \\
\cmidrule(lr){3-12}
\textbf{Method} & \textbf{$\epsilon$}
& $n_1$=100 & $n_1$=200 & $n_1$=300 & $n_1$=400 & $n_1$=500 
& $n_1$=600 & $n_1$=700 & $n_1$=800 & $n_1$=900 & $n_1$=1000  \\
\midrule

\multirow{3}{*}{RS+RFD}

& 0.1
& \textbf{9.175e-02} & 1.165e-01 & 1.387e-01
& 1.577e-01 & 1.998e-01 & 2.136e-01 & 2.553e-01 & 2.790e-01
& 2.994e-01 & 3.204e-01 \\

& 0.3
& 3.213e-02 & 3.875e-02 & 3.535e-02
& 3.792e-02 & 3.605e-02 & 4.130e-02 & 3.972e-02 & 4.417e-02
& 4.451e-02 & 4.638e-02 \\

& 0.5
& 1.697e-02 & 1.625e-02 & 1.585e-02
& 1.594e-02 & 1.811e-02 & 1.773e-02 & 1.849e-02 & 1.523e-02
& 1.856e-02 & 1.943e-02 \\
\midrule

\multirow{3}{*}{Corr-RR}

& 0.1
& 1.834e-01 & 1.788e-01 & 2.052e-01
& 2.302e-01 & 2.697e-01 & 3.111e-01 & 2.965e-01 & 3.446e-01
& 3.491e-01 & 3.618e-01 \\

& 0.3
& \textbf{1.885e-02} & 2.004e-02 & 2.162e-02
& 2.524e-02 & 2.487e-02 & 3.026e-02 & 3.214e-02 & 3.398e-02
& 3.918e-02 & 3.734e-02 \\

& 0.5
& \textbf{5.800e-03} & 6.353e-03 & 7.282e-03
& 8.081e-03 & 9.914e-03 & 1.013e-02 & 1.040e-02 & 1.143e-02
& 1.188e-02 & 1.356e-02 \\
\bottomrule
\end{tabular}
\label{tab:phase1_n_2000_star}

\end{table*}

\begin{table*}[t]

\centering

\caption{MSE of RS+RFD and Corr-RR across varying size of Phase~I users on SynA ($n=20{,}000$, $\rho=0.1$, $|D|=4$, $d=2$). The lowest MSE value between both methods across each $\epsilon$ is shown in bold.}
\centering
\small
\begin{tabular}{llcccccccccc}
\toprule
 & & \multicolumn{10}{c}{\textbf{Size of Phase I Users ($n_1$)}} \\
\cmidrule(lr){3-12}
\textbf{Method} & \textbf{$\epsilon$}
& $n_1$=1000 & $n_1$=2000 & $n_1$=3000 & $n_1$=4000 & $n_1$=5000 
& $n_1$=6000 & $n_1$=7000 & $n_1$=8000 & $n_1$=9000 & $n_1$=10000  \\
\midrule

\multirow{3}{*}{RS+RFD}

& 0.1
& 3.683e-02 & 3.358e-02 & 3.667e-02
& 3.522e-02 & 3.618e-02 & 3.801e-02 & 4.348e-02 & 4.244e-02
& 4.282e-02 & 4.343e-02 \\

& 0.3
& 5.586e-03 & 5.309e-03 & 5.508e-03
& 5.634e-03 & 6.200e-03 & 5.301e-03 & 6.595e-03 & 5.789e-03
& 6.297e-03 & 6.199e-03 \\

& 0.5
& 1.849e-03 & 1.762e-03 & 1.950e-03
& 1.851e-03 & 2.147e-03 & 2.245e-03 & 1.910e-03 & 1.993e-03
& 1.965e-03 & 2.154e-03 \\
\midrule

\multirow{3}{*}{Corr-RR}

& 0.1
& \textbf{1.464e-02} & 1.930e-02 & 2.082e-02
& 2.390e-02 & 2.388e-02 & 2.906e-02 & 2.911e-02 & 3.098e-02
& 3.564e-02 & 3.586e-02 \\

& 0.3
& \textbf{1.901e-03} & 2.020e-03 & 2.171e-03
& 2.366e-03 & 2.434e-03 & 2.888e-03 & 2.952e-03 & 3.559e-03
& 3.574e-03 & 3.657e-03 \\

& 0.5
& \textbf{5.729e-04} & 5.964e-04 & 7.411e-04
& 8.165e-04 & 8.379e-04 & 1.049e-03 & 1.032e-03 & 1.132e-03
& 1.142e-03 & 1.343e-03 \\
\bottomrule
\end{tabular}
\label{tab:phase1_n_20000_star}

\end{table*}

\begin{table*}[t]

\centering

\caption{MSE of RS+RFD and Corr-RR across varying size of Phase~I users on SynB ($n=200$, $\rho=0.1$, $|D|=4$, $d=2$). The lowest MSE value between both methods across each $\epsilon$ is shown in bold.}
\centering
\small
\begin{tabular}{llcccccccccc}
\toprule
 & & \multicolumn{10}{c}{\textbf{Size of Phase I Users ($n_1$)}} \\
\cmidrule(lr){3-12}
\textbf{Method} & \textbf{$\epsilon$}
& $n_1$=10 & $n_1$=20 & $n_1$=30 & $n_1$=40 & $n_1$=50 
& $n_1$=60 & $n_1$=70 & $n_1$=80 & $n_1$=90 & $n_1$=100 \\
\midrule

\multirow{3}{*}{RS+RFD}

& 0.1
& \textbf{4.448e-01} & 6.688e-01 & 9.081e-01 & 1.197e+00 & 1.640e+00
& 1.988e+00 & 2.144e+00 & 2.268e+00 & 2.691e+00 & 2.823e+00 \\

& 0.3
& \textbf{1.029e-01} & 1.253e-01 & 1.647e-01 & 1.838e-01 & 2.051e-01
& 2.249e-01 & 2.726e-01 & 2.818e-01 & 3.034e-01 & 3.318e-01 \\

& 0.5
& \textbf{5.553e-02} & 7.263e-02 & 8.177e-02 & 8.660e-02 & 8.945e-02
& 8.956e-02 & 1.136e-01 & 1.171e-01 & 1.140e-01 & 1.178e-01 \\
\midrule

\multirow{3}{*}{Corr-RR}

& 0.1
& 1.618e+00 & 1.970e+00 & 2.006e+00 & 2.471e+00 & 2.597e+00
& 2.590e+00 & 3.012e+00 & 3.313e+00 & 3.440e+00 & 3.727e+00 \\

& 0.3
& 1.829e-01 & 1.960e-01 & 2.061e-01 & 2.497e-01 & 2.536e-01
& 2.759e-01 & 3.038e-01 & 3.314e-01 & 3.950e-01 & 3.871e-01 \\

& 0.5
& 5.665e-02 & 6.606e-02 & 7.539e-02 & 7.693e-02 & 9.076e-02
& 9.548e-02 & 9.971e-02 & 1.070e-01 & 1.143e-01 & 1.150e-01 \\
\bottomrule
\end{tabular}
\label{tab:phase1_n_200_prog}

\end{table*}

\begin{table*}[t]

\centering

\caption{MSE of RS+RFD and Corr-RR across varying size of Phase~I users on SynB ($n=2{,}000$, $\rho=0.1$, $|D|=4$, $d=2$). The lowest value across both methods for each $\epsilon$ is shown in \textbf{bold}.}
\centering
\small
\begin{tabular}{llcccccccccc}
\toprule
 & & \multicolumn{10}{c}{\textbf{Size of Phase I Users ($n_1$)}} \\
\cmidrule(lr){3-12}
\textbf{Method} & \textbf{$\epsilon$}
& $n_1$=100 & $n_1$=200 & $n_1$=300 & $n_1$=400 & $n_1$=500 
& $n_1$=600 & $n_1$=700 & $n_1$=800 & $n_1$=900 & $n_1$=1000 \\
\midrule

\multirow{3}{*}{RS+RFD}

& 0.1
& \textbf{9.785e-02} & 1.176e-01 & 1.424e-01 & 1.743e-01 & 2.088e-01
& 2.137e-01 & 2.222e-01 & 2.816e-01 & 2.946e-01 & 3.436e-01 \\

& 0.3
& 3.645e-02 & 3.734e-02 & 3.576e-02 & 3.802e-02 & 3.869e-02
& 3.907e-02 & 3.647e-02 & 4.396e-02 & 4.173e-02 & 4.618e-02 \\

& 0.5
& 1.665e-02 & 1.636e-02 & 1.662e-02 & 1.690e-02 & 1.688e-02
& 1.716e-02 & 1.742e-02 & 1.667e-02 & 1.767e-02 & 1.909e-02 \\
\midrule

\multirow{3}{*}{Corr-RR}

& 0.1
& 1.621e-01 & 1.795e-01 & 2.087e-01 & 2.273e-01 & 2.280e-01
& 2.746e-01 & 2.817e-01 & 3.124e-01 & 3.913e-01 & 3.610e-01 \\

& 0.3
& \textbf{1.542e-02} & 2.062e-02 & 2.154e-02 & 2.491e-02 & 2.608e-02
& 2.857e-02 & 3.116e-02 & 3.129e-02 & 3.710e-02 & 4.040e-02 \\

& 0.5
& \textbf{5.545e-03} & 6.197e-03 & 6.928e-03 & 7.553e-03 & 9.000e-03
& 9.567e-03 & 1.069e-02 & 1.183e-02 & 1.168e-02 & 1.228e-02 \\
\bottomrule
\end{tabular}
\label{tab:phase1_n_2000_prog}

\end{table*}

\begin{table*}[t]

\centering

\caption{MSE of RS+RFD and Corr-RR across varying size of Phase~I users on SynB ($n=20{,}000$, $\rho=0.1$, $|D|=4$, $d=2$). The lowest MSE value between both methods across each $\epsilon$ is shown in bold.}
\centering
\small
\begin{tabular}{llcccccccccc}
\toprule
 & & \multicolumn{10}{c}{\textbf{Size of Phase I Users ($n_1$)}} \\
\cmidrule(lr){3-12}
\textbf{Method} & \textbf{$\epsilon$}
& $n_1$=1000 & $n_1$=2000 & $n_1$=3000 & $n_1$=4000 & $n_1$=5000 
& $n_1$=6000 & $n_1$=7000 & $n_1$=8000 & $n_1$=9000 & $n_1$=10000  \\
\midrule

\multirow{3}{*}{RS+RFD}

& 0.1
& 3.325e-02 & 3.585e-02 & 4.054e-02
& 3.675e-02 & 3.754e-02 & 4.432e-02 & 3.827e-02 & 3.823e-02
& 4.743e-02 & 4.811e-02 \\

& 0.3
& 5.226e-03 & 5.441e-03 & 6.169e-03
& 5.425e-03 & 5.819e-03 & 5.499e-03 & 5.957e-03 & 5.604e-03
& 6.152e-03 & 6.315e-03 \\

& 0.5
& 1.867e-03 & 1.704e-03 & 1.997e-03
& 2.036e-03 & 1.801e-03 & 1.898e-03 & 1.898e-03 & 2.168e-03
& 1.783e-03 & 2.108e-03 \\
\midrule

\multirow{3}{*}{Corr-RR}

& 0.1
& \textbf{1.530e-02} & 2.021e-02 & 2.071e-02
& 2.154e-02 & 2.452e-02 & 2.891e-02
& 3.023e-02 & 3.100e-02 & 3.448e-02 & 3.706e-02 \\

& 0.3
& \textbf{2.060e-03} & 1.885e-03 & 2.204e-03
& 2.381e-03 & 2.489e-03 & 2.659e-03
& 3.139e-03 & 3.292e-03 & 3.451e-03 & 3.690e-03 \\

& 0.5
& \textbf{5.888e-04} & 6.173e-04 & 7.240e-04
& 8.354e-04 & 8.474e-04 & 1.010e-03
& 1.075e-03 & 1.057e-03 & 1.131e-03 & 1.178e-03 \\
\bottomrule
\end{tabular}
\label{tab:phase1_n_20000_prog}

\end{table*}

\end{document}